\documentclass[pdftex,pagesize=auto,paper=letter,11pt,abstract=on,bibliography=oldstyle]{scrartcl}

\usepackage{etex}
\usepackage[utf8]{inputenc}
\usepackage[english]{babel}
\usepackage[pdftex]{xcolor}
\usepackage{amsthm}
\usepackage{amsmath}
\usepackage{amssymb}
\usepackage{xspace}
\usepackage{rotating}
\usepackage{multirow}
\usepackage{booktabs}
\usepackage{url}
\usepackage[protrusion=true,expansion=true]{microtype}
\usepackage{lmodern}
\usepackage[T1]{fontenc}
\usepackage{tikz}
\usetikzlibrary{arrows,decorations.pathreplacing}
\usepackage{contour}
\contourlength{1.5pt}
\contournumber{40}
\usepackage{mathtools}
\usepackage{paralist}
\usepackage{accents}
\usepackage{authblk}
\usepackage[margin=2.5cm]{geometry}

\usepackage[numbers,sort]{natbib}

\usepackage[font=small]{caption}
\usepackage{subcaption}
\usepackage[ruled, vlined, linesnumbered]{algorithm2e}
\DontPrintSemicolon

\SetCommentSty{commentfont}
\newcommand{\assign}{\longleftarrow}

\SetKwFunction{queueIsEmpty}{isEmpty}%
\SetKwFunction{queueIsNotEmpty}{isNotEmpty}%
\SetKwFunction{queueMinElement}{minElement}%
\SetKwFunction{queueDeleteMin}{deleteMin}%
\SetKwFunction{queueDeleteMax}{deleteMax}%
\SetKwFunction{queueInsert}{insert}%
\SetKwFunction{queueUpdate}{update}%
\SetKwFunction{setInsert}{insert}%
\SetKwFunction{setDeleteDominatedLabels}{deleteDominatedLabels}%
\SetKw{KwNot}{not}%
\SetKw{KwDominates}{dominates}%

\usepackage{pgfplots}
\pgfplotsset{compat=1.9}
\definecolor{col1}         {rgb}{0.27,0.39,0.67}
\definecolor{col2}         {rgb}{0.0, 0.59,0.51}
\definecolor{col3}         {rgb}{0.64,0.13,0.14}
\definecolor{col4}         {rgb}{0.87,0.61,0.11}
\definecolor{col5}         {rgb}{0.0, 0.59,0.51}
\definecolor{plotColor1}         {rgb}{0.27,0.39,0.67}
\definecolor{plotColor2}         {rgb}{0.0, 0.59,0.51}
\definecolor{plotColor3}         {rgb}{0.64,0.13,0.14}
\definecolor{plotColor4}         {rgb}{0.87,0.61,0.11}
\definecolor{plotColor5}         {rgb}{0.0, 0.59,0.51}

\newcommand{\eg}{e.\,g.\xspace}
\newcommand{\ie}{i.\,e.\xspace}
\newcommand{\wrt}{wrt.\xspace}
\newcommand{\cf}{cf.\xspace}

\newcommand{\case}[1]{\emph{#1:}}

\DeclareRobustCommand{\ubar}[1]{\ensuremath{\underaccent{\bar}{#1}}}
\newcommand{\obar}[1]{\ensuremath{\bar{#1}}}
\newcommand*{\minx}{\mathmakebox[\widthof{$\max$}][r]{\min}}
\newcommand{\mytilde}{\ensuremath{\operatorname{\hspace{2pt}\widetilde{\protect\textcolor{white}{\protect\raisebox{1.0pt}{c}}}}\hspace{-9pt}}}
\newcommand{\inverse}{{\hspace{0.5pt}\protect\raisebox{0.5pt}{\protect\scalebox{0.6}[1.0]{$\scriptstyle-$}\hspace{-0.6pt}$\scriptscriptstyle1$}\hspace{-1pt}}}

\newcommand{\graph}{\ensuremath{G}}
\newcommand{\vertices}{\ensuremath{V}}

\newcommand{\arcs}{\ensuremath{A}}
\newcommand{\vertex}{\ensuremath{v}}
\newcommand{\vertexa}{\ensuremath{u}}
\newcommand{\vertexb}{\ensuremath{v}}
\newcommand{\vertexc}{\ensuremath{w}}
\newcommand{\vertexd}{\ensuremath{x}}
\newcommand{\vertexe}{\ensuremath{y}}
\newcommand{\vertexf}{\ensuremath{z}}

\newcommand{\arc}{\ensuremath{a}}

\newcommand{\reals}{\ensuremath{\mathbb{R}}}
\newcommand{\posreals}{\ensuremath{\reals_{\ge 0}}}
\newcommand{\strictposreals}{\ensuremath{\reals_{>0}}}
\newcommand{\naturals}{\ensuremath{\mathbb{N}}}
\newcommand{\apath}{\ensuremath{P}}
\newcommand{\patha}{\ensuremath{P}}
\newcommand{\pathb}{\ensuremath{Q}}

\newcommand{\source}{\ensuremath{s}}
\newcommand{\target}{\ensuremath{t}}

\newcommand{\NP}{\ensuremath{\mathcal{NP}}}

\newcommand{\queue}{\ensuremath{Q}}
\newcommand{\drivingtimefunction}{\ensuremath{d}\xspace}
\newcommand{\consumptionfunction}{\ensuremath{c}\xspace}

\DeclareMathOperator{\profileCost}{cost}
\DeclareMathOperator{\profileMinBat}{in}
\DeclareMathOperator{\profileMaxBat}{out}

\DeclareMathOperator{\mergeop}{merge}
\DeclareMathOperator{\linkop}{link}
\DeclareMathOperator{\shiftop}{shift}
\DeclareMathOperator{\extendop}{extend}
\DeclareMathOperator{\distance}{dist}

\newcommand{\atime}{\ensuremath{\tau}}
\newcommand{\soc}{\ensuremath{\beta}}
\newcommand{\labelset}{\ensuremath{L}}
\newcommand{\alabel}{\ensuremath{\ell}}
\newcommand{\chargingstations}{\ensuremath{S}}
\newcommand{\chargingfunction}{\ensuremath{\operatorname{cf}}}
\newcommand{\simplechargingfunction}{\ensuremath{\mytilde{}\operatorname{cf}}}
\newcommand{\simplechargingfunctionInv}{\ensuremath{\mytilde{}\operatorname{cf}^{\inverse}}}

\newcommand{\maxbattery}{\ensuremath{M}}
\newcommand{\potential}{\ensuremath{\pi}}
\newcommand{\maxchargingspeed}{\chargingfunction_{\max}}
\newcommand{\omegaweightfunction}{\omega}
\newcommand{\omegapotential}{\potential_\omegaweightfunction}
\newcommand{\convexweightfunction}{f}
\newcommand{\convexpotential}{\potential_\convexweightfunction}
\newcommand{\convexfunctionbreakpoints}{\ensuremath{F}}

\newcommand{\heuristic}{\operatorname{H}}
\newcommand{\queuekey}{\ensuremath{k}}
\newcommand{\switchingsequence}{\ensuremath{T}}

\newcommand{\reduceddrivingtimefunction}{\ensuremath{\underline{\drivingtimefunction}}}
\newcommand{\convexlowerboundfunction}{\ensuremath{\ubar{f}}}
\newcommand{\functionbreakpoints}{\ensuremath{F}}
\newcommand{\convexfunctiona}{\ensuremath{{f_1}}}
\newcommand{\convexfunctionb}{\ensuremath{{f_2}}}
\newcommand{\convexfunctionc}{\ensuremath{{f}}}
\newcommand{\convexfunctionbreakpointsa}{\ensuremath{{F_1}}}
\newcommand{\convexfunctionbreakpointsb}{\ensuremath{{F_2}}}
\newcommand{\convexfunctionbreakpointsc}{\ensuremath{{F}}}
\newcommand{\functionasubscript}{\ensuremath{1}}
\newcommand{\functionbsubscript}{\ensuremath{2}}
\newcommand{\convexfunctionslope}{\ensuremath{\sigma}}
\newcommand{\derivative}{\ensuremath{\partial}}
\newcommand{\afunction}{\ensuremath{f}}
\newcommand{\algoheuconvex}{\ensuremath{\heuristic_\convexweightfunction}\xspace}
\newcommand{\algoheuomega}{\ensuremath{\heuristic_\omegaweightfunction}\xspace}
\newcommand{\algoheuomegaaggressive}{\ensuremath{\heuristic_\omegaweightfunction^\text{A}}\xspace}
\newcommand{\socprofile}{\ensuremath{b}}

\newcommand{\departure}{\ensuremath{\text{dep}}}
\newcommand{\arrival}{\ensuremath{\text{arr}}}
\newcommand{\departuresoc}{\ensuremath{\soc^\departure}}
\newcommand{\arrivalsoc}{\ensuremath{\soc^\arrival}}
\newcommand{\mincharge}{\ensuremath{\soc^{\text{min}}}}
\newcommand{\maxcharge}{\ensuremath{\soc^{\text{max}}}}

\newcommand{\slopeOld}[1]{\ensuremath{\angle^{#1}_{\text{old}}}\xspace}
\newcommand{\slopeNew}[1]{\ensuremath{\angle^{#1}_{\text{new}}}\xspace}
\newcommand{\triptimesubscr}{\operatorname{trip}}
\newcommand{\triptime}{\ensuremath{\atime_{\triptimesubscr}}}
\newcommand{\arrangementtime}{\ensuremath{\atime_{\operatorname{init}}}}
\newcommand{\chargingtime}{\ensuremath{\atime_{\operatorname{charge}}}}
\newcommand{\maxtime}{\ensuremath{\atime^{\operatorname{max}}}}
\newcommand{\socfunction}[1]{\ensuremath{f\langle#1\rangle}}
\newcommand{\minfeasibletime}[1]{\ensuremath{\atime_{\operatorname{min}}(#1)}}
\newcommand{\charginggap}{\ensuremath{\Delta_{\operatorname{charge}}}}
\newcommand{\timegap}{\ensuremath{\Delta_{\operatorname{trip}}}}

\newcommand{\distancelabelWithChargingTime}[1]{\ensuremath{{#1\!\to\alabel}}}

\newcommand{\vertical}[1]{\begin{rotate}{90}#1\end{rotate}}

\definecolor{thesisblue}         {rgb}{0.337, 0.592, 0.773}
\definecolor{thesisblue-dark}    {rgb}{0.212, 0.471, 0.655}
\definecolor{thesisblue-light}   {rgb}{0.490, 0.710, 0.867}
\definecolor{thesisblue-vlight}  {rgb}{0.871, 0.937, 0.988}

\definecolor{thesisred}          {rgb}{0.776, 0.357, 0.396}
\definecolor{thesisred-dark}     {rgb}{0.663, 0.235, 0.271}
\definecolor{thesisred-light}    {rgb}{0.867, 0.502, 0.533}
\definecolor{thesisred-vlight}   {rgb}{0.984, 0.871, 0.878}

\definecolor{thesisgreen}        {rgb}{0.337, 0.765, 0.235}
\definecolor{thesisgreen-dark}   {rgb}{0.267, 0.612, 0.184}
\definecolor{thesisgreen-light}  {rgb}{0.443, 0.871, 0.341}
\definecolor{thesisgreen-vlight} {rgb}{0.710, 0.984, 0.643}

\definecolor{thesisyellow}       {rgb}{0.808, 0.659, 0.263}
\definecolor{thesisyellow-dark}  {rgb}{0.667, 0.529, 0.180}
\definecolor{thesisyellow-light} {rgb}{0.898, 0.765, 0.412}
\definecolor{thesisyellow-vlight}{rgb}{0.992, 0.933, 0.776}

\definecolor{black}              {rgb}{0,0,0}
\definecolor{black70}            {rgb}{0.30,0.30,0.30}
\definecolor{black50}            {rgb}{0.50,0.50,0.50}
\definecolor{black30}            {rgb}{0.70,0.70,0.70}
\definecolor{black25}            {rgb}{0.75,0.75,0.75}
\definecolor{black15}            {rgb}{0.85,0.85,0.85}
\definecolor{black7}             {rgb}{0.93,0.93,0.93}

\colorlet{primarycolor-dark}      {thesisblue-dark}
\colorlet{primarycolor}           {thesisblue}
\colorlet{primarycolor-light}     {thesisblue-light}
\colorlet{primarycolor-vlight}    {thesisblue-vlight}
\colorlet{secondarycolor-dark}    {thesisred-dark}
\colorlet{secondarycolor}         {thesisred}
\colorlet{secondarycolor-light}   {thesisred-light}
\colorlet{secondarycolor-vlight}  {thesisred-vlight}

\usetikzlibrary{arrows}
\usetikzlibrary{decorations.text}
\usetikzlibrary{decorations.pathmorphing}
\usetikzlibrary{through}
\usetikzlibrary{intersections}
\tikzset{>=stealth'}
\tikzstyle{figure}=[font=\footnotesize]
\tikzstyle{descrarrow}=[->,shorten <=4pt, shorten >=4pt]
\tikzstyle{functionlabel}=[pos=0.5,above=2pt,edgelabel,sloped]
\tikzstyle{vertex}=[circle, inner sep=0pt, minimum size=6pt, draw=black, fill=primarycolor-vlight]
\tikzstyle{bigvertex}=[vertex, minimum size=11pt]
\tikzstyle{textvertex}=[bigvertex, font=\footnotesize]
\tikzstyle{highlightedvertex}=[vertex, fill=secondarycolor-light]
\tikzstyle{chargingstationvertex}=[double distance=1pt, fill=secondarycolor-light]
\tikzstyle{vertexhull}=[line width=14pt, rounded corners=2pt]
\tikzstyle{edgelabel}=[inner sep=1pt, fill=white]
\tikzstyle{edges}=[shorten >=1pt, shorten <=1pt]
\tikzstyle{function}=[draw=primarycolor,very thick,line cap=rect]
\tikzstyle{functionterm}=[text=primarycolor,edgelabel]
\tikzstyle{discontinuity}=[circle, inner sep=0pt, minimum size=3.5pt,very thick]
\tikzstyle{discontinuityblank}=[discontinuity,draw=primarycolor,fill=white]
\tikzstyle{discontinuityfilled}=[discontinuity,draw=primarycolor,fill=primarycolor]
\tikzstyle{plotgrid}=[color=black30,very thin,dash pattern=on 2pt off 1pt]
\tikzstyle{chargingfunction}=[function,draw=secondarycolor]
\tikzstyle{helperline}=[black]
\tikzstyle{breakpoint}=[fill,primarycolor,circle,inner sep=0pt,minimum size=3.5pt,very thick]
\tikzstyle{plotpoint}=[fill=black,circle,inner sep=0pt,minimum size=3.5pt]
\tikzstyle{tangentline}=[shorten >=-20pt,shorten <= -15pt]
\tikzstyle{markSign} = [mark=*]
\tikzstyle{heuristicMarkSign} = [mark=o]
\tikzstyle{exactMarkSign} = [mark=*]
\tikzstyle{shortenLines} = [shorten <= 3.5pt,shorten >= 3.5pt]

\pgfdeclarelayer{background}
\pgfdeclarelayer{foreground}
\pgfsetlayers{background,main,foreground}

\pgfplotsset{
  /pgfplots/legend line with mark/.style={
    legend image code/.code={
        \draw[##1,no markers,shorten >= 3.5pt,line width=1pt]
         plot coordinates {
         (0cm,0cm)
         (0.3cm,0cm)
        };
        \draw[##1,only marks,markSign,line width=1pt]
         plot coordinates {
         (0.3cm,0cm)
        };
        \draw[##1,no markers,shorten <= 3.5pt,line width=1pt]
         plot coordinates {
         (0.3cm,0cm)
         (0.6cm,0cm)
        };
    }
  }
}

\pgfplotsset{
  /pgfplots/legend line with heuristic mark/.style={
    legend image code/.code={
        \draw[##1,no markers,shorten >= 3.5pt,line width=1pt]
         plot coordinates {
         (0cm,0cm)
         (0.3cm,0cm)
        };
        \draw[##1,only marks,heuristicMarkSign,line width=1pt]
         plot coordinates {
         (0.3cm,0cm)
        };
        \draw[##1,no markers,shorten <= 3.5pt,line width=1pt]
         plot coordinates {
         (0.3cm,0cm)
         (0.6cm,0cm)
        };
    }
  }
}

\pgfplotsset{
  /pgfplots/legend line with exact mark/.style={
    legend image code/.code={
        \draw[##1,no markers,shorten >= 3.5pt,line width=1pt]
         plot coordinates {
         (0cm,0cm)
         (0.3cm,0cm)
        };
        \draw[##1,only marks,exactMarkSign,line width=1pt]
         plot coordinates {
         (0.3cm,0cm)
        };
        \draw[##1,no markers,shorten <= 3.5pt,line width=1pt]
         plot coordinates {
         (0.3cm,0cm)
         (0.6cm,0cm)
        };
    }
  }
}

\pgfplotsset{
  /pgfplots/custom box/.style={
    legend image code/.code={
       \draw[##1,/tikz/.cd,yshift=-0.37em,xshift=-0.35em,line width=0pt,fill opacity=0.5,draw opacity=0]
        (0cm,0cm) rectangle (13pt,0.8em);
    }
  }
}

\newcommand{\ExampleInftyOffset}{1}
\newcommand{\SoCProfileExampleInftyOffset}{1}

\newcommand{\PlotXAxisName}[3]
{
 \node[edgelabel,above=4pt] at (#1,#2) {#3};
}

\newcommand{\PlotYAxisName}[3]
{
 \node[edgelabel,right=2pt] at (#1,#2) {#3};
}

\newcommand{\DrawDiscontinuitySymbolWithLines}[1]
{
 \pgfgettransformentries{\mya}{\myb}{\myc}{\myd}{\mys}{\myt}
 \pgfmathsetmacro{\grow}{0.75/\mya}
 \begin{scope}[scale=\grow,rotate=#1]
  \clip (-0.15,-0.5/\grow) rectangle (0.15,0.5/\grow);
   
  \draw (0,0.06) -- (0,0.5/\grow);
  \draw (0,-0.06) -- (0,-0.5/\grow);

  \draw[line cap=rect,decorate,decoration={snake,segment length=6pt, amplitude=1.25pt}] (-0.435,0.06) -- (0.5,0.06);
  \draw[line cap=rect,decorate,decoration={snake,segment length=6pt, amplitude=1.25pt}] (-0.435,-0.06) -- (0.5,-0.06);  
 \end{scope}
} 

\newcommand{\SoCProfileExampleDrawCoordinateSystem}[2]
{
\draw [color=black30,very thin,dash pattern=on 2pt off 1pt] (0,0) grid (#1,#2);

\begin{pgfonlayer}{foreground}
 \draw[color=black] (0,0) -- (#1,0);
 \draw[color=black] (0pt,2pt) -- (0pt,-2pt);
 \foreach \x in {1,2,...,#1}
   \draw[shift={(\x,0)}] (0pt,2pt) -- (0pt,-2pt) node[below] {$\x$};

 \draw[color=black] (0,0) -- (0,#2);
 \foreach \y in {0,1,...,#2}
   \draw[shift={(0,\y)}] (2pt,0pt) -- (-2pt,0pt) node[left] {$\y$};
 \draw[shift={(0,-\SoCProfileExampleInftyOffset)}] (2pt,0pt) -- (-2pt,0pt) node[left] {$-\infty$};
 \draw[shift={(0,-\SoCProfileExampleInftyOffset)}] (0pt,2pt) -- (0pt,-2pt);
 
 \begin{scope}[shift={(0,-0.5*\SoCProfileExampleInftyOffset)}]
  \DrawDiscontinuitySymbolWithLines{0}
 \end{scope}

\end{pgfonlayer}
}

\newcommand{\SocProfileExampleClipAtPlotBorders}[1]
{
 \begin{pgfinterruptboundingbox}
  \clip (0,-\SoCProfileExampleInftyOffset-1) rectangle (#1,#1+1);
 \end{pgfinterruptboundingbox}
}

\newcommand{\SoCProfileExampleDrawFunction}[5]
{
\begin{scope}
  \SocProfileExampleClipAtPlotBorders{#1}
  \draw [function] (0,-\SoCProfileExampleInftyOffset) -- (#2,-\SoCProfileExampleInftyOffset);
  \draw [function] (#2,#3) -- (#4,#5) -- (#1,#5);
\end{scope}

\begin{pgfonlayer}{foreground}
 \node[discontinuityblank] (d1) at (#2,-\SoCProfileExampleInftyOffset) {};
 \node[discontinuityfilled] (d2) at (#2,#3) {};
\end{pgfonlayer}
}

\newcommand{\ExampleDrawCoordinateSystem}[6]
{
 \begin{scope}
  \begin{pgfinterruptboundingbox}
   \clip (#5,#3-1) rectangle (#2+1,#4+1);
   \draw[plotgrid] (#1-1,#3) grid (#2,#4);
  \end{pgfinterruptboundingbox}
 \end{scope}
 
 \begin{pgfonlayer}{foreground}
  \begin{scope}
   \begin{pgfinterruptboundingbox}
    \clip (#5,#3-1) rectangle (#2+1,#4+1);
    \draw[color=black] (#1-1,#6) -- (#2,#6);
    \draw[shift={(#1-1,#6)},color=black] (0pt,2pt) -- (0pt,-2pt);  
   \end{pgfinterruptboundingbox}
  \end{scope}

  \draw[color=black] (#5,#3) -- (#5,#4);
 \end{pgfonlayer}
 
 \node at (#2,#4) {};
}

\newcommand{\ExampleDrawCoordinateSystemWithTicks}[6]
{
 \ExampleDrawCoordinateSystem{#1}{#2}{#3}{#4}{#5}{#6}

 \begin{pgfonlayer}{foreground}
  \foreach \x in {#1,...,#2}
   \draw[shift={(\x,#6)}] (0pt,2pt) -- (0pt,-2pt) node[below=2pt,inner sep=1pt,fill=white] {\vphantom{$\x$}};
  \foreach \x in {#1,...,#2}
   \draw[shift={(\x,#6)}] (0pt,2pt) -- (0pt,-2pt) node[below=2pt,inner sep=1pt] {$\x$};
   
  \foreach \y in {#3,...,#4}
   \draw[shift={(#5,\y)}] (2pt,0pt) -- (-2pt,0pt) node[left] {$\y$};
 \end{pgfonlayer}
}

\newcommand{\ExampleDrawCoordinateSystemWithPositiveInfty}[6]
{
 \ExampleDrawCoordinateSystemWithTicks{#1}{#2}{#3}{#4}{#5}{#6}

 \begin{pgfonlayer}{foreground}
  \draw[shift={(#5,#4+\ExampleInftyOffset)}] (2pt,0pt) -- (-2pt,0pt) node[left] {$\infty$};
  \draw[shift={(#5,#4+\ExampleInftyOffset)}] (0pt,2pt) -- (0pt,-2pt);
 
  \begin{scope}[shift={(#5,#4+0.5*\ExampleInftyOffset)}]
   \DrawDiscontinuitySymbolWithLines{0}
  \end{scope}
 \end{pgfonlayer}
}

\newcommand{\ExampleDrawCoordinateSystemWithNegativeInfty}[6]
{
 \ExampleDrawCoordinateSystemWithTicks{#1}{#2}{#3}{#4}{#5}{#6}

 \begin{pgfonlayer}{foreground}
  \draw[shift={(#5,#3-\ExampleInftyOffset)}] (2pt,0pt) -- (-2pt,0pt) node[left] {$-\infty$};
  \draw[shift={(#5,#3-\ExampleInftyOffset)}] (0pt,2pt) -- (0pt,-2pt);
 
  \begin{scope}[shift={(#5,#3-0.5*\ExampleInftyOffset)}]
   \DrawDiscontinuitySymbolWithLines{0}
  \end{scope}
 \end{pgfonlayer}
}

\newcommand{\SocProfileExampleScale}{0.65}
\newcommand{\ChargingFunctionExampleScale}{0.65}
\newcommand{\ChargingFunctionBoundLinkScale}{0.65}

\usepackage[
colorlinks=true,
linkcolor=thesisred-dark,
urlcolor=thesisred-dark,
citecolor=thesisblue-dark,
pdfauthor={Moritz Baum, Julian Dibbelt, Andreas Gemsa, Dorothea Wagner, and Tobias Zündorf},
pdftitle={Shortest Feasible Paths with Charging Stops for Battery Electric Vehicles}
]{hyperref}

\usepackage[nopostdot,style=super,nonumberlist,toc,nowarn]{glossaries}

\makeglossaries

\newacronym[description={Battery Electric Vehicle}]     {ev}{EV}{battery electric vehicle}
\newacronym[description={State of Charge}]              {soc}{SoC}{state of charge}
\newacronym{csp}{CSP}{Constrained Shortest Path}
\newacronym{alt}{ALT}{A*, Landmarks, Triangle Inequality}
\newacronym{ch}{CH}{Contraction Hierarchies}
\newacronym{cch}{CCH}{Customizable Contraction Hierarchies}
\newacronym{srtm}{SRTM}{Shuttle Radar Topography Mission}
\newacronym{phem}{PHEM}{Passenger Car and Heavy Duty Emission Model}
\newacronym{hbefa}{HBEFA}{Handbook on Emission Factors for Road Traffic}
\newacronym{osm}{OSM}{OpenStreetMap}
\newacronym[description={First-In-First-Out}]           {fifo}{FIFO}{first-in-first-out}
\newacronym{ed}{ED}{Edge Difference}
\newacronym{cq}{CQ}{Cost of Queries}
\newacronym{sc}{SC}{Shortcut Complexity}
\newacronym{dn}{DN}{Deleted Neighbors}
\newacronym[description={Bicriteria Shortest Path (Algorithm)}]      {bsp}{BSP}{bicriteria shortest path}
\newacronym[description={Charging Function Propagating (Algorithm)}] {cfp}{CFP}{charging function propagating}
\newacronym                                             {charge}{CHArge}{CH, A*, Charging Stops}
\newacronym[description={Battery Swapping Station}]     {bss}{BSS}{battery swapping stations}
\newacronym[description={Fully Polynomial-Time Approximation Scheme}] {fptas}{FPTAS}{fully polynomial-time approximation scheme}

\newtheorem{theorem}{Theorem}
\newtheorem{lemma}{Lemma}

\newcommand{\email}[1]{\texttt{#1}}
\title{Shortest Feasible Paths with Charging Stops for Battery Electric Vehicles\thanks{This paper is the full version of an extended abstract~\cite{Bau15} published at the 23rd ACM SIGSPATIAL International Conference on Advances in Geographic Information Systems.
It builds upon theses of two of the authors~\cite{Bau18b,Zue15}.
Our work was supported by DFG Research Grant WA\,654/16--2 and DFG Research Grant WA\,654/23--1.}}
\author[1]{Moritz Baum}
\author[2]{Julian Dibbelt}
\author[1]{Andreas Gemsa}
\author[1]{Dorothea Wagner}
\author[1]{Tobias Zündorf}
\affil[1]{Karlsruhe Institute of Technology, Germany\\\email{moritz@ira.uka.de,gemsa@ira.uka.de,dorothea.wagner@kit.edu,tobias.zuendorf@kit.edu}}
\affil[2]{Sunnyvale, CA, United States\\\email{algo@dibbelt.de}}
\date{}

\begin{document}
%%%%%%%%%%%%%%%%

\maketitle

\begin{abstract}
We study the problem of minimizing \emph{overall} trip time for battery electric vehicles in road networks. As battery capacity is limited, stops at charging stations may be inevitable. Careful route planning is crucial, since charging stations are scarce and recharging is time-consuming.
We extend the Constrained Shortest Path problem for electric vehicles with \emph{realistic} models of charging stops, including varying charging power and battery swapping stations.
While the resulting problem is \NP-hard, we propose a combination of algorithmic techniques to achieve good performance in practice.
Extensive experimental evaluation shows that our approach~(CHArge) enables computation of \emph{optimal} solutions on realistic inputs, even of continental scale.
Finally, we investigate heuristic variants of CHArge that derive high-quality routes in well below a second on sensible instances.
\end{abstract}

%%%%%%%%%%%%%%%%%%%%%%%%%%%%%%%%%%%%%%%%%%%%%%%%%%%%%%%%%%%%%%%%%%%%%%%%%%%%%%%%
\section{Introduction}\label{sec:introduction}
%%%%%%%%%%%%%%%%%%%%%%%%%%%%%%%%%%%%%%%%%%%%%%%%%%%%%%%%%%%%%%%%%%%%%%%%%%%%%%%%

Electromobility promises independence of fossil fuels, zero (local) emissions, and higher energy-efficiency, especially for city traffic. Yet, most of the significant algorithmic progress on route planning in road networks has focused on conventional combustion-engine cars; see Bast~et~al.~\cite{Bas14} for an overview.
\Glspl*{ev}, however, are different: Most importantly, for the foreseeable future, they have limited driving range, charging stations are still much rarer than gas stations, and recharging is time-consuming.
Thus, routes can become infeasible~(the battery runs empty), and fast routes may be less favorable when accounting for longer recharging time.
In this work, we aim to find the overall fastest~(capacity-constrained) route, considering charging stops when necessary.

\paragraph{Related Work.}
Classic route planning approaches apply Dijkstra's algorithm~\cite{Dij59} to a graph representation of the transportation network, using fixed scalar arc weights corresponding to, \eg, driving time. 
For faster queries, \emph{speedup techniques} have been proposed, with different benefits in terms of preprocessing time and space, query speed, and simplicity~\cite{Bas14}. A*~search~\cite{Har68} uses vertex potentials to guide the search towards the target.
A successful variant, \acrshort*{alt}\glsunset{alt} (\acrlong*{alt})~\cite{Gol05b}, obtains good potentials from precomputed distances to selected landmark vertices.
\Gls*{ch}~\cite{Gei12b}, on the other hand, iteratively \emph{contract} vertices in increasing order of (heuristic) importance during preprocessing, maintaining distances between all remaining vertices by adding \emph{shortcut} arcs where necessary. The \gls*{ch} query is then bidirectional, starting from source and target, and proceeds only from less important to more important vertices.
Combining both techniques, Core-\gls*{alt}~\cite{Bau08b} contracts all but the most important vertices~(\eg, the top 1\,\%), performing \gls*{alt} on the remaining \emph{core} graph. This approach can also be extended to more complex scenarios, such as edge constraints to model,~\eg,~maximum allowed vehicle height or weight~\cite{Gei10b}.
More recently, techniques (including variants of \gls*{ch} and \gls*{alt}) were introduced that allow an additional \emph{customization} after preprocessing, to account for dynamic or user-dependent metrics~\cite{Del13b,Dib16,Efe13a}. Also, approaches towards extended scenarios exist, such as batched shortest paths~\cite{Del11a} or shortest via paths~\cite{Abr12b,Del13a}. Here, a common approach is to make use of a relatively fast \emph{target selection} phase, precomputing distances to relevant points of interest to enable faster queries. However, these techniques were only evaluated for single-criterion search, where the distance between two vertices is always a unique scalar value. For multi-criteria scenarios, on the other hand, problem complexity and solution sizes increase significantly, and practical approaches are only known for basic problem variants~\cite{Fun13,Gei10b,Sto12d}. For a more complete overview of techniques and combinations, see Bast~et~al.~\cite{Bas14}.

Regarding route planning for \glspl*{ev}, the battery must not fully deplete during a ride, but energy is recuperated when braking~(\eg, going downhill), up to the maximum capacity. These \emph{battery constraints} must be checked during route computation.
This can be done implicitly, by weighting each road segment with a \emph{profile}, mapping current \emph{\gls*{soc}} to actual consumption \wrt battery constraints~\cite{Bau17b,Eis11,Sch14a}.
If battery constraints have to be obeyed, both preprocessing and queries of speedup techniques are slower compared to the standard shortest path problem~\cite{Bau13a,Eis11,Sto12b}.
Several natural problem formulations specific to EV routing exist.
Some optimize energy consumption~\cite{Bau13a,Bau17b,Eis11,Sac11,Sch14a}, often, however, resulting in routes resorting to minor~(\ie, slow) roads to save energy.
\emph{\Gls*{csp}} formulations~\cite{Han80a} ask to find the most energy-efficient route without exceeding a certain driving time---or finding the fastest route that does not violate battery constraints~\cite{Sto12c}.
More recent works additionally take speed planning into account to trade driving speed for energy consumption~\cite{Bau17a,Har14}. These works derive exact and heuristic algorithms, which achieve good query times after preprocessing the input network, but do not include charging stops in route optimization.
Another approach separates queries into two phases, optimizing driving time and energy consumption, respectively~\cite{Goo14}.
This allows for reasonably fast queries (even without preprocessing), but results are not optimal in general, \ie, there might be paths with lower trip time that respect battery constraints.

Without recharging, large parts of the road network are simply not reachable by an~\gls*{ev}, rendering long-distance trips impossible. For conventional cars, broad availability of gas stations and short refuel duration allow to neglect this issue in route optimization.
Strehler~et~al.~\cite{Mer15} give theoretical insights into the \gls*{csp} problem for~\glspl*{ev} including charging stops. Most importantly, they develop a \gls*{fptas} for this problem. Unfortunately, the algorithm is slow in practice.

Often, previous works dealing with the \gls*{csp} problem for \glspl*{ev} have considered charging stations under the simplifying assumptions that the charging process takes \emph{constant} time~(independent of the initial \gls*{soc}) and always results in a \emph{fully} recharged battery~\cite{Goo14,Sto12b,Sto12c}.
Then, feasible paths between charging stations are independent of source and target, hence easily precomputed. Routes with a minimum number of intermediate charging stops can then be computed in less than a second on subcountry-scale graphs~\cite{Sto12b,Sto12c}.
Nevertheless, the simple model used in these works only applies to battery swapping stations, which are still an unproven technology and business model.

For regular charging stations, charging time depends on both the current \gls*{soc} and the desired \gls*{soc}, as well as the charging power provided by the station.
Kobayashi~et~al.~\cite{Kob11} distinguish slow and fast charging stations, but assume that the battery is always fully recharged. They propose a heuristic search based on preselected candidates of charging stations for a given pair of source and target. Their approach takes several seconds for suboptimal results on metropolitan-scale routes.

In reality, while nearly linear for low~\gls*{soc}, the charging rate decreases when approaching the battery's limit. Thus, it can be reasonable to only charge up to a \emph{fraction} of the limit.
Sweda~et~al.~\cite{Swe14} reflect this behavior by combining a linear with an exponential function for high \gls*{soc}. Liu~et~al.~\cite{Liu14} model charging between 0\,\% and 80\,\%~\gls*{soc} with a linear function, but recharging above 80\,\%~\gls*{soc} is suppressed altogether.
Neither approach was shown to scale to road networks of realistic size.
Also, even though omitting the possibility of charging beyond 80\% might be appropriate for regions well covered with charging stations, it drastically deteriorates reachability in regions with few stations, where recharging to a full battery can be inevitable~\cite{Mon15}.

Other works discretize possible charging durations to model different options~\cite{Hub15,Mer15,Wan13}. This enables search algorithms that closely resemble the well-known bicriteria shortest path algorithm~\cite{Han80b}.

Some recent works only optimize energy consumption along routes with charging stops, while ignoring travel time entirely to simplify the problem setting~\cite{Bau17b,Sto12b}.
A technique based on dynamic programming to optimize a certain cost function on a route with charging stops is given by Sweda and Klabjan~\cite{Swe12}, extending previous approaches on refueling strategies for conventional cars~\cite{Lin07,Khu11}.
Similarly, Liao~et~al.~\cite{Lia16} use dynamic programming to solve different routing problems for~\glspl*{ev}, including the possibility of battery swaps along the way.

The publications mentioned so far deal with shortest path problems for a \emph{single}~\gls*{ev}.  In addition to that, charging stops are also considered in several works dealing with \emph{vehicle (fleet) routing problems} for \glspl*{ev}~\cite{Des16,Goe15,Lia16,Mon15,Mur16,Pou16,Sch18,Sch14b,Yan15}. These problems are handled by means of mathematical programming or heuristics. Aiming at complex scenarios (often involving optimization for \emph{multiple} vehicles and variants of the Traveling Salesman Problem), the approaches do not scale to large input instances.
Furthermore, most of them rely on the simplified charging models discussed above, \eg, allowing only battery swaps or (uniform) linear charging. A notable exception is the work of Montoya~et~al.~\cite{Mon15}, which explicitly models nonlinear charging processes to achieve more realistic results when computing charging stops for a vehicle along a given route. In fact, it builds upon the model we propose (and generalize) in this work, available from previous publications~\cite{Bau15,Zue15}.

\paragraph{Contribution and Outline.}
In this work, we extend the \gls*{csp} problem to planning routes that, while respecting battery constraints, minimize overall trip time, including time spent at charging stations.
Unlike previous works, our solution handles \emph{all} types of charging stations accurately: battery swapping stations, regular charging stations with various charging powers, as well as superchargers.
In particular, charging times are \emph{not} independent of the current \gls*{soc} when arriving at a charging station. Additionally, the charging process can be interrupted as soon as further charging would delay the arrival at the target.
This results in a challenging, \NP-hard problem, for which even the construction of basic exponential-time algorithms is nontrivial. We propose a label-setting search, which is capable of propagating \emph{continuous} tradeoffs induced by charging stops using labels of \emph{constant} size.
Since the problem is \NP-hard, we do not guarantee polynomial running times. However, carefully incorporating recharging models in speedup techniques, our approach is the first to solve a realistic setting \emph{optimally} and within seconds, for road networks of country-scale and beyond.
Key ingredients of our speedup techniques are a bicriteria extension of \gls*{ch} and generalizations of A*~search, where we incorporate \gls*{soc} and required charging time into the derived bounds.
For even faster queries, we propose heuristic approaches that offer high (empirical) quality. Extensive experiments on detailed and realistic data show that our approach produces sensible results and, even though it is designed to solve a more complex and realistic problem, clearly outperforms and broadens the state-of-the-art.

The remainder of this work is organized as follows. Section~\ref{sec:problem} sets necessary notation and specifies our model and the problem. Section~\ref{sec:approach} presents our basic label-propagating approach.
It turns out that as its most crucial ingredient, a limited number of new labels must be generated at charging stations to ensure termination and correctness of the approach. We derive methods to construct such labels in Section~\ref{sec:spawning}. To improve practical performance of our exponential-time algorithm, we propose tuning based on A*~search in Section~\ref{sec:astar} and the speedup technique \gls*{ch} in Section~\ref{sec:ch}. We also discuss how both techniques can be combined to further reduce practical running times in Section~\ref{sec:charge}.
Section~\ref{sec:heuristics} introduces heuristics, which drop correctness for faster queries. Section~\ref{sec:experiments} experimentally evaluates all approaches on large, realistic input. We close with final remarks in Section~\ref{sec:conclusion}.

%%%%%%%%%%%%%%%%%%%%%%%%%%%%%%%%%%%%%%%%%%%%%%%%%%%%%%%%%%%%%%%%%%%%%%%%%%%%%%%%
\section{Model and Problem Statement}\label{sec:problem}
%%%%%%%%%%%%%%%%%%%%%%%%%%%%%%%%%%%%%%%%%%%%%%%%%%%%%%%%%%%%%%%%%%%%%%%%%%%%%%%%

We consider \emph{directed, weighted}~graphs~$\graph = (\vertices, \arcs)$ with vertices~$\vertices$ and arcs~$\arcs \subseteq \vertices \times \vertices$, together with two arc weight functions $\drivingtimefunction \colon \arcs \to \posreals$ and $\consumptionfunction \colon \arcs \to \reals$, representing driving time and energy consumption on an arc, respectively.
Vertices are \emph{neighbors} if they are connected by an arc.
We simplify notation by defining $\drivingtimefunction(\vertexa,\vertexb):=\drivingtimefunction((\vertexa,\vertexb))$ and $\consumptionfunction(\vertexa,\vertexb):=\consumptionfunction((\vertexa,\vertexb))$ for an arc~$(\vertexa,\vertexb)\in\arcs$.

An $\source$--$\target$ \emph{path} in a graph~$\graph$ is defined as a sequence $\apath~=~[\source=\vertex_1,\vertex_2\dots,\vertex_k=\target]$ of vertices, such that~$(\vertex_i,\vertex_{i+1})\in\arcs$ for~$1\le i < k$. If $\source = \target$, we call $\apath$ a \emph{cycle}.
Given two paths $\patha = [\vertex_1, \dots, \vertex_i]$ and $\pathb = [\vertex_i, \dots, \vertex_k]$, we denote by $\patha \circ \pathb := [\vertex_1, \dots, \vertex_i, \dots, \vertex_k]$ their concatenation.
The \emph{driving time} on a path~$\apath$ is $\drivingtimefunction(\apath) = \sum_{i=1}^{k-1} \drivingtimefunction(\vertex_i, \vertex_{i+1})$. For energy consumption, this is more involved:
First of all, note that consumption can be negative due to recuperation, though cycles with negative consumption are physically ruled out.
Moreover, the battery has a limited capacity~$\maxbattery \in \posreals$, and the \gls*{soc} can neither exceed this limit nor drop below~$0$. Therefore, we must take \emph{battery constraints} into account.
For an~$\source$--$\target$~path~$\apath$ to be \emph{feasible}, the battery's \gls*{soc} must be within the interval~$[0,\maxbattery]$ at every vertex of~$\apath$.
To reflect battery constraints on an~$\source$--$\target$~path~$\apath$, we define its \emph{\gls*{soc} profile} $\socprofile_\apath\colon[0,\maxbattery]\to[0,\maxbattery]\cup\{-\infty\}$ that maps \gls*{soc}~$\soc_\source$ at the source~$\source$ to the resulting \gls*{soc}~$\soc_\target$ at the target~$\target$. Hence, the \gls*{soc} at~$\target$ is~$\soc_\target = \socprofile_\apath(\soc_\source)$; see Figure~\ref{fig:socprofile-linking} for an example.  
We use the value~$-\infty$ to indicate that the \gls*{soc} at~$\source$ is not sufficient to reach~$\target$, \ie, $\apath$ is not feasible for the corresponding \gls*{soc} at~\source.
As shown by Eisner~et~al.~\cite{Eis11}, $\socprofile_\apath$ can be represented using only three values, namely, the minimum \gls*{soc}~$\profileMinBat_\apath\in [0, \maxbattery]$ required to traverse~$\apath$, its energy consumption~$\profileCost_\apath \in [-\maxbattery, \maxbattery]$  (which can be less than~$\profileMinBat_\apath$ due to recuperation), and the maximum possible \gls*{soc} after traversing~$\apath$, denoted~$\profileMaxBat_\apath\in [0, \maxbattery]$.
We then have
\begin{align*}
 \socprofile_\apath(\soc) :=
 \begin{cases}
 -\infty                    & \mbox{if } \soc < \profileMinBat_\apath, \\
 \profileMaxBat_\apath      & \mbox{if } \soc - \profileCost_\apath > \profileMaxBat_\apath, \\
 \soc - \profileCost_\apath & \mbox{otherwise.}
 \end{cases}
\end{align*}
The \gls*{soc} profile~$\socprofile_{[\vertex]}$ of a path consisting of a single vertex~$\vertex\in\vertices$ is given by~$\profileMinBat_{[\vertex]}:=0$, $\profileCost_{[\vertex]}:=0$, and~$\profileMaxBat_{[\vertex]}:=\maxbattery$, which yields $\socprofile_{[\vertex]}(\soc)=\soc$ for all~$\soc\in[0,\maxbattery]$.
For an arc~$\arc=(\vertexa, \vertexb) \in \arcs$, the profile $\socprofile_{\arc} = \socprofile_{[\vertexa, \vertexb]}$ is given by $\profileCost_{\arc}:=\consumptionfunction(\arc)$, $\profileMinBat_{\arc}:=\max\{0,\consumptionfunction(\arc)\}$, and $\profileMaxBat_{\arc}:=\min\{\maxbattery,\maxbattery-\consumptionfunction(\arc)\}$.
For two \gls*{soc} profiles~$\socprofile_\patha$ and $\socprofile_\pathb$ of paths~$\patha$ and~$\pathb$, we obtain the \emph{linked} profile~$\socprofile_{\patha \circ \pathb}$ by setting
\begin{align*}
 \profileMinBat_{\patha \circ \pathb} & := \max \{ \profileMinBat_\patha, \profileCost_\patha + \profileMinBat_\pathb \}                                    \\
 \profileMaxBat_{\patha \circ \pathb} & := \minx \{ \profileMaxBat_\pathb, \profileMaxBat_\patha - \profileCost_\pathb \}                                   \\
 \profileCost_{\patha \circ \pathb}   & := \max \{ \profileCost_\patha + \profileCost_\pathb, \profileMinBat_\patha - \profileMaxBat_\pathb \}\text{,}
\end{align*}
provided that $\profileMaxBat_\patha \ge \profileMinBat_\pathb$. Otherwise, the path is infeasible for arbitrary \gls*{soc}, and thus $\socprofile_{\patha \circ \pathb}\equiv-\infty$.
An example of two \gls*{soc} profiles as well as the result of linking them is shown in Figure~\ref{fig:socprofile-linking}.
Finally, given two \gls*{soc} profiles $\socprofile_1$ and~$\socprofile_2$, we say that $\socprofile_1$ \emph{dominates} $\socprofile_2$ if $\socprofile_1(\soc)\ge\socprofile_2(\soc)$ holds for all~$\soc\in[0,\maxbattery]$.

\begin{figure}[t]
 \centering%
 \begin{subfigure}[b]{.14\textwidth}%
 \centering%
 \begin{tikzpicture}[figure,yscale=1.0]
  \node[textvertex] (a) at (0,0) {$\source$};
  \node[textvertex] (b) at (0,-1) {$\vertexa$};
  \node[textvertex] (c) at (0,-2) {$\vertexb$};
  \node[textvertex] (d) at (0,-3) {$\vertexc$};
  \node[textvertex] (e) at (0,-4) {$\target$};
  
  \path[edges] (a) edge node[edgelabel] {$2$} (b)
               (b) edge node[edgelabel] {$-3$} (c)
               (c) edge node[edgelabel] {$-2$} (d)
               (d) edge node[edgelabel] {$3$} (e);
               
  \draw[<->] (0.5,0.1) -- (0.5,-1.95) node [midway,right] {$\patha$};
  \draw[<->] (0.5,-2.05) -- (0.5,-4.1) node [midway,right] {$\pathb$};
 \end{tikzpicture}%
 \caption{}%
 \label{fig:socprofile-linking:graph}%
 \end{subfigure}%
 \begin{subfigure}[b]{.28\textwidth}%
 \centering%
 \begin{tikzpicture}[figure,scale=\SocProfileExampleScale]
  \SoCProfileExampleDrawCoordinateSystem{4.0}{4.0}
  \begin{scope}
   \clip (0,0) rectangle (4,4);
   \draw[color=black30,line cap=rect] (0,0) -- (4,4);
  \end{scope}
  \SoCProfileExampleDrawFunction{4.0}{2.0}{3.0}{3.0}{4.0}
  
  \PlotXAxisName{4.4}{-0.5}{$\soc$}
  \PlotYAxisName{0}{4}{$\socprofile_{\patha}(\soc)$}
  
  \draw[<->] (0.2,3) -- (1.8,3) node [midway,above,edgelabel] {$\profileMinBat_{\patha}$};
  \draw[<->] (4,0.2) -- (4,3.8) node [midway,left,edgelabel] {$\profileMaxBat_{\patha}$};
  \draw[<->] (2,2.2) -- (2,2.8) node [midway,left=2pt,edgelabel] {$\profileCost_{\patha}$};
 \end{tikzpicture}%
 \caption{}%
 \label{fig:socprofile-linking:firstpath}%
 \end{subfigure}%
 \begin{subfigure}[b]{.28\textwidth}%
 \centering%
 \begin{tikzpicture}[figure,scale=\SocProfileExampleScale]
  \SoCProfileExampleDrawCoordinateSystem{4.0}{4.0}
  \begin{scope}
   \clip (0,0) rectangle (4,4);
   \draw[color=black30,line cap=rect] (0,0) -- (4,4);
  \end{scope}
  \SoCProfileExampleDrawFunction{4.0}{1.0}{0.0}{2.0}{1.0}
  
  \PlotXAxisName{4.4}{-0.5}{$\soc$}
  \PlotYAxisName{0}{4}{$\socprofile_{\pathb}(\soc)$}
  
  \draw[<->] (0.2,0.2) -- (0.8,0.2) node [midway,above=2pt,edgelabel] {$\profileMinBat_{\patha}$};
  \draw[<->] (4,0.2) -- (4,0.8) node [midway,left=2pt,edgelabel] {$\profileMaxBat_{\pathb}$};
  \draw[<->] (2,1.2) -- (2,1.8) node [midway,right=2pt,edgelabel] {$\profileCost_{\pathb}$};
 \end{tikzpicture}%
 \caption{}%
 \label{fig:socprofile-linking:secondpath}%
 \end{subfigure}%
 \begin{subfigure}[b]{.28\textwidth}%
 \centering%
 \begin{tikzpicture}[figure,scale=\SocProfileExampleScale]
  \SoCProfileExampleDrawCoordinateSystem{4.0}{4.0}
  \begin{scope}
   \clip (0,0) rectangle (4,4);
   \draw[color=black30,line cap=rect] (0,0) -- (4,4);
  \end{scope}
  \SoCProfileExampleDrawFunction{4.0}{2.0}{1.0}{4.0}{1.0}
  
  \PlotXAxisName{4.4}{-0.5}{$\soc$}
  \PlotYAxisName{0}{4}{$\socprofile_{\patha\circ\pathb}(\soc)$}
  
  \draw[<->] (0.2,1) -- (1.8,1) node [midway,above=2pt,edgelabel] {$\profileMinBat_{\patha\circ\pathb}$};
  \draw[<->] (4,0.2) -- (4,0.8) node [midway,left=2pt,edgelabel] {$\profileMaxBat_{\patha\circ\pathb}$};
  \draw[<->] (2,1.2) -- (2,1.8) node [midway,right=2pt,edgelabel] {$\profileCost_{\patha\circ\pathb}$};
 \end{tikzpicture}
 \caption{}%
 \label{fig:socprofile-linking:resultpath}%
 \end{subfigure}%
 \caption{Two \gls*{soc} profiles and the result after linking them. The battery capacity is~$\maxbattery=4$. (a)~The underlying paths~$\patha=[\source,\vertexa,\vertexb]$ and~$\pathb=[\vertexb,\vertexc,\target]$, with indicated arc costs. (b)~The \gls*{soc}~profile~$\socprofile_{\patha}$ is represented by~$\profileMinBat_{\patha}=2$,~$\profileCost_{\patha}=-1$, and~$\profileMaxBat_{\patha}=4$. (c)~The \gls*{soc}~profile~$\socprofile_{\pathb}$ is given by~$\profileMinBat_{\pathb}=1$,~$\profileCost_{\pathb}=1$, and~$\profileMaxBat_{\pathb}=1$. (d)~Linking the profiles $\socprofile_\patha$ and $\socprofile_\pathb$ yields the profile $\socprofile_{\patha\circ\pathb}$ with~$\profileMinBat_{\patha\circ\pathb}=2$,~$\profileCost_{\patha\circ\pathb}=1$, and~$\profileMaxBat_{\patha\circ\pathb}=1$. Observe that due to a subpath of length $-5<-\maxbattery$, the value~$\profileCost_{\patha\circ\pathb}=1$ is greater than the sum~$\profileCost_{\patha}+\profileCost_{\pathb}=\consumptionfunction(\patha\circ\pathb)=0$.
 }%
\label{fig:socprofile-linking}%
\end{figure}
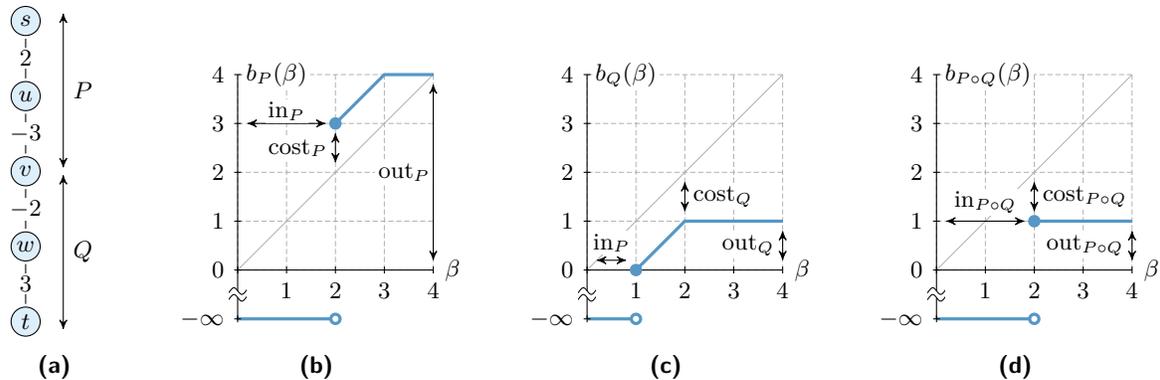

\paragraph{Constrained Shortest Paths.}
Given a graph, two (nonnegative) functions representing length (driving time in our case) and consumption on its arcs, and vertices~$\source\in\vertices$ and~$\target\in\vertices$, the \gls*{csp} Problem asks for a path of minimum length such that its consumption does not exceed a certain bound~$\soc_\source\in\posreals$. Being \NP-hard~\cite{Gar79,Han80a}, \gls*{csp} can be solved using an exponential-time bicriteria variant of Dijkstra's algorithm~\cite{Mar84}, which we refer to as \emph{\gls*{bsp}} algorithm.
It maintains a \emph{label set}~$\labelset(\cdot)$ of \emph{labels} for each vertex~$\vertex \in \vertices$.
In our scenario, each label is a tuple of driving time and~\gls*{soc}.
A label \emph{(Pareto)~dominates} another label of the same vertex if it is better in one criterion and not worse in the other. Initially, all label sets are empty, except for the label~$(0,\soc_\source)$ at the source~$\source$, which is also inserted into a priority queue.
In each step, the algorithm \emph{settles} the minimum label~$\alabel$ of the queue. This is done by extracting the label~$\alabel=(\atime,\soc)$ assigned to some vertex $\vertexa\in\vertices$ and scanning all arcs~$(\vertexa, \vertexb)\in\arcs$ outgoing from $\vertexa$.
If the new label~$\alabel' := (\atime + \drivingtimefunction(\vertexa, \vertexb), \soc - \consumptionfunction(\vertexa, \vertexb))$ is not dominated by any label in $\labelset(\vertexb)$, it is added to $\labelset(\vertexb)$ and the queue, removing labels dominated by $\alabel'$ from~$\labelset(\vertexb)$ and the queue. Using driving time as key of tuples in the priority queue (breaking ties by~\gls*{soc}, \ie, giving preference to the label with highest \gls*{soc} if two or more labels have the same driving time), the algorithm is \emph{label setting}, \ie, extracted labels are never dominated. An optimal (constrained) path is then found once a label with nonnegative \gls*{soc} is extracted at~$\target$.
Battery constraints can be incorporated on-the-fly by additional checks during the algorithm: When scanning an arc~$(\vertexa,\vertexb)\in\arcs$, we set the \gls*{soc} of the new label~$\alabel'$ to~$\min\{\maxbattery,\soc-\consumptionfunction(\vertexa,\vertexb)\}$. If this \gls*{soc} is negative, we discard the label~$\alabel'$.

\paragraph{Speedup Techniques.}
A \emph{potential function}~$\potential\colon\vertices\to\posreals$ on the vertices is~\emph{consistent} (\wrt driving time) if the \emph{reduced arc costs} $\reduceddrivingtimefunction(\vertexa, \vertexb):=\drivingtimefunction(\vertexa, \vertexb) - \potential(\vertexa) + \potential(\vertexb)$ are nonnegative for all~$(\vertexa, \vertexb) \in \arcs$.
The A*~algorithm~\cite{Har68} adds the potential of a vertex to the keys of its labels, changing the order in which they are extracted from the queue. The A*~algorithm was extended to the multicriteria case by Mandow and P{\'e}rez-de-la-Cruz~\cite{Man10}.

In \gls*{ch}~\cite{Gei12b}, extended to multicriteria scenarios by Geisberger~et~al.~\cite{Gei10b} and Funke and Storandt~\cite{Fun13}, vertices are iteratively contracted during preprocessing according to a given vertex order, while introducing \emph{shortcuts} between their neighbors to maintain distances (\wrt $\drivingtimefunction$ and $\consumptionfunction$). To avoid unnecessary shortcuts, \emph{witness searches} are run between neighbors to identify existing paths that dominate a shortcut candidate. Adding shortcuts may lead to nondominated multi-arcs.
To answer~$\source$--$\target$-queries, the original graph is enriched with all shortcuts added during the contraction phase. The \gls*{ch} query then runs a bidirectional variant of the \gls*{bsp} algorithm, starting from both~$\source$ and~$\target$, but scanning only arcs leading to vertices of higher importance (\wrt the vertex order).

\paragraph{Modeling Charging Stops.}
In this work, we consider stops at charging stations to recharge the battery (while spending charging time).
In our model, a subset~$\chargingstations\subseteq\vertices$ of the vertices represents charging stations.
Each vertex~$\vertex\in\chargingstations$ has a designated \emph{charging function}~$\chargingfunction_\vertex\colon[0,\maxbattery]\times\posreals\to[0,\maxbattery]$, which maps \emph{arrival \gls*{soc}} and the spent charging time to the resulting \emph{departure \gls*{soc}}.
We presume that charging functions are continuous and monotonically increasing \wrt charging time (\ie, charging for a longer time never decreases the~\gls*{soc}).
Further, we assume that for arbitrary charging times~$\atime_1\in\posreals$, $\atime_2\in\posreals$, and \gls*{soc} values~$\soc\in[0,\maxbattery]$, the \emph{shifting property} $\chargingfunction_\vertex(\chargingfunction_\vertex(\soc,\atime_1),\atime_2)=\chargingfunction_\vertex(\soc,\atime_1+\atime_2)$ holds.
Hence, charging speed only depends on the current~\gls*{soc}, but not on the arrival~\gls*{soc}.
These conditions are met by realistic physical models of charging stations~\cite{Mon15,Pel15,Uhr15}.
Moreover, exploiting the shifting property, it is possible to represent the~(bivariate) charging function~$\chargingfunction_\vertex$ using a univariate function~${\simplechargingfunction_\vertex\colon\posreals\to[0,\maxbattery]}$ with~$\simplechargingfunction_\vertex(\atime):=\chargingfunction_\vertex(0,\atime)$; see Figure~\ref{fig:charging-function} and our explanation further below.

\begin{figure}[t]
 \centering%
 \begin{tikzpicture}[figure,scale=\ChargingFunctionExampleScale]
  \ExampleDrawCoordinateSystemWithTicks{1}{8}{0}{6}{0}{0}
 
  % Axis names.
  \PlotXAxisName{8}{0}{$\atime$}
  \PlotYAxisName{0}{6}{$\simplechargingfunction_\vertex(\atime)$}
 
  % Arrows.
  \draw[<->] (0.2,0.2) -- (1.8,0.2) node [pos=0.68,above=2pt,edgelabel] {$\simplechargingfunctionInv_\vertex(3)$};
  \draw[<->] (2.2,0.2) -- (3.8,0.2) node [midway,above=2pt,edgelabel] {$\chargingtime$};
 
  % Lines.
  \draw[helperline] (0,3) -- node [midway,above,edgelabel] {$\soc^{\text{arr}}$} (2,3) -- (2,0);
  \draw[helperline] (0,5) -- node [pos=0.25,below=2pt,edgelabel] {$\soc^{\text{dep}}$} (4,5) -- (4,0);
 
  % Draw function.
  \begin{scope}
   \begin{pgfinterruptboundingbox}
    \clip (0,0) rectangle (8,6+\ExampleInftyOffset);
    \draw[chargingfunction] (0,0) -- (3,4.5) -- (5,5.5) -- (7,6) -- (8,6);
   \end{pgfinterruptboundingbox}
  \end{scope}
 \end{tikzpicture}%
 \caption{A univariate charging function $\simplechargingfunction_\vertex$ of a charging station $\vertex\in\chargingstations$ with minimum \gls*{soc} value $\mincharge_\vertex=0$ and maximum \gls*{soc} value~$\maxcharge_\vertex=6$. Reaching $\vertex$ with an arrival \gls*{soc} of $\arrivalsoc=3$ and spending a charging time of $\chargingtime=2$ yields a departure \gls*{soc} $\departuresoc=\chargingfunction_\vertex(\arrivalsoc,\chargingtime)$, which we obtain by evaluating $\simplechargingfunction_\vertex(\simplechargingfunctionInv_\vertex(\arrivalsoc)+\chargingtime)=5$.}%
\label{fig:charging-function}%
\end{figure}

Given a vertex $\vertex\in\chargingstations$ with a charging function $\chargingfunction_\vertex$ that has the above properties, we further presume that there is a finite value $\maxtime_\vertex\in\posreals$, such that $\chargingfunction_\vertex(\soc,\atime)=\chargingfunction_\vertex(\soc,\maxtime_\vertex)$ for all $\soc\in[0,\maxbattery]$ and~$\atime\ge\maxtime_\vertex$.
In other words, some maximum \gls*{soc} is reached after a finite charging time (the charging function does not converge to some~\gls*{soc} without reaching it eventually).
Then, the minimum \gls*{soc} value $\mincharge_\vertex:=\min_{\atime\in\posreals}\chargingfunction_\vertex(0,\atime)=\chargingfunction_\vertex(0,0)$ and the maximum \gls*{soc} value $\maxcharge_\vertex:=\max_{\atime\in\posreals}\chargingfunction_\vertex(0,\atime)=\chargingfunction_\vertex(0,\maxtime_\vertex)$ of a charging function induce a range $[\mincharge_\vertex,\maxcharge_\vertex]$ of possible \gls*{soc} values after charging at~$\vertex$.
We allow the cases $\mincharge_\vertex>0$ and $\maxcharge_\vertex<\maxbattery$ to model certain restrictions of charging stations. Thereby, our notion of charging functions is flexible enough to capture features of realistic charging stations. For example, we include swapping stations by setting $\chargingfunction_\vertex(\atime,\soc)=\maxbattery$ for all values $\atime\in\posreals$ and~$\soc\in[0,\maxbattery]$. Hence, we obtain~$\mincharge_\vertex=\maxcharge_\vertex=\maxbattery$.
Finally, we assign to every charging station~$\vertex\in\chargingstations$ a constant \emph{initialization time}~$\arrangementtime(\vertex)$ that is spent when charging at~$\vertex$. This enables us to model time overhead at a charging station for, \eg, parking the car or swapping the battery.

As mentioned above, we represent the bivariate charging function $\chargingfunction_\vertex$ of a vertex $\vertex\in\chargingstations$ with a univariate function $\simplechargingfunction_\vertex$, as follows.
Consider the \emph{inverse function} $\chargingfunction_\vertex^\inverse$ mapping a desired departure \gls*{soc} $\soc\in[\mincharge_\vertex,\maxcharge_\vertex)$ to the required charging time $\atime\in\posreals$ when the arrival \gls*{soc} is~$0$, \ie, $\chargingfunction_\vertex^\inverse(\soc)=\atime$ implies that~$\simplechargingfunction_\vertex(\atime)=\chargingfunction_\vertex(0,\atime)=\soc$. Since $\simplechargingfunction_\vertex$ is strictly increasing on the interval $[\mincharge_\vertex,\maxcharge_\vertex)$ by definition, the function $\chargingfunction_\vertex^\inverse$ is well-defined on the domain~$[\mincharge_\vertex,\maxcharge_\vertex)$.
Given the minimum charging time $\maxtime_\vertex\in\posreals$ required to charge to an \gls*{soc} $\maxcharge_\vertex$ at~$\vertex$ from an arrival \gls*{soc} of $0$, we define the \emph{expanded inverse function} $\simplechargingfunctionInv_\vertex\colon[0,\maxbattery]\to\posreals$ by setting
\begin{align*}
 \simplechargingfunctionInv_\vertex(\soc):=
 \begin{cases}
  0                                        & \mbox{if }\soc<\mincharge_\vertex\text{,}\\
  \maxtime_\vertex                         & \mbox{if }\soc\ge\maxcharge_\vertex\text{,}\\
  \chargingfunction_\vertex^\inverse(\soc) & \mbox{otherwise.}
 \end{cases}
\end{align*}
This yields the equivalence~$\chargingfunction_\vertex(\soc,\atime)=\simplechargingfunction_\vertex(\simplechargingfunctionInv_\vertex(\soc)+\atime)$ for arbitrary values $\soc\in[0,\maxbattery]$ with $\soc\le\maxcharge_\vertex$ and~$\atime\in\posreals$; see Figure~\ref{fig:charging-function}.
Further, we denote by~$\chargingfunction_\vertex^\inverse(\soc_1,\soc_2):=\simplechargingfunctionInv_\vertex(\soc_2)-\simplechargingfunctionInv_\vertex(\soc_1)$ the time to charge the battery from some arrival \gls*{soc} $\soc_1\in[0,\maxbattery]$ to a desired departure \gls*{soc}~$\soc_2\in[0,\maxbattery]$ with~$\soc_1\le\soc_2$.

Existing models of charging functions use linear, polynomial, and exponential functions, or piecewise combinations thereof~\cite{Swe14,Mon15,Pel15}. Typically, these functions are also \emph{concave} \wrt charging time (\ie, charging speed only decreases as the battery's \gls*{soc} increases). However, charging functions in our model are not limited to such functions per se. Section~\ref{sec:spawning} discusses necessary conditions for charging functions besides those mentioned above (continuity, monotonicity, and the shifting property) to ensure that our algorithms terminate.
For the sake of simplicity and motivated by data input in our experimental evaluation (see Section~\ref{sec:experiments}), examples in subsequent sections use piecewise linear, concave charging functions.

\paragraph{Problem Statement.}
We consider the following objective: For a given source~$\source\in\vertices$, a target~$\target\in\vertices$, and an initial \gls*{soc}~$\soc_\source\in[0,\maxbattery]$, we want to find a feasible $\source$--$\target$~path that minimizes overall \emph{trip time}, \ie, the sum of driving time and total time spent at charging stations.
Note that if the input graph contains no charging stations ($\chargingstations=\emptyset$) and consumption values are nonnegative for all arcs, we have an instance of~\gls*{csp}, hence the considered problem is \NP-hard, too.

%%%%%%%%%%%%%%%%%%%%%%%%%%%%%%%%%%%%%%%%%%%%%%%%%%%%%%%%%%%%%%%%%%%%%%%%%%%%%%%%
\section{Basic Approach}\label{sec:approach}
%%%%%%%%%%%%%%%%%%%%%%%%%%%%%%%%%%%%%%%%%%%%%%%%%%%%%%%%%%%%%%%%%%%%%%%%%%%%%%%%

Adapting the bicriteria algorithm described in Section~\ref{sec:problem} to our setting is difficult, for several reasons: When reaching a charging station, we do not know how much energy should be recharged, since it depends on the remaining route to the target and the charging stations available on this route. To overcome this issue, a key idea of our algorithm is to \emph{delay} this decision until we reach the target or the next charging station.
However, since charging functions are \emph{continuous}, there is no straightforward way to apply the bicriteria algorithm in this case, as it might require an infinite number of nondominated labels after settling a charging station with a continuous charging function.
In this section, we show how the algorithm can be to generalized to our setting.
The \emph{\gls*{cfp}} algorithm extends labels to maintain infinite, continuous sets of solutions.
The core idea is that a label represents all possible tradeoffs between charging time and resulting \gls*{soc} induced by the last visited charging station (if it exists), but still has constant size.

Before we describe the \gls*{cfp} algorithm in detail, we illustrate its behavior in a simple example based on charging functions that are piecewise linear and concave.
Afterwards, we show how to represent paths containing charging stations with labels of \emph{constant} size and describe the \gls*{cfp} algorithm more formally. In Section~\ref{sec:spawning}, we discuss the implementation of its crucial part, namely, generating new labels at charging stations (for general charging functions).

\begin{figure}[t]
 \centering
 \begin{tikzpicture}[figure,scale=1.7]
 \node[textvertex] (s) at (0,0) {$\source$};
 \node[textvertex] (v1) at (1,0) {$\vertexa$};
 \node[textvertex] (v2) at (2,0) {$\vertexb$};
 \node[textvertex] (v3) at (3,0) {$\vertexc$};
 \node[textvertex] (v4) at (4,0) {$\vertexd$};
 \node[textvertex] (v5) at (5,0) {$\vertexe$};
 \node[textvertex] (v6) at (6,0) {$\vertexf$};
 \node[textvertex] (t) at (7,0) {$\target$};

 \node (c1) at (2.05,0.3) {\includegraphics[scale=0.5]{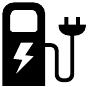}};
 \node (c1) at (5.05,0.3) {\includegraphics[scale=0.5]{fig/charging-station}};

 \path[edges,->] (s) edge node[below] {$-2$} (v1)
                 (v1) edge node[below] {$4$} (v2)
                 (v2) edge node[below] {$-1$} (v3)
                 (v3) edge node[below] {$3$} (v4)
                 (v4) edge node[below] {$-1$} (v5)
                 (v5) edge node[below] {$2$} (v6)
                 (v6) edge node[below] {$1$} (t);

 % Charging functions.
 \begin{scope}[xshift=0.9cm,yshift=0.8cm,scale=0.35]
  \ExampleDrawCoordinateSystemWithTicks{1}{6}{0}{5}{0}{0}
  \PlotXAxisName{6}{0}{$\atime$}
  \PlotYAxisName{0}{5}{$\simplechargingfunction_\vertexb(\atime)$}
  \begin{scope}
   \begin{pgfinterruptboundingbox}
    \clip (0,0) rectangle (6,7);
    \draw[chargingfunction] (0,0) -- (1.5,3) -- (5.5,5) -- (6,5);
   \end{pgfinterruptboundingbox}
  \end{scope}
 \end{scope}
 
 \begin{scope}[xshift=3.9cm,yshift=0.8cm,scale=0.35]
  \ExampleDrawCoordinateSystemWithTicks{1}{6}{0}{5}{0}{0}
  \PlotXAxisName{6}{0}{$\atime$}
  \PlotYAxisName{0}{5}{$\simplechargingfunction_\vertexe(\atime)$}
  \begin{scope}
   \begin{pgfinterruptboundingbox}
    \clip (0,0) rectangle (6,7);
    \draw[chargingfunction] (0,0) -- (5,5) -- (6,5);
   \end{pgfinterruptboundingbox}
  \end{scope}
 \end{scope}
 
 \end{tikzpicture}
 \caption{Example graph with indicated charging stations $\vertexb$ and $\vertexe$ and their respective charging functions $\simplechargingfunction_\vertexb$ and~$\simplechargingfunction_\vertexe$. The initialization time is 0 for both stations. Arc labels correspond to their energy consumption, while we presume that the driving time on every arc is~$1$. On the (unique) path from $\source$ to~$\target$, we seek charging times at $\vertexb$ and $\vertexe$ that minimize the overall travel time (assuming $\soc_\source=4$ and~$\maxbattery=5$.)}
 \label{fig:cfp-example-path}
\end{figure}

\paragraph{Running \gls*{cfp} on an Example Graph.}
To get a basic understanding of the challenges described above and how they are handled by our algorithm, consider the simple example graph depicted in Figure~\ref{fig:cfp-example-path}.
(We note that neither consumption values nor charging functions in our examples were chosen to be very realistic, but rather to illustrate different cases that may occur both in theory and in practice.)
There is only one path from the source vertex $\source$ to the target vertex~$\target$, but it contains two charging stations. Hence, we essentially have to determine the charging time at each station that minimizes the overall trip time. Let the initial \gls*{soc} at $\source$ be $\soc_\source=4$ and assume that the battery capacity is~$\maxbattery=5$.

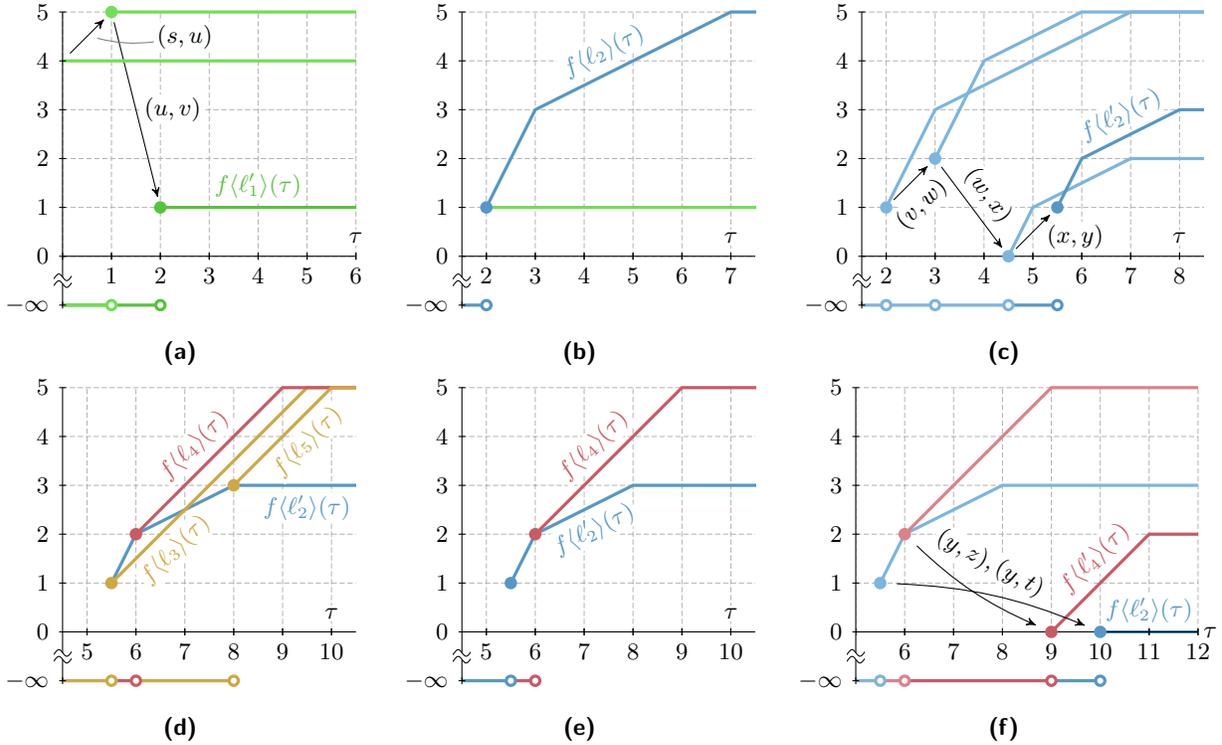
\begin{figure}[p]
 \centering
 \begin{subfigure}[b]{.32\textwidth}%
  \centering%
  \colorlet{primaryfunctioncolor}{thesisgreen}
  \colorlet{secondaryfunctioncolor}{thesisgreen-light}
  \colorlet{backgroundfunctioncolor}{black30}
  \begin{tikzpicture}[figure,scale=\ChargingFunctionExampleScale]
  \ExampleDrawCoordinateSystemWithNegativeInfty{1}{6}{0}{5}{0}{0}
  \PlotXAxisName{6}{0}{$\atime$}
  \draw[descrarrow] (0,4) -- node (arrmid) {} (1,5);
  \draw[descrarrow] (1,5) -- node[edgelabel,right=2pt] {$(\vertexa,\vertexb)$} (2,1);
  \node[edgelabel] (desc) at (2.5,4.5) {$(\source,\vertexa)$};
  \path[draw,black50,bend right=10] (arrmid) edge (desc);
  \begin{scope}
   \begin{pgfinterruptboundingbox}
   \clip (0,-\ExampleInftyOffset-1) rectangle (6,6);
   \draw[function,primaryfunctioncolor] (0,-\ExampleInftyOffset) -- (2,-\ExampleInftyOffset);
   \draw[function,secondaryfunctioncolor] (0,-\ExampleInftyOffset) -- (1,-\ExampleInftyOffset);
   \draw[function,secondaryfunctioncolor] (0,4) -- (6,4);
   \draw[function,secondaryfunctioncolor] (1,5) -- (6,5);
   \draw[function,primaryfunctioncolor] (2,1) -- node[functionlabel] {$\socfunction{\alabel'_1}(\atime)$} (6,1);
   \end{pgfinterruptboundingbox}
  \end{scope}
  \begin{pgfonlayer}{foreground}
   \node[discontinuityblank,draw=secondaryfunctioncolor] (d1) at (1,-\ExampleInftyOffset) {};
   \node[discontinuityblank,draw=primaryfunctioncolor] (d2) at (2,-\ExampleInftyOffset) {};
   \node[discontinuityfilled,secondaryfunctioncolor] (d3) at (1,5) {};
   \node[discontinuityfilled,primaryfunctioncolor] (d4) at (2,1) {};
  \end{pgfonlayer}
 \end{tikzpicture}
  \caption{}%
  \label{fig:cfp-example-path-labels:a}%
 \end{subfigure}%
 \begin{subfigure}[b]{.32\textwidth}%
  \centering%
  \colorlet{primaryfunctioncolor}{thesisblue}
  \colorlet{secondaryfunctioncolor}{thesisgreen-light}
  \colorlet{backgroundfunctioncolor}{thesisgreen-light}
  \begin{tikzpicture}[figure,scale=\ChargingFunctionExampleScale]
  \ExampleDrawCoordinateSystemWithNegativeInfty{2}{7.5}{0}{5}{1.5}{0}
  \PlotXAxisName{7}{0}{$\atime$}
  \begin{scope}
   \begin{pgfinterruptboundingbox}
    \clip (1.5,-\ExampleInftyOffset-1) rectangle (7.5,6);
    \draw[function,primaryfunctioncolor] (0,-\ExampleInftyOffset) -- (2,-\ExampleInftyOffset);
    \draw[function,draw=backgroundfunctioncolor] (2,1) -- (10,1);
    \draw[function,primaryfunctioncolor] (2,1) -- (3,3) -- node [functionlabel,pos=0.4] {$\socfunction{\alabel_2}(\atime)$} (7,5) -- (10,5);
   \end{pgfinterruptboundingbox}
  \end{scope}
  \begin{pgfonlayer}{foreground}
   \node[discontinuityblank,draw=primaryfunctioncolor] (d1) at (2,-\ExampleInftyOffset) {};
   \node[discontinuityfilled,primaryfunctioncolor] (d2) at (2,1) {};
  \end{pgfonlayer}
 \end{tikzpicture} 
 \caption{}%
 \label{fig:cfp-example-path-labels:b}%
 \end{subfigure}%
 \begin{subfigure}[b]{.36\textwidth}%
  \centering%
  \colorlet{primaryfunctioncolor}{thesisblue}
  \colorlet{secondaryfunctioncolor}{thesisblue-light}
  \colorlet{backgroundfunctioncolor}{black30}
  \begin{tikzpicture}[figure,scale=\ChargingFunctionExampleScale]
  \ExampleDrawCoordinateSystemWithNegativeInfty{2}{8.5}{0}{5}{1.5}{0}
  \PlotXAxisName{8}{0}{$\atime$}
  \draw[descrarrow] (2,1) -- node[edgelabel,sloped,pos=0.4,below=3pt] {$(\vertexb,\vertexc)$} (3,2);
  \draw[descrarrow] (3,2) -- node[edgelabel,sloped,above=2pt] {$(\vertexc,\vertexd)$} (4.5,0);
  \draw[descrarrow] (4.5,0) -- node[edgelabel,pos=0.4,right=6pt] {$(\vertexd,\vertexe)$} (5.5,1);
  \begin{scope}
   \begin{pgfinterruptboundingbox}
    \clip (1.5,-\ExampleInftyOffset-1) rectangle (8.5,6);
    \draw[function,secondaryfunctioncolor] (0,-\ExampleInftyOffset) -- (4.5,-\ExampleInftyOffset);
    \draw[function,primaryfunctioncolor] (4.5,-\ExampleInftyOffset) -- (5.5,-\ExampleInftyOffset);
    \draw[function,secondaryfunctioncolor] (2,1) -- (3,3) -- (7,5) -- (10,5);
    \draw[function,secondaryfunctioncolor] (3,2) -- (4,4) -- (6,5) -- (10,5);
    \draw[function,secondaryfunctioncolor] (4.5,0) -- (5,1) -- (7,2) -- (10,2);
    \draw[function,primaryfunctioncolor] (5.5,1) -- (6,2) -- node[functionlabel] {$\socfunction{\alabel'_2}(\atime)$} (8,3) -- (10,3);
   \end{pgfinterruptboundingbox}
  \end{scope}
  \begin{pgfonlayer}{foreground}
   \node[discontinuityblank,draw=secondaryfunctioncolor] (d1) at (2,-\ExampleInftyOffset) {};
   \node[discontinuityfilled,secondaryfunctioncolor] (d2) at (2,1) {};
   \node[discontinuityblank,draw=secondaryfunctioncolor] (d3) at (3,-\ExampleInftyOffset) {};
   \node[discontinuityfilled,secondaryfunctioncolor] (d4) at (3,2) {};
   \node[discontinuityblank,draw=secondaryfunctioncolor] (d5) at (4.5,-\ExampleInftyOffset) {};
   \node[discontinuityfilled,secondaryfunctioncolor] (d6) at (4.5,0) {};
   \node[discontinuityblank,draw=primaryfunctioncolor] (d7) at (5.5,-\ExampleInftyOffset) {};
   \node[discontinuityfilled,primaryfunctioncolor] (d8) at (5.5,1) {};
  \end{pgfonlayer}
 \end{tikzpicture}
 \caption{}%
 \label{fig:cfp-example-path-labels:c}%
 \end{subfigure}%
 
 \begin{subfigure}[b]{.32\textwidth}%
  \centering%
  \colorlet{primaryfunctioncolor}{thesisred}
  \colorlet{secondaryfunctioncolor}{thesisyellow}
  \colorlet{backgroundfunctioncolor}{thesisblue}
  \begin{tikzpicture}[figure,scale=\ChargingFunctionExampleScale]
  \ExampleDrawCoordinateSystemWithNegativeInfty{5}{10.5}{0}{5}{4.5}{0}
  \PlotXAxisName{10}{0}{$\atime$}
  \begin{scope}
   \begin{pgfinterruptboundingbox}
    \clip (4.5,-\ExampleInftyOffset-1) rectangle (10.5,6);
    \draw[function,secondaryfunctioncolor] (0,-\ExampleInftyOffset) -- (5.5,-\ExampleInftyOffset);
    \draw[function,primaryfunctioncolor] (5.5,-\ExampleInftyOffset) -- (6,-\ExampleInftyOffset);
    \draw[function,secondaryfunctioncolor] (6,-\ExampleInftyOffset) -- (8,-\ExampleInftyOffset);
    \draw[function,backgroundfunctioncolor] (5.5,1) -- (6,2) -- (8,3) -- node[functionlabel,below=4pt] {$\socfunction{\alabel'_2}(\atime)$} (11,3);
    \draw[function,secondaryfunctioncolor] (5.5,1) -- node[functionlabel,pos=0.25,below=4pt] {$\socfunction{\alabel_3}(\atime)$} (9.5,5) -- (11,5);
    \draw[function,primaryfunctioncolor] (6,2) -- node[functionlabel] {$\socfunction{\alabel_4}(\atime)$} (9,5) -- (11,5);
    \draw[function,secondaryfunctioncolor] (8,3) -- node[functionlabel,pos=0.6,below=4pt] {$\socfunction{\alabel_5}(\atime)$} (10,5) -- (11,5);
   \end{pgfinterruptboundingbox}
  \end{scope}
  \begin{pgfonlayer}{foreground}
   \node[discontinuityblank,draw=secondaryfunctioncolor] (d1) at (5.5,-\ExampleInftyOffset) {};
   \node[discontinuityfilled,secondaryfunctioncolor] (d2) at (5.5,1) {};
   \node[discontinuityblank,draw=primaryfunctioncolor] (d3) at (6,-\ExampleInftyOffset) {};
   \node[discontinuityfilled,primaryfunctioncolor] (d4) at (6,2) {};
   \node[discontinuityblank,draw=secondaryfunctioncolor] (d5) at (8,-\ExampleInftyOffset) {};
   \node[discontinuityfilled,secondaryfunctioncolor] (d6) at (8,3) {};
  \end{pgfonlayer}
  \end{tikzpicture} 
 \caption{}%
 \label{fig:cfp-example-path-labels:d}%
 \end{subfigure}%
 \begin{subfigure}[b]{.32\textwidth}%
  \centering%
  \colorlet{primaryfunctioncolor}{thesisred}
  \colorlet{secondaryfunctioncolor}{thesisblue}
  \colorlet{backgroundfunctioncolor}{black30}
  \begin{tikzpicture}[figure,scale=\ChargingFunctionExampleScale]
  \ExampleDrawCoordinateSystemWithNegativeInfty{5}{10.5}{0}{5}{4.5}{0}
  \PlotXAxisName{10}{0}{$\atime$}
  \begin{scope}
   \begin{pgfinterruptboundingbox}
    \clip (4.5,-\ExampleInftyOffset-1) rectangle (10.5,6);
    \draw[function,secondaryfunctioncolor] (0,-\ExampleInftyOffset) -- (5.5,-\ExampleInftyOffset);
    \draw[function,primaryfunctioncolor] (5.5,-\ExampleInftyOffset) -- (6,-\ExampleInftyOffset);
    \draw[function,secondaryfunctioncolor] (5.5,1) -- (6,2) -- node[functionlabel,below=4pt] {$\socfunction{\alabel'_2}(\atime)$} (8,3) -- (11,3);
    \draw[function,primaryfunctioncolor] (6,2) -- node[functionlabel] {$\socfunction{\alabel_4}(\atime)$} (9,5) -- (11,5);
   \end{pgfinterruptboundingbox}
  \end{scope}
  \begin{pgfonlayer}{foreground}
   \node[discontinuityblank,draw=secondaryfunctioncolor] (d1) at (5.5,-\ExampleInftyOffset) {};
   \node[discontinuityfilled,secondaryfunctioncolor] (d2) at (5.5,1) {};
   \node[discontinuityblank,draw=primaryfunctioncolor] (d3) at (6,-\ExampleInftyOffset) {};
   \node[discontinuityfilled,primaryfunctioncolor] (d4) at (6,2) {};
  \end{pgfonlayer}
  \end{tikzpicture}
 \caption{}%
 \label{fig:cfp-example-path-labels:e}%
 \end{subfigure}%
 \begin{subfigure}[b]{.36\textwidth}%
  \centering%
  \colorlet{primaryfunctioncolorA}{thesisred-light}
  \colorlet{primaryfunctioncolorB}{thesisred}
  \colorlet{secondaryfunctioncolorA}{thesisblue-light}
  \colorlet{secondaryfunctioncolorB}{thesisblue}
  \colorlet{backgroundfunctioncolor}{black30}
  \begin{tikzpicture}[figure,scale=\ChargingFunctionExampleScale]
  \ExampleDrawCoordinateSystemWithNegativeInfty{6}{12}{0}{5}{5}{0}
  \PlotXAxisName{12.25}{-0.4}{$\atime$}
  \begin{scope}
   \begin{pgfinterruptboundingbox}
    \clip (5,-\ExampleInftyOffset-1) rectangle (12,6);
    \draw[function,secondaryfunctioncolorA] (0,-\ExampleInftyOffset) -- (5.5,-\ExampleInftyOffset);
    \draw[function,primaryfunctioncolorA] (5.5,-\ExampleInftyOffset) -- (6,-\ExampleInftyOffset);
    \draw[function,primaryfunctioncolorB] (6,-\ExampleInftyOffset) -- (9,-\ExampleInftyOffset);
    \draw[function,secondaryfunctioncolorB] (9,-\ExampleInftyOffset) -- (10,-\ExampleInftyOffset);
    \draw[function,secondaryfunctioncolorA] (5.5,1) -- (6,2) -- (8,3) -- (12,3);
    \draw[function,primaryfunctioncolorA] (6,2) -- (9,5) -- (12,5);
    \draw[function,secondaryfunctioncolorB] (10,0) -- node[functionlabel] {$\socfunction{\alabel'_2}(\atime)$} (12,0);
    \draw[function,primaryfunctioncolorB] (9,0) -- node[functionlabel,pos=0.6] {$\socfunction{\alabel'_4}(\atime)$} (11,2) -- (12,2);
   \end{pgfinterruptboundingbox}
  \end{scope}
  \begin{pgfonlayer}{foreground}
   \node[discontinuityblank,draw=secondaryfunctioncolorA] (d1) at (5.5,-\ExampleInftyOffset) {};
   \node[discontinuityfilled,secondaryfunctioncolorA] (d2) at (5.5,1) {};
   \node[discontinuityblank,draw=primaryfunctioncolorA] (d3) at (6,-\ExampleInftyOffset) {};
   \node[discontinuityfilled,primaryfunctioncolorA] (d4) at (6,2) {};
   \node[discontinuityblank,draw=primaryfunctioncolorB] (d5) at (9,-\ExampleInftyOffset) {};
   \node[discontinuityfilled,primaryfunctioncolorB] (d6) at (9,0) {};
   \node[discontinuityblank,draw=secondaryfunctioncolorB] (d7) at (10,-\ExampleInftyOffset) {};
   \node[discontinuityfilled,secondaryfunctioncolorB] (d8) at (10,0) {};
  \end{pgfonlayer}
  \path[draw,descrarrow,bend right=10] (d4) edge (d6);
  \path[draw,descrarrow,bend left=10] (d2) edge (d8);
  \node[edgelabel,rotate=-28] (descr) at (7.75,1.4) {$(\vertexe,\vertexf),(\vertexe,\target)$};
  \end{tikzpicture}
  \caption{}%
  \label{fig:cfp-example-path-labels:f}%
 \end{subfigure}%
 \caption{Labels generated by the \gls*{cfp} algorithm on the graph shown in Figure~\ref{fig:cfp-example-path}. (a)~The functions (green) created at the vertices~$\source$,~$\vertexa$, and~$\vertexb$ (upon arrival). (b)~A new function (blue) corresponding to a label spawned at $\vertexb$ to model departure \gls*{soc} after charging at the station. (c)~The label is propagated to~$\vertexc$,~$\vertexd$, and~$\vertexe$. Note that certain function values are ``cropped'', due to battery constraints. (d)~New functions (red, yellow) spawned at the charging station~$\vertexe$, two of which are dominated (yellow). (e)~The resulting nondominated functions upon departure at~$\vertexe$. Note that in one of them, no energy is charged at~$\vertexe$. (f)~The functions are propagated to the target~$\target$. The minimum feasible trip time is 9 and requires charging at both stations: two units (departure \gls*{soc}~3) at $\vertexb$ and one unit (departure \gls*{soc}~3) at~$\vertexe$. }
 \label{fig:cfp-example-path-labels}
\end{figure}

Our algorithm propagates labels which represent continuous functions mapping the trip time to the resulting \gls*{soc} at some vertex, depending on the amount of energy charged at the previous station. Figure~\ref{fig:cfp-example-path-labels} shows these functions as they are generated by our algorithm when running on the instance from Figure~\ref{fig:cfp-example-path}. The algorithm is initialized with a label $\alabel_1$ representing a function denoted $\socfunction{\alabel_1}$ at the source vertex~$\source$; see Figure~\ref{fig:cfp-example-path-labels:a}. The function $\socfunction{\alabel_1}$ evaluates to the initial \gls*{soc} 4 for arbitrary trip times (since $\source$ is not a charging station, spending time at $\source$ cannot increase the~\gls*{soc}).

To propagate this function through the graph, the (single) outgoing arc $(\source,\vertexa)$ with driving time 1 and energy consumption -2 is scanned. Thus, the vertex $\vertexa$ is reached after a trip time of~1, which we model by setting the function of the (propagated) label to $-\infty$ for trip times below~1. Moreover, the \gls*{soc} increases to the maximum \gls*{soc}~5 (note battery constraints prevent the \gls*{soc} from exceeding this limit). Next, this label is used to scan the arc $(\vertexa,\vertexb)$ with trip time 1 and energy consumption~4, which results in a label $\alabel'_1$ at $\vertexb$ containing the function $\socfunction{\alabel'_1}$ shown in Figure~\ref{fig:cfp-example-path-labels:a}. It evaluates to an \gls*{soc} of 1 if the trip time is at least~2.

According to the label~$\alabel'_1$, we reach the first charging station $\vertexb$ with an \gls*{soc} of~1 and after a trip time of~2. Since $\vertexb$ is a charging station, we want to incorporate the possibility of charging into the label. A function that reflects tradeoffs between trip time after charging and the resulting \gls*{soc} is obtained after shifting the charging function $\simplechargingfunction_\vertexb$ along the x-axis such that it intersects the point with x-coordinate (\ie, trip time)~2 and y-coordinate (\ie, arrival~\gls*{soc})~1.
This results in a new label $\alabel_2$ which \emph{dominates} the old one, \ie, for each trip time it yields greater or equal \gls*{soc} compared to~$\alabel_1$; see Figure~\ref{fig:cfp-example-path-labels:b} (however, this is only the case because we presume that the initialization time is~0, whereas in general we might obtain labels that do not dominate each other).
The nondominated label $\alabel_2$ is then propagated to the vertices~$\vertexc$,~$\vertexd$, and~$\vertexe$. Again, we apply battery constraints, which in hindsight renders some charging times at $\vertexb$ infeasible (not enough energy is charged to traverse the following path) or unprofitable (charging too much energy wastes energy gained through recuperation). We ``crop'' the image of the resulting function $\socfunction{\alabel'_2}$ by setting corresponding function values to $-\infty$ or~$\maxbattery=5$, respectively; see Figure~\ref{fig:cfp-example-path-labels:c}.

The search reaches the next charging station at the vertex~$\vertexe$. Again, we may shift the function $\simplechargingfunction_\vertexe$ until its plot intersects a (nondominated) pair of trip time and \gls*{soc} of the label~$\alabel'_2$. However, unlike before, there no longer is a unique \emph{nondominated} pair. As before, we may pick the least feasible trip time~(5.5 in this case), which results in the label~$\alabel_3$. Observe that this yields a function $\socfunction{\alabel_3}$ with lower \gls*{soc} compared to the function $\socfunction{\alabel'_2}$ for some values of trip time; see Figure~\ref{fig:cfp-example-path-labels:d}. This is due to the fact that, for low \gls*{soc} values, the charging station at $\vertexb$ allows faster charging. Therefore, it pays off to charge more energy at $\vertexb$, up to the next breakpoint of the function~$\socfunction{\alabel'_2}$ (where its the slope falls below that of the charging function at~$\vertexe$). Spawning the label $\alabel_4$ after picking this breakpoint results in a nondominated function (see Figure~\ref{fig:cfp-example-path-labels:d}). Apparently, breakpoints of functions are good candidates for spawning new labels, since the charging rate changes in these points. Consequently, our algorithm spawns one new label for each breakpoint of the function $\socfunction{\alabel'_2}$. In Section~\ref{sec:spawning}, we show that this is actually sufficient to find an optimal solution. Two of these labels ($\alabel_3$ and~$\alabel_5$) are dominated, though, so they can be discarded. (Our algorithm actually only performs \emph{pairwise} dominance checks to reduce overhead, which would spare the label $\alabel_3$ as it is only dominated by the upper envelope of $\socfunction{\alabel'_2}$ and~$\socfunction{\alabel_4}$; this increases the number of propagated labels, but does not affect correctness of the approach).

The nondominated labels $\socfunction{\alabel_2}$ and~$\socfunction{\alabel_4}$ are shown in Figure~\ref{fig:cfp-example-path-labels:e}. They are both propagated to $\vertexf$ and finally, to the target~$\target$. The resulting labels are shown in Figure~\ref{fig:cfp-example-path-labels:f}. Note that both contain feasible pairs of trip time and \gls*{soc}. The least feasible trip time among both labels is~9. Note that this trip time requires charging at both stations, despite the fact that it is possible to reach the target with a single charging stop at~$\vertexb$. This is due to the fact that the charging speed at both station differs. At the same time, both stations lie on the (unique) route to the target and there is no initialization overhead, so switching between different charging stations comes at a low cost in this example.
In what follows, we formally define labels and their corresponding functions, before we describe the \gls*{cfp} algorithm in more detail.

\begin{figure}[t]
 \centering
 \begin{subfigure}[b]{.32\textwidth}%
  \centering%
  \begin{tikzpicture}[figure,scale=\ChargingFunctionExampleScale]
   \ExampleDrawCoordinateSystemWithTicks{1}{5.5}{0}{4}{0}{0}
 
   % Bounding box stuff.
   \node (dummy) at (0,-\ExampleInftyOffset) {\phantom{$-\infty$}};
 
   % Axis names.
   \PlotXAxisName{5}{0}{$\atime$}
   \PlotYAxisName{0}{4}{$\simplechargingfunction_\vertexa(\atime)$}
  
   % Draw function.
   \begin{scope}
    \begin{pgfinterruptboundingbox}
     \clip (0,0) rectangle (5.5,4+\ExampleInftyOffset);
     \draw[chargingfunction] (0,0) -- (1,2) -- (2,3) -- (5,4) -- (5.5,4);
    \end{pgfinterruptboundingbox}
   \end{scope}
  \end{tikzpicture}
  \caption{}%
  \label{fig:cfp-example-label:charging}%
 \end{subfigure}%
 \begin{subfigure}[b]{.3\textwidth}%
  \centering%
  \begin{tikzpicture}[figure,scale=\ChargingFunctionExampleScale]
   \SoCProfileExampleDrawCoordinateSystem{4.0}{4.0}
 
   % Line with arrow.
   \begin{scope}
    \begin{pgfinterruptboundingbox}
     \clip (0,0) rectangle (4,5);
     \draw[color=black30,line cap=rect] (0,0) -- (4,4);
     \draw[->] (1,1.2) -- (1,1.8);
    \end{pgfinterruptboundingbox}
   \end{scope}
 
   \SoCProfileExampleDrawFunction{4.0}{1.0}{2.0}{2.0}{3.0}
  
   % Axis names.
   \PlotXAxisName{4}{0}{$\soc$}
   \PlotYAxisName{0}{4}{$\socprofile_{[\vertexa,\dots,\vertexb]}(\soc)$}
  \end{tikzpicture}
  \caption{}%
  \label{fig:cfp-example-label:profile}%
 \end{subfigure}%
 \begin{subfigure}[b]{.38\textwidth}%
  \centering%
  \begin{tikzpicture}[figure,scale=\ChargingFunctionExampleScale]
   \ExampleDrawCoordinateSystemWithNegativeInfty{3}{8}{0}{4}{2}{0}
 
   % Axis names.
   \PlotXAxisName{8}{0}{$\atime$}
   %\PlotYAxisName{2}{4}{$\socfunction{\alabel}(\atime)$}
  
   % Draw function.
   \begin{scope}
    \begin{pgfinterruptboundingbox}
     \clip (2,-\ExampleInftyOffset-1) rectangle (8,5);
     \draw[chargingfunction,secondarycolor-light] (2.75,0) -- (3.75,2) -- (4.75,3) -- node [midway,above,edgelabel,sloped] {$\chargingfunction_\vertexa(\soc_\vertexa,\atime-\triptime)$} (7.75,4) -- (8,4);
     \draw[function] (0,-\ExampleInftyOffset) -- (3.25,-\ExampleInftyOffset);
     \draw[function] (3.25,2) -- (3.75,3) -- node [pos=0.65,below=2pt,edgelabel,sloped,text=primarycolor] {$\socfunction{\alabel}(\atime)$} (8,3);
    \end{pgfinterruptboundingbox}
   \end{scope}
 
   % Arrow.
   \draw[->] (3.25,1.2) -- (3.25,1.8);
 
   % Markers at discontinuities.
   \begin{pgfonlayer}{foreground}
    \node[discontinuityblank] (d1) at (3.25,-\ExampleInftyOffset) {};
    \node[discontinuityfilled] (d2) at (3.25,2) {};
   \end{pgfonlayer}
  \end{tikzpicture}
  \caption{}%
  \label{fig:cfp-example-label:label}%
 \end{subfigure}%
 \caption{Constructing the \gls*{soc} function~$\socfunction{\alabel}$ of a label~$\alabel=(\triptime,\soc_\vertexa,\vertexa,\socprofile_{[\vertexa,\dots,\vertexb]})$, with $\triptime=3$ and~$\soc_\vertexa=0.5$. (a)~The function~$\simplechargingfunction_\vertexa$. Assume that the initialization time at $\vertexa$ is~$\arrangementtime(\vertexa)=0$. (b)~The \gls*{soc} profile~$\socprofile_{[\vertexa,\dots,\vertexb]}$ of the~$\vertexa$--$\vertexb$~subpath. Note that the path has negative consumption (the \gls*{soc} increases as indicated by the arrow). (c)~The \gls*{soc} function $\socfunction{\alabel}$. The function $\chargingfunction_\vertexa(\soc_\vertexa,\atime-\triptime)$ (red) reflects pairs of trip time and \gls*{soc} when charging at~$\vertexa$, but ignores consumption on the~$\vertexa\text{--}\vertexb$~subpath. It is equivalent to the function obtained after shifting $\simplechargingfunction_\vertexa$ to the right by $\triptime=3$ minus~$\chargingfunction_\vertexa^{\inverse}(0,\soc_\vertexa)=0.25$. We apply battery constraints \wrt the~$\vertexa$--$\vertexb$~path to this function and obtain the depicted \gls*{soc} function~$\socfunction{\alabel}$ (blue). Its minimum feasible trip time is~$\minfeasibletime{\alabel}=3.25$, because we must spend a charging time of at least $0.25$ at~$\vertexa$. Moreover, we obtain $\socfunction{\alabel}(\atime)=3$ for~$\atime\ge 3.75$ (charging beyond an \gls*{soc} of $2$ at $\vertexa$ never pays off, as it wastes energy gains from recuperation).}
 \label{fig:cfp-example-label}
\end{figure}

\paragraph{Labels and SoC Functions.}
More generally, assume we are given a path $\apath$ from the source $\source\in\vertices$ to some vertex~$\vertexb\in\vertices$, such that $\apath$ contains a charging station $\vertexa\in\chargingstations$ and the arrival \gls*{soc} at $\vertexa$ is~$\soc_\vertexa$.
Every possible charging time $\chargingtime\in[0,\chargingfunction_\vertexa^\inverse(\soc_\vertexa,\maxcharge_\vertexa)]$ at $\vertexa$ results in a certain trip time and an \gls*{soc} at~$\vertexb$.
In general, this yields an infinite amount of feasible, nondominated pairs of trip time and corresponding \gls*{soc} for the path.
We implicitly represent these pairs in one label by storing the charging station $\vertexa$ in the label. However, this no longer allows us to apply battery constraints on-the-fly: For vertices visited after~$\vertexa$, labels have no fixed \gls*{soc}, as it depends on how much energy is charged at~$\vertexa$.
Hence, we compute the \gls*{soc} profile $\socprofile_{[\vertexa,\dots,\vertexb]}$ of the subpath from~$\vertexa$ to~$\vertexb$; see Section~\ref{sec:problem}. The label~$\alabel=(\triptime,\soc_\vertexa,\vertexa,\socprofile_{[\vertexa,\dots,\vertexb]})$ at the vertex~$\vertexb$ then consists of the trip time~$\triptime$ of the path from~$\source$ to~$\vertexb$ (including charging time on every previous charging station except~$\vertexa$ on the path from~$\source$ to~$\vertexa$), the \gls*{soc}~$\soc_\vertexa$ when reaching~$\vertexa$, the last visited charging station~$\vertexa$, and the \gls*{soc} profile~$\socprofile_{[\vertexa,\dots,\vertexb]}$ of the subpath from~$\vertexa$ to~$\vertexb$.
Recall that this \gls*{soc} profile can be represented by three values; see Section~\ref{sec:problem}.
Consequently, even though charging functions can have arbitrary descriptive complexity, we propagate them using labels of constant size.
The trip time $\triptime$ excludes charging at~$\vertexa$, but includes its initialization time $\arrangementtime(\vertexa)$.
Thus, we can think of $\triptime$ as the least trip time to reach $\vertexb$ if we stop at~$\vertexa$ (and ignore battery constraints on the $\vertexa$--$\vertexb$~path).

Accordingly, we define the \emph{\gls*{soc} function} $\socfunction{\alabel}$ of a label $\alabel$ to represent all feasible pairs of trip time and \gls*{soc} associated with the label~$\alabel=(\triptime,\soc_\vertexa,\vertexa,\socprofile_{[\vertexa,\dots, \vertexb]})$.
The \gls*{soc} function~$\socfunction{\alabel}\colon\posreals\to[0,\maxbattery]\cup\{-\infty\}$ mapping trip time to \gls*{soc} is given as
\begin{align}
\label{eq:charge-socfunction}
\socfunction{\alabel}(\atime):=
\begin{cases}
 \socprofile_{[\vertexa,\dots,\vertexb]}(\chargingfunction_\vertexa(\soc_\vertexa, \atime-\triptime)) & \mbox{if } \atime\ge\triptime\text{,} \\
 -\infty                                                                                              & \mbox{otherwise.}
\end{cases} 
\end{align}
To obtain the value~$\socfunction{\alabel}(\atime)$, \ie, the (arrival) \gls*{soc} at~$\vertexb$ when allowing a trip time~$\atime\ge\triptime$, Equation~\ref{eq:charge-socfunction} first evaluates the \gls*{soc} $\chargingfunction_\vertexa(\soc_\vertexa,\atime-\triptime)$ after charging at~$\vertexa$ for a total time of $\atime-\triptime$ with an arrival \gls*{soc} of~$\soc_\vertexa$.
Afterwards, the \gls*{soc} profile~$\socprofile_{[\vertexa,\dots,\vertexb]}$ is applied, which takes account of energy consumption (respecting battery constraints) on the path from the charging station~$\vertexa$ to the current vertex~$\vertexb$. This yields the desired \gls*{soc} at~$\vertexb$.
Note that $\socfunction{\alabel}$ can evaluate to $-\infty$ for values greater than~$\triptime$, due to battery constraints applied by the \gls*{soc} profile~$\socprofile_{[\vertexa,\dots,\vertexb]}$.
We denote by
\begin{align*}
 \minfeasibletime{\socfunction{\alabel}}:=\min\{\atime\in\posreals\mid\socfunction{\alabel}(\atime)\neq-\infty\}
\end{align*}
the smallest value for which $\socfunction{\alabel}$ is greater than~$-\infty$, \ie, the minimum trip time required for the path represented by~$\alabel$ to be feasible.
Figure~\ref{fig:cfp-example-label} shows an example of an \gls*{soc} function of a label.
The definition of \gls*{soc} functions reflects the interpretation of our labels, which represent all tradeoffs between charging time and resulting \gls*{soc} induced by the charging function $\chargingfunction_\vertexa$ of the last station~$\vertexa$.

\begin{algorithm}[tp]
  \caption{Pseudocode of the \acrshort*{cfp} algorithm.}
  \label{alg:cfp}
  \BlankLine
  \tcp{initialize label sets}%
  \ForEach{$\vertex\in\vertices$}{\label{line:cfp:init}%
    $\labelset_{\operatorname{set}}(\vertex)\assign\emptyset$\;%
    $\labelset_{\operatorname{uns}}(\vertex)\assign\emptyset$\;%
  }%
  $\vertex^*\assign$ dummy vertex without incident edges that is (temporarily) added to $\vertices$\;% 
  $\chargingstations\assign\chargingstations\cup\{\vertex^*\}$\;%
  $\simplechargingfunction_{\vertex^*}\assign[(0, \soc_\source)]$\;%
  $\labelset_{\operatorname{uns}}(\source)\assign\{(0,\soc_\source,\vertex^*,\socprofile_{[\source]})\}$\;%
  $\queue$.\queueInsert{$\source,0$}\label{line:cfp:initsource}%
  \BlankLine
  \tcp{run main loop}%
  \While{$\queue$.\queueIsNotEmpty{}}{%
    \BlankLine
    $\vertexb\assign\queue$.\queueMinElement{}\label{line:cfp:extractmin}\;%
    $\alabel=(\triptime,\soc_\vertexa,\vertexa,\socprofile_{[\vertexa,\dots,\vertexb]})\assign\labelset_{\operatorname{uns}}(\vertexb)$.\queueDeleteMin{}\;%
    $\labelset_{\operatorname{set}}(\vertexb)$.\queueInsert{$\alabel$}\label{line:cfp:settle}\;%
    \If{$\vertexb=\target$}{%
      \Return $\minfeasibletime{\socfunction{\alabel}}$\label{line:cfp:stoppingcriterion}\;%
      }%
    \BlankLine
    \tcp{handle charging stations; see Section~\ref{sec:spawning}}%
    \If{$\vertexb\in\chargingstations\setminus\{\vertexa\}$}{\label{line:cfp:csfirst}%
      \ForEach{$\chargingtime\in\switchingsequence(\slopeOld{\alabel},\slopeNew{\alabel})\setminus\{\infty\}$}{\label{line:cfp:switchingsequence}%
      $\labelset_{\operatorname{uns}}(\vertexb)$.\queueInsert{$(\triptime+\chargingtime+\arrangementtime(\vertexb),\socfunction{\alabel}(\triptime+\chargingtime),\vertexb,\socprofile_{[\vertexb]})$}\label{line:cfp:cslast}\;%
      }%
    }%
    \BlankLine
    \tcp{update priority queue}%
    \eIf{$\labelset_{\operatorname{uns}}(\vertexb)$.\queueIsNotEmpty{}}{\label{line:cfp:updatemin}%
      $\alabel'\assign\labelset_{\operatorname{uns}}(\vertexb)$.\queueMinElement{}\;%
      $\queue$.\queueUpdate{$\vertexb,\queuekey(\alabel')$}\;\label{line:cfp:updatekey}%
    }{%
      $\queue$.\queueDeleteMin{}\label{line:cfp:deletemin}\;%
    }%
    \BlankLine
    \tcp{scan outgoing arcs}%
    \ForEach{$(\vertexb,\vertexc)\in\arcs$}{\label{line:cfp:scanedges}%
      $\socprofile_{[\vertexa,\dots,\vertexc]}\assign\linkop(\socprofile_{[\vertexa,\dots,\vertexb]},\socprofile_{(\vertexb,\vertexc)})$\;
      \If{$\socprofile_{[\vertexa,\dots,\vertexc]}(\maxcharge_\vertexa)\neq-\infty$}{%
        $\alabel'\assign(\triptime+\drivingtimefunction(\vertexb,\vertexc),\soc_\vertexa,\vertexa,\socprofile_{[\vertexa,\dots,\vertexc]})$\;%
        $\labelset_{\operatorname{uns}}(\vertexc)$.\queueInsert{$\alabel'$}\;%
        \If{$\alabel'=\labelset_{\operatorname{uns}}(\vertexc)$.\queueMinElement{}}{%
         $\queue$.\queueUpdate{$\vertexc,\queuekey(\alabel')$}\;\label{line:cfp:inserthead}%
        }%
      }%
    }%
  }%
\end{algorithm}

\paragraph{Algorithm Description.}
We are now ready to describe the actual \gls*{cfp} algorithm; see Algorithm~\ref{alg:cfp} for pseudocode.
It propagates labels that are quadruples as defined above. Given two labels $\alabel$ and~$\alabel'$, we say that $\alabel$ \emph{dominates} $\alabel'$ if $\socfunction{\alabel}(\atime)\ge\socfunction{\alabel'}(\atime)$ holds for all $\atime\in\posreals$.
The \emph{key} of a label~$\alabel$, denoted~$\queuekey(\alabel):=\minfeasibletime{\socfunction{\alabel}}$, is defined as its minimum feasible trip time. This value is not stored explicitly in the label, but computed on-the-fly by evaluating the inverse of the charging function of the previous station at the minimum \gls*{soc} for which the subpath from this station to the current vertex becomes feasible (one could also store the key with the label, but this did not improve performance in our experiments).

The algorithm stores two sets $\labelset_\text{set}(\vertex)$ and $\labelset_\text{uns}(\vertex)$ for each vertex~$\vertex\in\vertices$, containing \emph{settled} (\ie, extracted) and \emph{unsettled} labels, respectively.
Sets~$\labelset_\text{uns}(\cdot)$ are organized as priority queues (implemented as binary heaps), allowing efficient extraction of the unsettled label with minimum key (breaking ties by the corresponding minimum \gls*{soc} of a label). We maintain the invariant that for each~$\vertex\in\vertices$, $\labelset_\text{uns}(\vertex)$ is empty or the minimum label~$\alabel$~(\wrt~its key) in $\labelset_\text{uns}(\vertex)$ is not dominated by any label in~$\labelset_\text{set}(\vertex)$.
Every time the minimum element of the heap changes, because an element is removed or added, we check whether the new minimum is dominated by a label in~$\labelset_\text{set}(\vertex)$ and remove it in this case (as we know that it cannot lead to an optimal solution).
For piecewise-defined \gls*{soc} functions, a dominance test requires a linear scan over the subfunctions of both \gls*{soc} functions.
By using heaps for unsettled labels, we avoid unnecessary dominance checks for labels that are never settled.
(A more straightforward variant could use a single set per vertex and follow the basic algorithm outlined in Section~\ref{sec:problem} to identify dominated labels, however, this led to slower running times in preliminary tests.)

Given a source $\source\in\vertices$, a target $\target\in\vertices$, and the initial \gls*{soc} $\soc_\source\in[0,\maxbattery]$, the algorithm is initialized in Lines~\ref{line:cfp:init}--\ref{line:cfp:initsource} with a single label $(0,\soc_\source,\vertex^*,\socprofile_{[\source]})$ at the source~$\source$, while all other label sets are empty.
Note that $\vertex^*$ is a special vertex that is (temporarily) added to the graph as a charging station with the charging function~$\chargingfunction_{\vertex^*}\equiv\soc_\source$. Thereby, we avoid explicit handling of special cases when reaching the first actual charging station.
The \gls*{soc} profile stored in the label is initialized with the identity function~$\socprofile_{[\source]}$ (\ie, the \gls*{soc} is not affected when applying this function).
The source vertex is also inserted into a priority queue (denoted $\queue$ in Algorithm~\ref{alg:cfp}). The key of a vertex $\vertex\in\vertices$ in the priority queue is the key of the minimum element in $\labelset_\text{uns}(\vertex)$, \ie, the minimum feasible trip time among the \gls*{soc} functions of all unsettled labels.

The algorithm then proceeds along the lines of the \gls*{bsp} algorithm.
In each step of the main loop, it first extracts a vertex $\vertexb\in\vertices$ with minimum key (breaking ties by~\gls*{soc}) from the priority queue and settles it; see Lines~\ref{line:cfp:extractmin}--\ref{line:cfp:settle} of Algorithm~\ref{alg:cfp}.
Note that at this point, the key of the corresponding label $\alabel=(\triptime,\soc_\vertexa,\vertexa,\socprofile_{[\vertexa,\dots,\vertexb]})$ extracted from $\labelset_\text{uns}(\vertexb)$ is not greater than that of any label that has not been settled yet.

Next, we check whether $\vertexb$ is a charging station that differs from the one stored in the current label~$\alabel$, \ie,~$\vertexb\in\chargingstations\setminus\{\vertexa\}$.
If this is the case, we create new labels to incorporate possible recharging at~$\vertexb$; see Lines~\ref{line:cfp:csfirst}--\ref{line:cfp:cslast}.
This means that we have to spawn new labels~$\alabel'$ that replace the previous charging station $\vertexa$ by~$\vertexb$.
We can do so by fixing a charging time $\chargingtime\in\posreals$ at~$\vertexa$. For the resulting \gls*{soc} at~$\vertexa$, we evaluate the \gls*{soc} profile $\socprofile_{[\vertexa,\dots,\vertexb]}$ of the $\vertexa$--$\vertexb$~path to determine the \gls*{soc} at~$\vertexb$.
We update the trip time accordingly by adding the charging time $\chargingtime$ at $\vertexa$ and the initialization time $\arrangementtime(\vertexb)$ at the new charging station~$\vertexb$. We obtain the new label
\begin{align}
 \alabel'&:=(\triptime+\chargingtime+\arrangementtime(\vertexb),\socprofile_{[\vertexa,\dots,\vertexb]}(\chargingfunction_\vertexa(\soc_\vertexa,\chargingtime)),\vertexb,\socprofile_{[\vertexb]})\notag\\
 &\hphantom{:}=(\triptime+\chargingtime+\arrangementtime(\vertexb),\socfunction{\alabel}(\triptime+\chargingtime),\vertexb,\socprofile_{[\vertexb]})\text{.}\label{eq:cfp-spawned-label}
\end{align}
However, we still face the problem that in general, there are infinitely many possible charging times~$\chargingtime$ at the previous charging station~$\vertexa$ held in $\alabel$.
In Section~\ref{sec:spawning}, we show that for realistic models of charging stations, we only have to consider a small (finite) number of relevant charging times at $\vertexa$ when charging at~$\vertexb$. Thus, spawning a limited number of new labels, each fixing a certain charging time at $\vertexa$ and setting the last charging station to~$\vertexb$, is sufficient to represent all nondominated solutions.
Note that the original label $\alabel$ is not discarded, to reflect the possibility of not stopping at the charging station~$\vertexb$.

In Lines~\ref{line:cfp:updatemin}--\ref{line:cfp:deletemin} of Algorithm~\ref{alg:cfp}, the key of~$\vertexb$ in the priority queue is updated. Since the label $\alabel$ was settled and new labels may have spawned in case $\vertexb$ is a charging station, we update the key of $\vertexb$ to the new smallest key of an unsettled label, if it exists. Otherwise, $\vertexb$ is removed from the queue.

Afterwards, we scan all outgoing arcs~$(\vertexb,\vertexc)\in\arcs$; see Lines~\ref{line:cfp:scanedges}--\ref{line:cfp:inserthead}. Given the current label~$\alabel=(\triptime,\soc_\vertexa,\vertexa,\socprofile_{[\vertexa,\dots,\vertexb]})$, traversing the arc~$(\vertexb,\vertexc)$ means to increase trip time by $\drivingtimefunction(\vertexb,\vertexc)$ and apply the (constant-time) link operation to the \gls*{soc} profile $\socprofile_{[\vertexa,\dots,\vertexb]}$ of $\alabel$ and the \gls*{soc} profile $\socprofile_{[\vertexb,\vertexc]}$ induced by the energy consumption $\consumptionfunction(\vertexb,\vertexc)$; see Section~\ref{sec:problem}.
We compute $\socprofile_{[\vertexa,\dots,\vertexc]}:=\linkop(\socprofile_{[\vertexa,\dots,\vertexb]},\socprofile_{[\vertexb,\vertexc]})$ and construct the label
\begin{align*}
 \alabel':=(\triptime+\drivingtimefunction(\vertexb,\vertexc),\soc_\vertexa,\vertexa,\socprofile_{[\vertexa,\dots,\vertexc]})\text{.}
\end{align*}
Unless the \gls*{soc} profile $\socprofile_{[\vertexa,\dots,\vertexc]}$ of $\alabel'$ indicates that the $\vertexa$--$\vertexc$~subpath is infeasible for arbitrary~\gls*{soc}, the new label $\alabel'$ is added to the label set at the vertex~$\vertexc$.
Note that we perform no dominance checks at this point (unless the minimum element in the label set $\labelset_\text{uns}(\vertexc)$ changes), exploiting the fact that labels are organized in two sets per vertex.

When extracting a label $\alabel$ at the target vertex $\target$ for the first time, we pick the least charging time at the last station such that $\target$ can be reached, \ie, the minimum feasible trip time $\minfeasibletime{\socfunction{\alabel}}$ of~$\socfunction{\alabel}$, and the algorithm terminates; see Line~\ref{line:cfp:stoppingcriterion} in Algorithm~\ref{alg:cfp}.
Correctness of the \gls*{cfp} algorithm follows from Lemma~\ref{lem:charging-station-settling} shown in Section~\ref{sec:spawning} below and the fact that the first extracted label $\alabel$ at $\target$ minimizes the feasible trip time (recall that the algorithm is label setting and minimum feasible trip time is used as key in the priority queue).
Theorem~\ref{thm:cfp-correctness} at the end of Section~\ref{sec:spawning} summarizes these insights.
The asymptotic running time of the algorithm is exponential in the input graph in the worst case (for reasonable charging models; see Section~\ref{sec:spawning}).

For path retrieval, we add two pointers to each label, storing its parent vertex and parent label. For a charging station~$\vertex\in\chargingstations$, the vertex $\vertex$ can be its own parent.
Two consecutive identical parents then imply the use of a (previous) charging station $\vertexa\in\chargingstations$, which is stored in the former label. The according charging time is the difference between the trip times of both labels.

%%%%%%%%%%%%%%%%%%%%%%%%%%%%%%%%%%%%%%%%%%%%%%%%%%%%%%%%%%%%%%%%%%%%%%%%%%%%%%%%
\section{Spawning Labels at Charging Stations}\label{sec:spawning}
%%%%%%%%%%%%%%%%%%%%%%%%%%%%%%%%%%%%%%%%%%%%%%%%%%%%%%%%%%%%%%%%%%%%%%%%%%%%%%%%

As described in Section~\ref{sec:approach}, the \gls*{cfp} algorithm constructs new labels at charging stations to represent all nondominated solutions. We now prove that for reasonable models of charging functions, it suffices to spawn a small number of labels that replace the previous charging station with the new one.
The key idea is that we only require labels that correspond to charging at the station that offers the better charging speed at a certain (relative) point in time. We define \emph{switching sequences} for pairs of functions, containing points at which the charging speed of the new function surpasses the old one. Lemma~\ref{lem:charging-station-settling} proves that spawning one label per element of the switching sequence suffices. Moreover, switching sequences are finite (and linear in the descriptive complexity) for typical models of charging functions, which implies that the \gls*{cfp} algorithm terminates. Before proving Lemma~\ref{lem:charging-station-settling}, we introduce helpful tools. We also formalize switching sequences and the slope of an \gls*{soc} function.

Consider a label $\alabel=(\triptime,\soc_\vertexa,\vertexa,\socprofile_{[\vertexa,\dots,\vertexb]})$ extracted at some charging station~$\vertexb\in\chargingstations$.
We want to create new labels that reflect charging at~$\vertexb$.
This requires us to fix a charging time $\chargingtime\in\posreals$ at the previous station~$\vertexa$, so that we can set $\vertexb$ as the last visited charging station of a new label~$\alabel'$; see Equation~\ref{eq:cfp-spawned-label} and Figure~\ref{fig:cfp-label-spawn}.
We denote the resulting label by~$(\distancelabelWithChargingTime{\chargingtime}):=\alabel'$, as it is obtained after setting the charging time in $\alabel$ to~$\chargingtime$.
Recall that in the label~$\distancelabelWithChargingTime{\chargingtime}$, we replace the old charging station $\vertexa$ by the new station~$\vertexb$. Moreover, we set $\triptime+\chargingtime+\arrangementtime(\vertexb)$ as its overall trip time and~$\socfunction{\alabel}(\triptime+\chargingtime)$ as the corresponding arrival \gls*{soc} at~$\vertexb$.
The \gls*{soc} function~$\socfunction{\distancelabelWithChargingTime{\chargingtime}}$ represents all tradeoffs between charging time at the new charging station $\vertexb$ and resulting \gls*{soc}.
If the label $\distancelabelWithChargingTime{\chargingtime}$ is not feasible, \ie,~$\triptime+\chargingtime<\minfeasibletime{\socfunction{\alabel}}$, we obtain~$\socfunction{\distancelabelWithChargingTime{\chargingtime}}\equiv-\infty$.

Not every charging time $\chargingtime\in\posreals$ at $\vertexa$ yields a reasonable solution. In particular, if we can find a charging time $\chargingtime'\in\posreals$ such that $\socfunction{\distancelabelWithChargingTime{\chargingtime'}}$ dominates~$\socfunction{\distancelabelWithChargingTime{\chargingtime}}$, we know that a charging time of $\chargingtime$ is never beneficial. Intuitively, if the new charging station $\vertex$ allows fast charging, it could pay off to charge less energy at the previous station $\vertexa$ and spend more time at $\vertexb$ instead, so $\socfunction{\distancelabelWithChargingTime{\chargingtime-\varepsilon}}$ dominates $\socfunction{\distancelabelWithChargingTime{\chargingtime}}$ for some~$\varepsilon>0$. Similarly, if the charging station $\vertexa$ offers better charging speed, a charging time $\chargingtime+\varepsilon$ might be the better choice.
In other words, the best choice of the value $\chargingtime$ depends on the slopes of the two \gls*{soc} functions $\socfunction{\alabel}$ and~$\socfunction{\distancelabelWithChargingTime{\chargingtime}}$.

\begin{figure}[t]
 \centering%
 \begin{tikzpicture}[figure,scale=\ChargingFunctionExampleScale]
  \ExampleDrawCoordinateSystemWithNegativeInfty{1}{8}{0}{6}{0}{0}
 
  % Axis names.
  \PlotXAxisName{8}{0}{$\atime$}
 
  % Lines. 
  \draw[helperline] (2.5,0) -- (2.5,6) node[above,edgelabel] {$\atime_1$};
  \draw[helperline] (3.5,0) -- (3.5,6) node[above,edgelabel] {$\atime_2$};
 
  % Arrows.
  \draw[<->] (0.2,0.2) -- (1.45,0.2) node [pos=0.75,above=1.5pt,edgelabel] {$\triptime$};
  \draw[<->] (1.55,0.2) -- (2.4,0.2) node [pos=1.1,above=1.5pt,edgelabel] {$\chargingtime$};
  \draw[<->] (2.6,2.5) -- (3.3,2.5) node [pos=1.25,below=3pt,edgelabel] {$\arrangementtime(\vertexb)$};
 
  % Function label.
  \node[edgelabel,text=secondarycolor-dark] (f) at (6.05,3.1) {$\socfunction{\distancelabelWithChargingTime{\chargingtime}}(\atime)$};
 
  % Draw function.
  \begin{scope}
   \begin{pgfinterruptboundingbox}
    \clip (0,-1-\ExampleInftyOffset) rectangle (8,7);
    \draw[chargingfunction,secondarycolor-light] (0,0) -- node [pos=0.4,above,edgelabel,sloped] {$\simplechargingfunction_\vertexb(\atime)$} (4.5,6) -- (8,6);
    \draw[chargingfunction,secondarycolor-dark] (0,-\ExampleInftyOffset) -- (3.5,-\ExampleInftyOffset);
    \draw[function] (0,-\ExampleInftyOffset) -- (2,-\ExampleInftyOffset);
    \draw[function] (2,1.5) -- (3,3.5) -- (6,5) -- node [midway,above,edgelabel,text=primarycolor] {$\socfunction{\alabel}(\atime)$} (8,5);
    \draw[chargingfunction,secondarycolor-dark] (3.5,2.5) -- (6.125,6) -- (8,6);
   \end{pgfinterruptboundingbox}
  \end{scope}
 
  % Markers at discontinuities.
  \begin{pgfonlayer}{foreground}
   \node[discontinuityblank] (d1) at (2,-\ExampleInftyOffset) {};
   \node[discontinuityfilled] (d2) at (2,1.5) {};
   \node[discontinuityblank,draw=secondarycolor-dark] (d1) at (3.5,-\ExampleInftyOffset) {};  
   \node[discontinuityfilled,secondarycolor,secondarycolor-dark] (d2) at (3.5,2.5) {};
  \end{pgfonlayer}
 \end{tikzpicture}
 \caption{Spawning a label at a charging station. Given the \gls*{soc} function $\socfunction{\alabel}$ of a label $\alabel=(\triptime,\soc_\vertexa,\vertexa,\socprofile_{[\vertexa,\dots,\vertexb]})$ at a charging station $\vertexb\in\chargingstations$, we can spawn a new label $\distancelabelWithChargingTime{\chargingtime}$ by picking a charging time $\chargingtime$ at the station~$\vertexa$. We compare the slopes of $\socfunction{\alabel}$ at $\atime_1=\triptime+\chargingtime$ and $\socfunction{\distancelabelWithChargingTime{\chargingtime}}$ at $\atime_2=\triptime+\chargingtime+\arrangementtime(\vertexb)$ to determine which one represents the better choice. Note that $\triptime$ is smaller than~$\minfeasibletime{\socfunction{\alabel}}$, due to battery constraints.}%
\label{fig:cfp-label-spawn}%
\end{figure}

We define the \emph{slope} of a given function~$\afunction$ at some $\atime'\in\posreals$ as the corresponding \emph{right derivative} $(\derivative\afunction(\atime)/\derivative\atime)(\atime')$, to ensure that the slope is well-defined also for piecewise-defined \gls*{soc} functions and at the minimum feasible trip time of an \gls*{soc} function.
As before, let $\alabel=(\triptime,\soc_\vertexa,\vertexa,\socprofile_{[\vertexa,\dots,\vertexb]})$ be a label at a charging station~$\vertexb\in\chargingstations$. We introduce a function $\slopeOld{\alabel}\colon\posreals\to\posreals\cup\{\infty\}$ that describes the slope of the \gls*{soc} function $\socfunction{\alabel}$ at $\triptime+\chargingtime$ as a function of the charging time~$\chargingtime\in\posreals$.
Formally, we define
\begin{align*}
\slopeOld{\alabel}(\chargingtime) :=& \begin{cases}
 \frac{\derivative\socfunction{\alabel}(x)}{\derivative x}(\triptime+\chargingtime) & \text{if}~\triptime+\chargingtime\geq\minfeasibletime{\socfunction{\alabel}}, \\
 \infty & \mbox{otherwise.}
\end{cases}
\end{align*}
Note that the slope $\slopeOld{\alabel}(\chargingtime)$ of the \gls*{soc} function $\socfunction{\alabel}$ at $\triptime+\chargingtime$ is equivalent to the slope of the charging function $\chargingfunction_\vertexa$ of the vertex $\vertexa$ for the arrival \gls*{soc} $\soc_\vertexa$ and the charging time~$\chargingtime$.
Hence, we obtain $\slopeOld{\alabel}(\chargingtime)=(\derivative \chargingfunction_\vertexa(\soc_\vertexa, x)/\derivative x)(\chargingtime)$ for all $\chargingtime\in\posreals$, unless battery constraints on the $\vertexa$--$\vertexb$~path render a charging time of $\chargingtime$ infeasible or unprofitable (in which case the slope $\slopeOld{\alabel}$ is either $\infty$ or~$0$).
Thus, $\slopeOld{\alabel}(\chargingtime)$ can be interpreted as the charging speed when continuing to charge the battery at $\vertexa$ after a charging time of~$\chargingtime$.

Alternatively, one could interrupt charging at $\vertexa$ after a charging time of $\chargingtime$, continue the journey, and switch to the new charging station upon arrival at~$\vertexb$.
The charging speed that can be achieved in this case is given by the slope of the \gls*{soc} function $\socfunction{\distancelabelWithChargingTime{\chargingtime}}(\atime)$ at~$\atime=\triptime+\chargingtime+\arrangementtime(\vertexb)$; see Figure~\ref{fig:cfp-label-spawn}.
Similar to~$\slopeOld{\alabel}$, we define the function $\slopeNew{\alabel}\colon\posreals\to\posreals\cup\{\infty\}$ as 
\begin{align*}
\slopeNew{\alabel}(\chargingtime) :=& \begin{cases}
 \frac{\derivative\socfunction{\distancelabelWithChargingTime{\chargingtime}}(\atime)}{\derivative\atime}(\triptime\!+\!\chargingtime\!+\!\arrangementtime(\vertexb)) &\!\text{if}~\triptime\!+\!\chargingtime\geq\minfeasibletime{\socfunction{\alabel}}, \\
 \infty &\!\text{otherwise.}
\end{cases}
\end{align*}
It maps total charging time $\chargingtime$ at $\vertexa$ to the slope of the \gls*{soc} function $\socfunction{\distancelabelWithChargingTime{\chargingtime}}$ at the time $\triptime+\chargingtime+\arrangementtime(\vertexb)$ that corresponds to arrival at $\vertexb$ and starting to charge the battery.
Hence, the value of $\slopeNew{\alabel}(\chargingtime)$ is equivalent to the slope of the new charging function $\chargingfunction_\vertexb$ of the vertex $\vertexb$ for the arrival \gls*{soc} $\socfunction{\alabel}(\triptime+\chargingtime)$ and the charging time~$0$.
More formally, we have $\slopeNew{\alabel}(\chargingtime)=(\derivative\chargingfunction_\vertexb(\socfunction{\alabel}(\triptime+\chargingtime),\atime)/\derivative\atime)(0)$ if $\triptime+\chargingtime$ is at least~$\minfeasibletime{\socfunction{\alabel}}$ and $\slopeNew{\alabel}(\chargingtime)=\infty$ otherwise.

Given the two functions $\slopeOld{\alabel}$ and $\slopeNew{\alabel}$ defined above, we are interested in points in time where $\slopeNew{\alabel}$ surpasses~$\slopeOld{\alabel}$, because at such points it may pay off to interrupt charging at $\vertexa$ to benefit from a better charging rate at $\vertexb$ later on. Additionally, the minimum charging time $\minfeasibletime{\socfunction{\alabel}}-\triptime$ that is necessary to reach the new charging station $\vertexb$ may be relevant in cases where the charging function of $\vertexb$ has low slope but a large minimum \gls*{soc} value (\eg, if $\vertexb$ is a swapping station). We define the \emph{switching sequence} of $\slopeOld{\alabel}$ and $\slopeNew{\alabel}$, denoted
\begin{align*}
 \switchingsequence(\slopeOld{\alabel},\slopeNew{\alabel}):=[\minfeasibletime{\socfunction{\alabel}}-\triptime=\atime_1,\atime_2,\dots,\atime_{k-1},\atime_k=\infty], 
\end{align*}
as the sequence of all candidate points in time to interrupt charging at the old station~$\vertexa$, arranged in ascending order. (We only include the value $\atime_k=\infty$ to avoid additional case distinctions in the proof of Lemma~\ref{lem:charging-station-settling} below.)
Formally, we demand for $\switchingsequence(\slopeOld{\alabel},\slopeNew{\alabel})$ that~$\atime_1=\minfeasibletime{\socfunction{\alabel}}-\triptime$,~$\atime_k=\infty$, $\atime_i<\atime_{i+1}$ for~$i\in\{1,\dots,k-1\}$, and for all~$i\in\{2,\dots,k-1\}$ there exists a value $\varepsilon>0$ such that for all $0<\delta<\varepsilon$ it holds that $\slopeOld{\alabel}(\atime_i-\delta)\ge\slopeNew{\alabel}(\atime_i-\delta)$ and~$\slopeOld{\alabel}(\atime_i+\delta)<\slopeNew{\alabel}(\atime_i+\delta)$. Moreover, we assume $\switchingsequence(\slopeOld{\alabel},\slopeNew{\alabel})$ to be \emph{maximal}, \ie, it contains all values with the above property.
In general, pairs of functions do not necessarily have a switching sequence of finite length~$k\in\naturals$. At the end of this section we argue that the length of a switching sequence is linear in the descriptive complexity (and thus, finite) for reasonable models of charging functions.

\begin{figure}[t]
 \centering
 \begin{subfigure}[b]{.5\textwidth}%
  \centering%
  \begin{tikzpicture}[figure,scale=\ChargingFunctionExampleScale]
   \ExampleDrawCoordinateSystemWithNegativeInfty{1}{8}{0}{6}{0}{0}

   % Axis names.
   \PlotXAxisName{8}{0}{$\atime$}

   % Lines. 
   \draw[helperline] (3,0) -- (3,6) node[above,edgelabel] {$\vphantom{\chargingtime}\atime_i$};
   \draw[helperline] (5,0) -- (5,6) node[above,edgelabel] {$\chargingtime$};
   \node[edgelabel] at (3,7.02) {$\triptime$};
   \node[edgelabel] at (3,6.69) {$+$};
   \node[edgelabel] at (5,7.02) {$\triptime$};
   \node[edgelabel] at (5,6.69) {$+$};

   % Arrows.
   \draw[<-] (5.2,5) -- (5.8,5) node [pos=0.8,above=2pt,edgelabel] {$\Delta$};

   % Function label.
   \node[edgelabel,text=secondarycolor-dark] (f2) at (6.05,2.75) {$\socfunction{\distancelabelWithChargingTime{\chargingtime}}(\atime)$};
 
   % Draw function.
   \begin{scope}
    \begin{pgfinterruptboundingbox}
     \clip (0,-1-\ExampleInftyOffset) rectangle (8,7);
     \draw[chargingfunction,black50] (2,-\ExampleInftyOffset) -- (3,-\ExampleInftyOffset);
     \draw[chargingfunction,secondarycolor-dark] (3,-\ExampleInftyOffset) -- (5,-\ExampleInftyOffset);
     \draw[function] (0,-\ExampleInftyOffset) -- (2,-\ExampleInftyOffset);
     \draw[chargingfunction,black30,dashed] (1,0) -- (2,2) -- (3,3);
     \draw[chargingfunction,secondarycolor-light,dashed] (2,0) -- (3,2) -- (5,4);
     \draw[function] (2,0) -- (3,3) -- node [pos=0.75,below=2pt,edgelabel,sloped,text=primarycolor] {$\socfunction{\alabel}(\atime)$} (7,5) -- (8,5);
     \draw[chargingfunction,black50] (3,3) -- node [pos=0.45,above=2pt,edgelabel,sloped,text=black50] {$\socfunction{\distancelabelWithChargingTime{\atime_i}}(\atime)$} (6,6) -- (8,6);
     \draw[chargingfunction,secondarycolor-dark] (5,4) -- (7,6) -- (8,6);
    \end{pgfinterruptboundingbox}
   \end{scope}
 
   % Markers at discontinuities.
   \begin{pgfonlayer}{foreground}
    \node[discontinuityblank] (d1) at (2,-\ExampleInftyOffset) {};
    \node[discontinuityfilled] (d2) at (2,0) {};
    \node[discontinuityblank,draw=black50] (d3) at (3,-\ExampleInftyOffset) {};
    \node[discontinuityfilled,secondarycolor,black50] (d4) at (3,3) {};
    \node[discontinuityblank,draw=secondarycolor-dark] (d5) at (5,-\ExampleInftyOffset) {};
    \node[discontinuityfilled,secondarycolor,secondarycolor-dark] (d6) at (5,4) {};
   \end{pgfonlayer}
  \end{tikzpicture}
  \caption{}%
  \label{fig:cfp-switchingsequence-proof:a}%
 \end{subfigure}%
 \begin{subfigure}[b]{.5\textwidth}%
  \centering%
  \begin{tikzpicture}[figure,scale=\ChargingFunctionExampleScale]
   \ExampleDrawCoordinateSystemWithNegativeInfty{1}{8}{0}{6}{0}{0}

   % Axis names.
   \PlotXAxisName{8}{0}{$\atime$}

   % Lines. 
   \draw[helperline] (4,0) -- (4,6) node[above,edgelabel,shift={(-0.1,0)}] {$\chargingtime$};
   \draw[helperline] (5,0) -- (5,6) node[above,edgelabel,shift={(0.1,0)}] {$\vphantom{\chargingtime}\atime_{i+1}$};
   \node[edgelabel] at (3.9,7.02) {$\triptime$};
   \node[edgelabel] at (3.9,6.69) {$+$};
   \node[edgelabel] at (5.1,7.02) {$\triptime$};
   \node[edgelabel] at (5.1,6.69) {$+$};

   % Arrows.
   \draw[<-] (2.2,2) -- (2.8,2) node [pos=0.8,above=2pt,edgelabel] {$\Delta$};

   % Function label.
   \node[edgelabel,text=black50] (f1) at (3.25,5.5) {$\socfunction{\distancelabelWithChargingTime{\atime_{i+1}}}(\atime)$};
   \node[edgelabel,text=secondarycolor-dark] (f2) at (6.05,2.75) {$\socfunction{\distancelabelWithChargingTime{\chargingtime}}(\atime)$};

   % Draw function.
   \begin{scope}
    \begin{pgfinterruptboundingbox}
     \clip (0,-1-\ExampleInftyOffset) rectangle (8,7);
     \draw[chargingfunction,black50] (4,-\ExampleInftyOffset) -- (5,-\ExampleInftyOffset);
     \draw[chargingfunction,secondarycolor-dark] (3,-\ExampleInftyOffset) -- (4,-\ExampleInftyOffset);
     \draw[function] (0,-\ExampleInftyOffset) -- (3,-\ExampleInftyOffset);
     \draw[chargingfunction,black30,dashed] (1,0) -- (2,2) -- (5,5);
     \draw[chargingfunction,secondarycolor-light,dashed] (2,0) -- (3,2) -- (4,3);
     \draw[function] (3,1) -- (5,5) -- node [pos=0.65,below=2pt,edgelabel,text=primarycolor] {$\socfunction{\alabel}(\atime)$} (8,5);
     \draw[chargingfunction,black50] (5,5) -- (6,6) -- (8,6);
     \draw[chargingfunction,secondarycolor-dark] (4,3) -- (7,6) -- (8,6);
    \end{pgfinterruptboundingbox}
   \end{scope}
 
    % Markers at discontinuities.
   \begin{pgfonlayer}{foreground}
    \node[discontinuityblank] (d1) at (3,-\ExampleInftyOffset) {};
    \node[discontinuityfilled] (d2) at (3,1) {};
    \node[discontinuityblank,draw=black50] (d3) at (5,-\ExampleInftyOffset) {};
    \node[discontinuityfilled,black50] (d4) at (5,5) {};
    \node[discontinuityblank,draw=secondarycolor-dark] (d5) at (4,-\ExampleInftyOffset) {};
    \node[discontinuityfilled,secondarycolor,secondarycolor-dark] (d6) at (4,3) {};
   \end{pgfonlayer}
  \end{tikzpicture}
  \caption{}%
  \label{fig:cfp-switchingsequence-proof:b}%
 \end{subfigure}%
 \caption
 {Illustration of dominated \gls*{soc} functions at a charging station~$\vertexb\in\chargingstations$. For simplicity, we assume that~$\arrangementtime(\vertexb)=0$. Dashed segments indicate the (shifted) charging function $\simplechargingfunction_\vertexb$ of~$\vertexb$. Note that a function dominates the area beneath it. (a) The slope of the \gls*{soc} function~$\socfunction{\alabel}$ is lower than the slope of~$\socfunction{\distancelabelWithChargingTime{\chargingtime}}$ at~$\triptime+\chargingtime=5$. Hence, it pays off to decrease charging time at the previous station to $\atime_i=3$ and charge more energy at~$\vertexb$ instead. (b)~The slope of $\socfunction{\alabel}$ is greater than the slope of $\socfunction{\distancelabelWithChargingTime{\chargingtime}}$ for~$\triptime+\chargingtime=4$. Therefore, charging more energy at the previous charging station pays off and the \gls*{soc} functions $\socfunction{\alabel}$ and $\socfunction{\distancelabelWithChargingTime{\atime_{i+1}}}$ together dominate the function~$\socfunction{\distancelabelWithChargingTime{\chargingtime}}$.}
 \label{fig:cfp-switchingsequence-proof}
\end{figure}

We now prove Lemma~\ref{lem:charging-station-settling}, stating that we only have to spawn a bounded number of new labels at charging stations. In particular, it suffices to add at most one label per element in the switching sequence induced by the last visited charging station in the current label and the new station.

\begin{lemma}
\label{lem:charging-station-settling}
For a vertex~$\vertexb\in\chargingstations$ and a label~$\alabel=(\triptime,\soc_\vertexa,\vertexa,\socprofile_{[\vertexa,\dots,\vertexb]})$ at~$\vertexb$, let the (finite) switching sequence of $\slopeOld{\alabel}$ and $\slopeNew{\alabel}$ be given as~$\switchingsequence(\slopeOld{\alabel},\slopeNew{\alabel})=[\atime_1,\atime_2,\dots,\atime_{k-1},\atime_k]$.
For every charging time~$\chargingtime\in\posreals$, there exists an $i\in\{1,\dots,k-1\}$ such that the \gls*{soc} functions $\socfunction{\alabel}$ and $\socfunction{\distancelabelWithChargingTime{\atime_i}}$ together dominate the \gls*{soc} function~$\socfunction{\distancelabelWithChargingTime{\chargingtime}}$, \ie, for all $\atime\in\posreals$ we have the inequality~$\max\{\socfunction{\alabel}(\atime),\socfunction{\distancelabelWithChargingTime{\atime_i}}(\atime)\}\geq\socfunction{\distancelabelWithChargingTime{\chargingtime}}(\atime)$.
\end{lemma}

\begin{proof}
Consider an arbitrary charging time $\chargingtime\in\posreals$ at the previous charging station $\vertexa$ and the induced label~$\distancelabelWithChargingTime{\chargingtime}$. If~$\triptime+\chargingtime<\minfeasibletime{\socfunction{\alabel}}$, the \gls*{soc} function $\socfunction{\distancelabelWithChargingTime{\chargingtime}}\equiv-\infty$ is dominated by~$\socfunction{\alabel}$. Hence, we assume in what follows that~$\chargingtime\ge\atime_1=\minfeasibletime{\socfunction{\alabel}}-\triptime$. Then there exists a unique index~$i\in\{1,\dots,k-1\}$, such that~$\atime_i\leq\chargingtime<\atime_{i+1}$ holds.
We distinguish two cases, depending on the slopes $\slopeOld{\alabel}$ and $\slopeNew{\alabel}$ at $\chargingtime$ and show that, together with~$\socfunction{\alabel}$, the function $\socfunction{\distancelabelWithChargingTime{\atime_i}}$ or the function $\socfunction{\distancelabelWithChargingTime{\atime_{i+1}}}$ dominates~$\socfunction{\distancelabelWithChargingTime{\chargingtime}}$.

\case{Case 1}
$\slopeOld{\alabel}(\chargingtime)<\slopeNew{\alabel}(\chargingtime)$.
Intuitively, this means that the new charging station~$\vertexb$ provides a better charging speed than the old station~$\vertexa$ for the considered charging time~$\chargingtime$.
Hence, leaving~$\vertexa$ earlier to charge more energy at~$\vertexb$ and benefit from faster charging pays off. Consequently, $\socfunction{\distancelabelWithChargingTime{\atime_i}}$ dominates~$\socfunction{\distancelabelWithChargingTime{\chargingtime}}$, since the charging speed of $\vertexb$ is better (or equally good) for all~$\atime\in[\atime_i,\chargingtime]$; see Figure~\ref{fig:cfp-switchingsequence-proof:a} for an example.
Since~$\atime_i\le\chargingtime$ and a charging time of $\atime_i$ is sufficient to reach~$\vertexb$, we have $\minfeasibletime{\socfunction{\distancelabelWithChargingTime{\atime_i}}}\le\minfeasibletime{\socfunction{\distancelabelWithChargingTime{\chargingtime}}}$. In other words, if the \gls*{soc} function induced by $\chargingtime$ is finite for some $\atime\in\posreals$, so is the function induced by~$\atime_i$.
Further, recall that the vertex $\vertexb$ is reached with an arrival \gls*{soc} of $\socfunction{\alabel}(\triptime+\atime_i)$ or $\socfunction{\alabel}(\triptime+\chargingtime)$ when the charging time at $\vertexa$ is set to $\atime_i$ or~$\chargingtime$, respectively.
By assumption, the slope~$\slopeNew{\alabel}(\chargingtime)$ is strictly positive, which means that energy can still be charged at $\vertexb$ after the \gls*{soc} has reached~$\socfunction{\alabel}(\triptime+\chargingtime)$, so it must hold that~$\socfunction{\alabel}(\triptime+\chargingtime)<\maxcharge_\vertexb$. Hence, we can define the value
\begin{align*}
 \charginggap:=\chargingfunction_\vertex^\inverse(\socfunction{\alabel}(\triptime+\atime_i),\socfunction{\alabel}(\triptime+\chargingtime))
\end{align*}
that corresponds to the time to recharge the gap in arrival \gls*{soc} between $\socfunction{\distancelabelWithChargingTime{\atime_i}}$ and $\socfunction{\distancelabelWithChargingTime{\chargingtime}}$ at the new station~$\vertexb$.
For arbitrary values~$\atime\ge\minfeasibletime{\socfunction{\distancelabelWithChargingTime{\chargingtime}}}$, we exploit the shifting property to obtain
\begin{align*}
 \socfunction{\distancelabelWithChargingTime{\atime_i}}(\atime)&=\chargingfunction_\vertexb(\socfunction{\alabel}(\triptime+\atime_i),\atime-\atime_i-\triptime-\arrangementtime(\vertexb))\\
 &=\chargingfunction_\vertexb(\socfunction{\alabel}(\triptime+\chargingtime),\atime-\atime_i-\triptime-\arrangementtime(\vertexb)-\charginggap)\\
 &=\socfunction{\distancelabelWithChargingTime{\chargingtime}}(\atime+\chargingtime-\atime_i-\charginggap)\text{.}
\end{align*}
We claim that $\timegap:=\chargingtime-\atime_i-\charginggap\ge0$ holds. This follows immediately from the fact that $\slopeNew{\alabel}(\atime)\ge\slopeOld{\alabel}(\atime)$ holds for all~$\atime\in[\atime_i,\chargingtime]$. Thus, the time $\charginggap$ spent at $\vertex$ cannot take longer than $\chargingtime-\atime_i$, the time to charge the same amount of energy at~$\vertexa$.
Consequently, we obtain~$\timegap\ge0$. This, in turn, implies that for all~$\atime\ge\minfeasibletime{\socfunction{\distancelabelWithChargingTime{\chargingtime}}}$, we have
\begin{align*}
 \socfunction{\distancelabelWithChargingTime{\atime_i}}(\atime)&=\socfunction{\distancelabelWithChargingTime{\chargingtime}}(\atime+\timegap)\\
 &\ge\socfunction{\distancelabelWithChargingTime{\chargingtime}}(\atime)\text{.}
\end{align*}

\case{Case 2}
$\slopeOld{\alabel}(\chargingtime)\ge\slopeNew{\alabel}(\chargingtime)$.
In this case, the charging station~$\vertexa$ offers a more (or equally) favorable charging speed, so it pays off to spend more time charging at~$\vertexa$.
Recall that $\socfunction{\alabel}(\triptime+\chargingtime)=\socfunction{\distancelabelWithChargingTime{\chargingtime}}(\triptime+\chargingtime+\arrangementtime(\vertex))$ holds by definition. Given that the slope $\slopeOld{\alabel}$ of $\socfunction{\alabel}$ is greater or equal for all values~$\atime\in[\chargingtime,\atime_{i+1})$, it follows that the inequality $\socfunction{\alabel}(\triptime+\atime)\ge\socfunction{\distancelabelWithChargingTime{\chargingtime}}(\triptime+\atime+\arrangementtime(\vertex))$ holds for arbitrary~$\atime\in[\chargingtime,\atime_{i+1})$.
For real-valued $\atime\ge\atime_{i+1}$ (which only exist if $i+1<k$), we proceed along the lines of the first case to obtain a nonnegative value $\timegap\ge0$ that equals the difference between the time to charge from $\socfunction{\alabel}(\triptime+\chargingtime)$ to $\socfunction{\alabel}(\triptime+\atime_{i+1})$ at $\vertexa$ and~$\vertexb$, respectively. Since~$\vertexa$ offers a charging speed at least as high as the one at~$\vertexb$ in the whole interval~$[\chargingtime,\atime_{i+1}]$, this difference, and therefore~$\timegap$, is again nonnegative. Thus, we can show (similar to Case~1) that
\begin{align*}
 \socfunction{\distancelabelWithChargingTime{\atime_{i+1}}}(\atime)&=\socfunction{\distancelabelWithChargingTime{\chargingtime}}(\atime+\timegap)\\
 &\ge\socfunction{\distancelabelWithChargingTime{\chargingtime}}(\atime)
\end{align*}
holds for arbitrary real-valued~$\atime\ge\atime_{i+1}$; see Figure~\ref{fig:cfp-switchingsequence-proof:b} for an illustration. Altogether, we obtain that $\max\{\socfunction{\alabel}(\atime),\socfunction{\distancelabelWithChargingTime{\atime_{i+1}}}(\atime)\}\geq\socfunction{\distancelabelWithChargingTime{\chargingtime}}(\atime)$ holds for all~$\atime\in\posreals$, which completes the proof of the second case.
\end{proof}

Given the label~$\alabel=(\triptime,\soc_\vertexa,\vertexa,\socprofile_{[\vertexa, \dots, \vertexb]})$ at the charging station~$\vertexb\in\chargingstations$, we spawn one new label $\distancelabelWithChargingTime{\atime}$ for each (finite) element in the switching sequence~$\switchingsequence(\slopeOld{\alabel},\slopeNew{\alabel})$. If the label is not dominated, it is added to the corresponding label set at~$\vertexb$; see Lines~\ref{line:cfp:switchingsequence}--\ref{line:cfp:cslast} in Algorithm~\ref{alg:cfp}.
For instance, in the special case that $\vertexa$ or $\vertexb$ is a swapping station, the switching sequence $[\minfeasibletime{\socfunction{\alabel}}-\triptime,\infty]$ consists of two values and exactly one new label is spawned at~$\vertexb$ (which corresponds to spending the minimum amount of charging time at~$\vertexa$ such that~$\vertexb$ can be reached).
Recall that the original label $\alabel$ is not thrown away, reflecting the possibility of not using the charging station~$\vertexb$.
Furthermore, Lemma~\ref{lem:charging-station-settling} implies that all nondominated solutions that extend $\alabel$ by charging at $\vertexb$ are computed by the algorithm. By induction, correctness is also maintained for routes with two or more charging stops. Therefore, we obtain Theorem~\ref{thm:cfp-correctness} given below.

\begin{theorem}
\label{thm:cfp-correctness}
If the switching sequences induced by arbitrary pairs of charging functions have finite length, the \gls*{cfp} algorithm terminates and finds the shortest feasible path between a given pair of vertices $\source\in\vertices$ and~$\target\in\vertices$ for a given initial \gls*{soc}~$\soc_\source\in[0,\maxbattery]$.
\end{theorem}

\paragraph{Computing Switching Sequences.}
In a practical implementation of the \gls*{cfp} algorithm, we need to be able to efficiently compute the switching sequences for given labels and charging functions; \cf Line~\ref{line:cfp:switchingsequence} of Algorithm~\ref{alg:cfp}.
If certain properties of the available charging functions are known, a reasonable approach is to derive the switching sequences for such functions analytically beforehand and provide a specialized implementation for them.

To give an example, we discuss the case where all charging functions in the network (and hence, all \gls*{soc} functions) are \emph{piecewise linear} and \emph{concave}, as is the case in our experimental study (see Section~\ref{sec:experiments}).
We show how a superset of the switching sequence is easily determined in this case.
Given a label $\alabel=(\triptime,\soc_\vertexa,\vertexa,\socprofile_{[\vertexa, \dots, \vertexb]})$ at a charging station~$\vertexb\in\chargingstations$, let its piecewise linear \gls*{soc} function $\socfunction{\alabel}$ be given as a sequence $\functionbreakpoints=[(\atime_1,\soc_1),\dots,(\atime_k,\soc_k)]$ of breakpoints. In other words, we have $\socfunction{\alabel}(\atime)=-\infty$ for~$\atime<\atime_1$, $\socfunction{\alabel}(\atime)=\soc_k$ for~$\atime\ge\atime_k$, and for arbitrary values $\atime_{i}\le\atime<\atime_{i+1}$, with $i\in\{1,\dots,k-1\}$, we evaluate the function by linear interpolation between the breakpoints $(\atime_{i},\soc_{i})$ and~$(\atime_{i+1},\soc_{i+1})$. Observe that we have~$\atime_1=\minfeasibletime{\socfunction{\alabel}}$.
The following Lemma~\ref{lem:cfp-pwl-switching-sequences} shows that the switching sequence of the slopes induced by $\socfunction{\alabel}$ and the charging function $\chargingfunction_\vertexb$ of $\vertexb$ is a subsequence of~$[\atime_1-\triptime,\dots,\atime_k-\triptime,\infty]$. Note that in particular, the switching sequence does not depend on the charging function $\chargingfunction_\vertexb$ at~$\vertexb$.

\begin{lemma}
\label{lem:cfp-pwl-switching-sequences}
 Given a label $\alabel$ at a vertex~$\vertex\in\chargingstations$, let its piecewise linear and concave \gls*{soc} function $\socfunction{\alabel}$ be defined by the sequence $\functionbreakpoints=[(\atime_1,\soc_1),\dots,(\atime_k,\soc_k)]$ of breakpoints. Similarly, let the charging function $\simplechargingfunction_\vertex$ of $\vertex$ be piecewise linear and concave. The switching sequence $\switchingsequence(\slopeOld{\alabel},\slopeNew{\alabel})$ is a subsequence of $\obar{\switchingsequence}(\slopeOld{\alabel},\slopeNew{\alabel}):=[\atime_1-\triptime,\dots,\atime_k-\triptime,\infty]$, \ie, $\switchingsequence(\slopeOld{\alabel},\slopeNew{\alabel})$ contains only values that are also contained in~$\obar{\switchingsequence}(\slopeOld{\alabel},\slopeNew{\alabel})$.
\end{lemma}

\begin{proof}
Since both $\socfunction{\alabel}$ and $\simplechargingfunction_\vertexb$ are piecewise linear, their corresponding slope functions $\slopeOld{\alabel}$ and $\slopeNew{\alabel}$ are \emph{piecewise constant} functions, namely, $\slopeOld{\alabel}\equiv(\soc_{i+1}-\soc_{i})/(\atime_{i+1}-\atime_{i})$ holds in the interval $[\atime_i-\triptime,\atime_{i+1}-\triptime)$ for arbitrary~$i\in\{1,\dots,k-1\}$, and an analogous statement holds for~$\slopeNew{\alabel}$.
This implies that each value of the switching sequence must correspond to a breakpoint of $\slopeOld{\alabel}$ or~$\slopeNew{\alabel}$, because both functions are constant between these breakpoints.
Moreover, since $\socfunction{\alabel}$ and $\simplechargingfunction_\vertexb$ are both concave on their subdomain with finite image, the functions $\slopeOld{\alabel}$ and $\slopeNew{\alabel}$ are \emph{decreasing}.
Consequently, the slope $\slopeNew{\alabel}$ can surpass $\slopeOld{\alabel}$ only at the breakpoints of $\slopeOld{\alabel}$, \ie, at points where the latter function decreases. The x-coordinates of these breakpoints are exactly the finite values of~$\obar{\switchingsequence}(\slopeOld{\alabel},\slopeNew{\alabel})$.
\end{proof}

Lemma~\ref{lem:cfp-pwl-switching-sequences} implies that simply spawning one new label for each breakpoint of~$\socfunction{\alabel}$ is sufficient to maintain correctness. Unnecessary labels are detected implicitly during dominance checks.
Considering the more general case where charging functions are piecewise linear (but not necessarily concave), it is not hard to see that the length of the switching sequence must be linear in the number of breakpoints of both considered functions.
Similar observations can be made for other realistic models of charging based on, \eg, exponential functions or piecewise combinations of linear and exponential functions.

%%%%%%%%%%%%%%%%%%%%%%%%%%%%%%%%%%%%%%%%%%%%%%%%%%%%%%%%%%%%%%%%%%%%%%%%%%%%%%%%
\section{A* Search}\label{sec:astar}
%%%%%%%%%%%%%%%%%%%%%%%%%%%%%%%%%%%%%%%%%%%%%%%%%%%%%%%%%%%%%%%%%%%%%%%%%%%%%%%%

To accelerate basic~\gls*{cfp}, we present techniques that extend A*~search~\cite{Har68}.
The basic idea of A* search is to use vertex potentials that guide the search towards the target.
The potential of a vertex is added to the key of a label when updating the priority queue in Line~\ref{line:cfp:updatekey} or Line~\ref{line:cfp:inserthead} of Algorithm~\ref{alg:cfp}. Thereby, vertices that are closer to the target get smaller keys.
Below, we first generalize the notion of potential consistency to our setting, before we introduce different consistent potential functions.

\paragraph{Consistency of Potentials.}
We aim at potential functions that are based on backward searches from the target vertex~$\target$, providing lower bounds on the trip time from any vertex to~$\target$.
A consistent potential function is easily obtained from a single-criterion backward search as follows. It runs Dijkstra's algorithm from~$\target$, traversing arcs in backward direction and using the cost function~$\drivingtimefunction$, which represents driving time along arcs. This yields, for each vertex~$\vertex\in\vertices$, its minimum (unconstrained) driving time to reach~$\target$. These lower bounds on the remaining trip time induce a consistent potential function on the vertices~\cite{Tun92}.
However, we can do better, exploiting that the trip time from $\vertex$ to $\target$ depends on the \gls*{soc} of a label.
(Observe that both the charging time as well as the route and hence, the driving time, of an optimal solution can change for different \gls*{soc} values.)
We propose a potential function~$\potential\colon\vertices\times[0,\maxbattery]\cup\{-\infty\}\to\posreals\cup\{\infty\}$ taking the current \gls*{soc} into account, such that~$\potential(\vertex,\soc)$ yields a lower bound on the trip time from $\vertex$ to $\target$ if the \gls*{soc} at $\vertex$ is~$\soc$, for~$\soc\in[0,\maxbattery]\cup\{-\infty\}$. We define $\potential(\vertex,-\infty):=\infty$.
We now discuss how correctness of our approach is maintained in the presence of a potential function that incorporates~\gls*{soc}.

First, we generalize the notion of consistency of a potential function. We say that a potential function $\potential\colon\vertices\times[0,\maxbattery]\cup\{-\infty\}\to\posreals\cup\{\infty\}$ is \emph{consistent} if it results in nonnegative reduced driving times, \ie, we get
\begin{align*}
 \reduceddrivingtimefunction((\vertexa, \vertexb),\soc):=\drivingtimefunction(\vertexa,\vertexb)-\potential(\vertexa,\soc)+\potential(\vertexb,\socprofile_{[\vertexa,\vertexb]}(\soc))\ge0
\end{align*}
for all arcs $(\vertexa,\vertexb)\in\arcs$ and values~$\soc\in[0,\maxbattery]$, where $\socprofile_{[\vertexa,\vertexb]}$ is the \gls*{soc} profile of~$(\vertexa,\vertexb)$.
Second, for the \gls*{soc} function $\socfunction{\alabel}$ representing a label $\alabel=(\triptime,\soc_\vertexa,\vertexa,\socprofile_{[\vertexa,\dots,\vertexb]})$ at a vertex~$\vertexb\in\vertices$, we define its \emph{consistent key} (to be used in the priority queue of \gls*{cfp}) as~$\queuekey^*(\alabel):=\min_{\atime\in\posreals}\atime+\potential(\vertexb,\socfunction{\alabel}(\atime))$.
We claim that consecutive consistent keys of labels generated after arc scans are increasing if the potential function $\potential$ is consistent. To see this, consider a label $\alabel$ at a vertex $\vertexa\in\vertices$ and the label $\alabel'$ that is created after scanning an arc~$(\vertexa,\vertexb)\in\arcs$. Recall that $\socfunction{\alabel'}(\atime+\drivingtimefunction(\vertexa,\vertex))=\socprofile_{[\vertexa,\vertexb]}(\socfunction{\alabel}(\atime))$ holds by construction of~$\alabel'$, so we can substitute $\atime':=\atime+\drivingtimefunction(\vertexa,\vertexb)$ below to get
\begin{align*}
 \queuekey^*(\alabel)&=\min_{\atime\in\posreals}\atime+\potential(\vertexa,\socfunction{\alabel}(\atime))\\
 &\le\min_{\atime\in\posreals}\atime+\drivingtimefunction(\vertexa,\vertexb)+\potential(\vertexb,\socprofile_{[\vertexa,\vertexb]}(\socfunction{\alabel}(\atime)))\\
 &=\min_{\atime'\in\posreals}\atime'+\potential(\vertexb,\socfunction{\alabel'}(\atime'))\\
 &=\queuekey^*(\alabel')\text{.}
\end{align*}
Similarly, we have to ensure that labels spawned at charging stations never have a smaller consistent key than the original label.
Consider the charging function $\chargingfunction_\vertex$ of a vertex $\vertex\in\chargingstations$. We denote by $\maxchargingspeed(\vertex)$ the \emph{maximum slope} of the charging function~$\chargingfunction_\vertex$, which in most cases equals the value $\maxchargingspeed^*(\vertex):=\max_{\atime\in\strictposreals}\{\derivative\simplechargingfunction_\vertex(\atime)/\derivative\atime\}$. (As before, we use the \emph{right} derivative to ensure that slope is well-defined for piecewise-defined functions.)
However, in the special case that $\simplechargingfunction_\vertex(0)\neq0$, we incorporate initialization time to obtain a finite slope $\mincharge_\vertex/\arrangementtime(\vertex)$ for the initial \gls*{soc} gain, presuming that $\arrangementtime(\vertex)\neq0$.
In total, we obtain the maximum slope $\maxchargingspeed(\vertex):=\max\{\maxchargingspeed^*(\vertex),\mincharge_\vertex/\arrangementtime(\vertex)\}$ of~$\vertex$.
For instance, we get $\maxchargingspeed(\vertex)=\maxbattery/\arrangementtime(\vertex)$ if $\vertex$ is a swapping station.
To ensure that labels spawned at $\vertex$ do not result in decreasing keys, we demand for the slope of the potential $\potential(\vertex,\soc)$ at $\vertex$ that
\begin{align*}
 \frac{\derivative\potential(\vertex,\soc)}{\derivative\soc}\ge-\frac{1}{\maxchargingspeed(\vertex)}\text{.}
\end{align*}
We say that the potential $\potential$ \emph{overestimates charging speed} at $\vertex$ in this case. Observe that overestimation of charging speed implies that the term $\atime+\potential(\vertex,\simplechargingfunction_\vertex(\atime))$ is \emph{increasing} with $\atime\in\posreals$, so charging at $\vertex$ does not decrease the potential.
Finally, we demand that the potential at the target is $\potential(\target,\soc)=0$ for arbitrary \gls*{soc}~$\soc\in[0,\maxbattery]$.
In summary, when using consistent keys, there are the following three requirements to a potential function~$\potential$.
\begin{enumerate}
 \item The potential function $\potential$ is consistent.
 \item Potentials at charging stations overestimate charging speed.
 \item The target vertex has a potential of~$0$ (for any finite~\gls*{soc}).
\end{enumerate}
Then, the algorithm is label setting and the correct result is obtained as soon as~$\target$ is reached (since the key at $\target$ equals trip time, so any label extracted at a later point must have a higher trip time).
As a simple example, consider the plain \gls*{cfp} algorithm described in Section~\ref{sec:approach}, which uses no potential function. This is equivalent to a potential function that evaluates to $0$ at all vertices for arbitrary~\gls*{soc}.
Clearly, the smallest feasible trip time of~$\socfunction{\alabel}$ is in fact the consistent key of a label~$\alabel$.
Moreover, observe that the potential function $\potential\equiv0$ is consistent, overestimates charging speed, and equals $0$ at the target.
In what follows, we derive more sophisticated potential functions that make the search goal directed. For each potential function, we show that the requirements listed above are fulfilled.

\paragraph{Potentials Based on Single-Criterion Search.}
We introduce our first consistent potential function.
To obtain a lower bound on the necessary charging time on the path from some vertex $\vertex\in\vertices$ to the target, let $\maxchargingspeed:=\max_{\vertex\in\chargingstations}\maxchargingspeed(\vertex)$ denote the maximum slope of \emph{any} charging function in~$\chargingstations$, \ie,~the maximum charging speed available in the network.
We define a new cost function~$\omegaweightfunction\colon\arcs\to\reals$, given as $\omegaweightfunction(\arc):=\drivingtimefunction(\arc)+(\consumptionfunction(\arc)/\maxchargingspeed)$ for all~$\arc\in\arcs$. This function adds to the driving time of every arc a lower bound on the time required for charging the energy consumed on the arc. Note that the bound can become negative, due to negative consumption values.

Given the target vertex~$\target$, prior to running~\gls*{cfp}, we perform three (single-criterion) runs of Dijkstra's algorithm from~$\target$, traversing arcs in backward direction and each using one of the cost functions $\drivingtimefunction$,~$\consumptionfunction$, and~$\omegaweightfunction$.
Thereby, we obtain, for every vertex~$\vertex\in\vertices$, the distances~$\distance_{\drivingtimefunction}(\vertex,\target)$,~$\distance_{\consumptionfunction}(\vertex,\target)$, and~$\distance_{\omegaweightfunction}(\vertex,\target)$ from $\vertex$ to $\target$ \wrt the cost functions~$\drivingtimefunction$,~$\consumptionfunction$, and~$\omegaweightfunction$, respectively.
Note that computation of~$\distance_{\consumptionfunction}(\vertex,\target)$ and~$\distance_{\omegaweightfunction}(\vertex,\target)$ is label correcting, due to negative weights. We can apply potential shifting~\cite{Joh77} to remedy this issue, but the effect on overall running time is negligible in practice.
Using the obtained distances, the potential function $\omegapotential\colon\vertices\times[0,\maxbattery]\cup\{-\infty\}\to\posreals\cup\{\infty\}$ is defined by
\begin{align}
 \label{eq:cfp:omegapotential}
 \omegapotential(\vertex,\soc):=
 \begin{cases}
   \distance_{\drivingtimefunction}(\vertex,\target)                                 & \mbox{if } \soc \geq \distance_{\consumptionfunction}(\vertex,\target)\text{,} \\
   \distance_{\omegaweightfunction}(\vertex,\target)-\frac{\soc}{\maxchargingspeed}  & \mbox{otherwise,}
 \end{cases} 
\end{align}
for all $\vertex\in\vertices$ and~$\soc\in[0,\maxbattery]$.
It uses the minimum (unrestricted) driving time as lower bound on the remaining trip time in case that the current \gls*{soc} is greater or equal to the minimum energy required to reach the target without recharging.
Otherwise, we know that we have to spend some additional time for charging on the way to the target~$\target$, so we use a bound induced by the weight function~$\omegaweightfunction$.
Note that we can discard labels whose consumption exceeds $\maxbattery$ in the search that computes $\distance_{\consumptionfunction}(\vertex,\target)$ to save some time in practice.
Lemma~\ref{lem:omega-feasible-potential} formally proves that the potential function~$\omegapotential$ defined above is consistent.

\begin{lemma}
\label{lem:omega-feasible-potential}
The potential function $\omegapotential$ is consistent.
\end{lemma}

\begin{proof}
To prove the claim, we show that reduced costs $\reduceddrivingtimefunction(\cdot,\cdot)$ are nonnegative, \ie, the inequality $\reduceddrivingtimefunction((\vertexa,\vertexb),\soc)=\drivingtimefunction(\vertexa,\vertexb)-\omegapotential(\vertexa,\soc)+\omegapotential(\vertexb,\socprofile_{[\vertexa, \vertexb]}(\soc))\ge0$ holds for all $(\vertexa,\vertexb)\in\arcs$ and~$\soc\in[0,\maxbattery]$.
Consider an arbitrary arc $(\vertexa,\vertexb)\in\arcs$ and assume that the \gls*{soc} at $\vertexa$ is~$\soc\in[0,\maxbattery]$.
We distinguish four cases.

\case{Case 1} $\soc<\distance_{\consumptionfunction}(\vertexa, \target)$ and $\socprofile_{[\vertexa, \vertexb]}(\soc)<\distance_{\consumptionfunction}(\vertexb,\target)$.
The claim follows after a few simple substitutions. We can make use of the fact that $\soc-\consumptionfunction(\vertexa,\vertexb)\ge\socprofile_{[\vertexa, \vertexb]}(\soc)$ holds for all~$\soc\in[0, \maxbattery]$. Consistency then follows directly from the triangle inequality, after the following steps
\begin{align*}
\reduceddrivingtimefunction((\vertexa, \vertexb),\soc) &= \drivingtimefunction(\vertexa, \vertexb)-\omegapotential(\vertexa, \soc) + \omegapotential(\vertexb,\socprofile_{[\vertexa, \vertexb]}(\soc)) \\
 &=\drivingtimefunction(\vertexa, \vertexb)-\distance_{\omegaweightfunction}(\vertexa,\target)+\frac{\soc}{\maxchargingspeed}+\distance_{\omegaweightfunction}(\vertexb, \target)-\frac{\socprofile_{[\vertexa,\vertexb]}(\soc)}{\maxchargingspeed} \\
 &\ge\drivingtimefunction(\vertexa, \vertexb)+\frac{\consumptionfunction(\vertexa,\vertexb)}{\maxchargingspeed}-\distance_{\omegaweightfunction}(\vertexa,\target)+\distance_{\omegaweightfunction}(\vertexb,\target) \\
 &=\omegaweightfunction(\vertexa,\vertexb)-\distance_{\omegaweightfunction}(\vertexa,\target)+\distance_{\omegaweightfunction}(\vertexb,\target)\\
 &\ge0\text{.}
\end{align*}

\case{Case 2} $\soc<\distance_{\consumptionfunction}(\vertexa,\target)$ and $\socprofile_{[\vertexa,\vertexb]}(\soc)\ge\distance_{\consumptionfunction}(\vertexb,\target)$. We make use of both preconditions together with the fact that $\soc-\consumptionfunction(\vertexa,\vertexb)\ge\socprofile_{[\vertexa,\vertexb]}(\soc)$ holds for all~$\soc\in[0,\maxbattery]$ to obtain the inequalities
\begin{align*}
 \soc\ge\socprofile_{[\vertexa, \vertexb]}(\soc)+\consumptionfunction(\vertexa,\vertexb)\ge\distance_{\consumptionfunction}(\vertexb, \target)+\consumptionfunction(\vertexa, \vertexb)\ge\distance_{\consumptionfunction}(\vertexa, \target)>\soc\text{,}
\end{align*}
which yield a contradiction. Hence, this case cannot occur.

\case{Case 3} $\soc\ge\distance_{\consumptionfunction}(\vertexa,\target)$ and $\socprofile_{[\vertexa, \vertexb]}(\soc)<\distance_{\consumptionfunction}(\vertexb,\target)$.
We know that, due to the triangle inequality, $\drivingtimefunction(\vertexa, \vertexb)+\distance_{\drivingtimefunction}(\vertexb,\target)\ge\distance_{\drivingtimefunction}(\vertexa,\target)$ holds. Moreover, $\distance_{\consumptionfunction}(\vertexb,\target)-\socprofile_{[\vertexa, \vertexb]}(\soc)>0$ holds by assumption.
We can further exploit that $\distance_\omegaweightfunction(\vertexb,\target)$ is at least the sum of $\distance_{\drivingtimefunction}(\vertexb,\target)$ and~$\distance_{\consumptionfunction}(\vertexb,\target)/\maxchargingspeed$ (note that it is possibly greater if the shortest $\vertexa$--$\target$~paths in the graph \wrt the cost functions $\drivingtimefunction$ and $\consumptionfunction$ differ).
After some substitutions, we thus get
\begin{align*}
\reduceddrivingtimefunction((\vertexa, \vertexb),\soc)&=\drivingtimefunction(\vertexa, \vertexb)-\omegapotential(\vertexa,\soc)+\omegapotential(\vertexb,\socprofile_{[\vertexa, \vertexb]}(\soc))\\
 &=\drivingtimefunction(\vertexa, \vertexb)-\distance_{\drivingtimefunction}(\vertexa,\target)+\distance_{\omegaweightfunction}(\vertexb,\target)-\frac{\socprofile_{[\vertexa,\vertexb]}(\soc)}{\maxchargingspeed}\\
 &\ge\drivingtimefunction(\vertexa, \vertexb)-\distance_{\drivingtimefunction}(\vertexa, \target)+\distance_{\drivingtimefunction}(\vertexb,\target)+\frac{\distance_{\consumptionfunction}(\vertexb,\target)}{\maxchargingspeed}-\frac{\socprofile_{[\vertexa, \vertexb]}(\soc)}{\maxchargingspeed}\\
 &\ge0\text{.}
\end{align*}

\case{Case 4} $\soc\ge\distance_{\consumptionfunction}(\vertexa,\target)$ and $\socprofile_{[\vertexa, \vertexb]}(\soc)\ge\distance_{\consumptionfunction}(\vertexb,\target)$. This case is trivial; feasibility follows directly from the triangle inequality.
\end{proof}

The potential function $\omegapotential$ always evaluates to $0$ at the target and overestimates charging speed by construction.
As for consistent keys, defined for a label~$\alabel$ as~$\queuekey^*(\alabel)=\min_{\atime\in\posreals}\atime+\potential(\vertexb,\socfunction{\alabel}(\atime))$, observe that the corresponding terms $\atime+\distance_{\drivingtimefunction}(\vertex,\target)$ and $\atime+\distance_{\omegaweightfunction}(\vertex,\target)-(\socfunction{\alabel}(\atime)/\maxchargingspeed)$ in Equation~\ref{eq:cfp:omegapotential} increase with $\atime$ (the term $\socfunction{\alabel}(\atime)/\maxchargingspeed$ increases with a slope of at most~1).
Thus, we obtain the consistent key for the label $\alabel$ by computing the value of $\atime+\omegapotential(\vertex,\socfunction{\alabel}(\atime))$ at the minimum feasible trip time $\atime_1:=\minfeasibletime{\socfunction{\alabel}}$ of $\socfunction{\alabel}$ and at the minimum trip time $\atime_2$ with $\socfunction{\alabel}(\atime_2)\ge\distance_{\consumptionfunction}(\vertex,\target)$, if it exists.
The minimum of both values yields a consistent key, given as~$\queuekey^*(\alabel)=\min\{\atime_1,\atime_2\}$.
Together with Lemma~\ref{lem:omega-feasible-potential}, this implies correctness of \gls*{cfp} when applying potential shifting with the function~$\omegapotential$.

\paragraph{Potentials Based on Bound Function Propagation.}
Even though the potential function $\omegapotential$ incorporates~\gls*{soc}, lower bounds may be too conservative in that they presume recharging is possible at any time and with the best charging rate.
We attempt to be more precise, while keeping computational effort limited, by explicitly constructing lower bound functions.
Again, we run (at query time) a label-correcting backward search from the given target $\target$, but this time computing for each vertex $\vertex\in\vertices$ a \emph{piecewise linear function} $\convexlowerboundfunction\colon\reals\to\posreals\cup\{\infty\}$ mapping \gls*{soc} to a lower bound on the trip time from $\vertex$ to~$\target$. These piecewise linear functions are represented by sequences $\convexfunctionbreakpoints=[(\soc_1,\atime_1),\dots,(\soc_k,\atime_k)]$ of breakpoints such that $\convexlowerboundfunction(\soc)=\infty$ for~$\soc<\soc_1$, $\convexlowerboundfunction(\soc)=\atime_k$ for~$\soc\ge\soc_k$, and the function is evaluated by linear interpolation between breakpoints.
If the sequence $\convexfunctionbreakpoints$ of breakpoints is empty, denoted~$\convexfunctionbreakpoints=\emptyset$, we obtain $\convexlowerboundfunction\equiv\infty$.
During the backward search, each vertex stores a \emph{single} label consisting of such a piecewise linear function.
To simplify the search, we ignore battery constraints. Hence, domains of bounds are not restricted to~$[0,\maxbattery]$ and we compute (possibly negative) lower bounds on the \gls*{soc} necessary to reach~$\target$.
However, we maintain the invariant that all functions are \emph{decreasing} and \emph{convex} on the interval~$[\soc_1,\infty)$. This greatly simplifies label updates and the computation of functions, which we describe in detail below.

\begin{algorithm}[t]
  \caption{Pseudocode of the function-propagating potential search for \gls*{cfp}.}%
  \label{alg:cfp-potential-backward-search}%
  \BlankLine
  \tcp{initialize labels}%
  \ForEach{$\vertex\in\vertices$}{\label{line:cfp-potential-search-initialize-begin}%
    $\convexlowerboundfunction_\vertex\assign\emptyset$\;%
  }%
  $\convexlowerboundfunction_\target\assign[(0,0)]$\;%
  $\queue$.\queueInsert{$\target,0$}\;\label{line:cfp-potential-search-initialize-end}%
  \BlankLine
  \tcp{run main loop}%
  \While{$\queue$.\queueIsNotEmpty{}}{\label{line:cfp-potential-search-main-begin}%
    \BlankLine
    $\vertexb\assign\queue$.\queueDeleteMin{}\;%
    \BlankLine
    \tcp{handle charging stations}%
    \If{$\vertexb\in\chargingstations$}{%
      $\convexlowerboundfunction_\vertexb\assign\extendop(\convexlowerboundfunction_\vertexb,\simplechargingfunction_\vertexb)$\;\label{line:cfp-potential-search-link-cs}
    }%
    \BlankLine
    \tcp{scan outgoing arcs}%
    \ForEach{$(\vertexa,\vertexb)\in\arcs$}{%
      $\convexlowerboundfunction\assign\shiftop(\convexlowerboundfunction_\vertexb,[(\consumptionfunction(\vertexa,\vertexb),\drivingtimefunction(\vertexa,\vertexb))])$\;\label{line:cfp-potential-search-shift}%
      \If{$\exists \atime\in\reals\colon\convexlowerboundfunction(\atime)<\convexlowerboundfunction_\vertexa(\atime)$}{%
        $\convexlowerboundfunction_\vertexa\assign\mergeop(\convexlowerboundfunction_\vertexa,\convexlowerboundfunction)$\;\label{line:cfp-potential-search-merge}%
        $\queue$.\queueUpdate{$\vertexa,\queuekey(\convexlowerboundfunction_\vertexa)$}\;\label{line:cfp-potential-search-main-end}%
      }%
    }%
  }%
\end{algorithm}

The algorithm resembles (label-correcting) profile search~\cite{Del09b,Sch15} and its pseudocode is outlined in Algorithm~\ref{alg:cfp-potential-backward-search}. It is initialized in Lines~\ref{line:cfp-potential-search-initialize-begin}--\ref{line:cfp-potential-search-initialize-end} with a function~$\convexlowerboundfunction_\target$, represented by a single breakpoint $\convexfunctionbreakpoints_\target=[(0,0)]$ at the target vertex~$\target$, \ie, $\convexlowerboundfunction_\target(\soc)$ evaluates to $0$ for arbitrary values~$\soc\in\posreals$. All other labels are empty, so their functions always evaluate to~$\infty$. Each step of the algorithm's main loop (Lines~\ref{line:cfp-potential-search-main-begin}--\ref{line:cfp-potential-search-main-end}) scans a vertex with minimum key
(the key of a vertex $\vertex\in\vertices$ is the minimum function value $\min_{\soc\in\reals}\convexlowerboundfunction_\vertex(\soc)$ of its label~$\convexlowerboundfunction_\vertex$).

Whenever the search reaches a charging station~$\vertexb\in\chargingstations$, we have to ensure that the possibility of recharging is reflected in the label of $\vertexb$ and that the function overestimates charging speed. To this end, we \emph{extend} the (tentative) lower bound function $\convexlowerboundfunction_\vertexb$ with the charging function~$\simplechargingfunction_\vertexb$ (Line~\ref{line:cfp-potential-search-link-cs} of Algorithm~\ref{alg:cfp-potential-backward-search}). The result of this operation is a new (tentative) lower bound on the trip time that incorporates recharging at $\vertexb$ and has a slope of at least~$-1/\maxchargingspeed(\vertexb)$ (on its subdomain with finite image).
Assume we are given a piecewise linear, convex function $\convexlowerboundfunction_\vertexb$ at $\vertexb$ with breakpoints $\convexfunctionbreakpoints_\vertexb=[(\soc_1,\atime_1),\dots,(\soc_k,\atime_k)]$ that maps \gls*{soc} to trip time without recharging at~$\vertexb$. We obtain the result~$\convexlowerboundfunction$ of extending $\convexlowerboundfunction_\vertexb$ with $\simplechargingfunction_\vertexb$ as follows; see Figure~\ref{fig:cfp-backward-search-linking} for an example.
In accordance with our requirements for the potential at $\vertexb$, we approximate $\simplechargingfunction_\vertexb$ with a lower bound given by a single segment with slope~$1/\maxchargingspeed(\vertexb)$; see Figure~\ref{fig:cfp-backward-search-linking:a}.
Distributing (lower bounds on) charging time among~$\simplechargingfunction_\vertexb$ and stations represented by $\convexlowerboundfunction_\vertexb$ then corresponds to shifting this segment along the y-axis such that it intersects the graph of~$\convexlowerboundfunction_\vertexb$; see Figure~\ref{fig:cfp-backward-search-linking:b}.
To find a lower bound on the best possible distribution, let~$i\in\{2,\dots,k-1\}$ be the unique index (if it exists) such that
\begin{align*}
 \frac{\atime_{i} - \atime_{i-1}}{\soc_{i} - \soc_{i-1}} \le -\frac{1}{\maxchargingspeed(\vertexb)} < \frac{\atime_{i+1} - \atime_{i}}{\soc_{i+1} - \soc_{i}} \text{,}
\end{align*}
\ie, the (negative, inverse) maximum slope of $\simplechargingfunction_\vertexb$ is at least the slope of the segment from $\soc_{i-1}$ to~$\soc_{i}$, but lower than the slopes of all subsequent segments.
We set $i:=1$ if the maximum slope of $\simplechargingfunction_\vertexb$ is below the slope of the first segment (from $\soc_{1}$ to~$\soc_{2}$), and $i:=k$ if the maximum slope of $\simplechargingfunction_\vertexb$ is at least the slope of the last segment (from $\soc_{k-1}$ to~$\soc_{k}$).
Then, if $\soc_i\le0$, the function $\convexlowerboundfunction_\vertexb$ remains unchanged, \ie,~$\convexlowerboundfunction=\convexlowerboundfunction_\vertexb$, as charging at $\vertexb$ cannot decrease the lower bound in this case.
Otherwise, $\convexlowerboundfunction$ is defined by the sequence~$\convexfunctionbreakpoints:=[(0,\atime_i+\soc_i/\maxchargingspeed(\vertexb)),(\soc_i,\atime_i),\dots,(\soc_k,\atime_k)]$; see Figure~\ref{fig:cfp-backward-search-linking:b}.
The first segment of this function corresponds to an estimate of the time to charge to an \gls*{soc} of $\soc_i$ plus the remaining trip time to~$\target$. 
By construction, $\convexlowerboundfunction$ is convex and overestimates charging speed (provided that the same holds for~$\convexlowerboundfunction_\vertexb$).

\begin{figure}[t]
 \centering%
 \begin{subfigure}[b]{.5\textwidth}%
  \centering%
  \begin{tikzpicture}[figure,scale=\ChargingFunctionExampleScale]
   \ExampleDrawCoordinateSystemWithTicks{1}{5}{0}{5}{0}{0}
 
   % Axis names.
   \PlotXAxisName{5}{0}{$\soc$}
   \PlotYAxisName{0}{5}{$\simplechargingfunctionInv_\vertexb(\soc)$}
 
   % Draw function.
   \begin{scope}
    \begin{pgfinterruptboundingbox}
     \clip (0,0) rectangle (5,5);
     \draw[chargingfunction,secondarycolor-light,dashed] (3,2) -- (6,4);
     \draw[chargingfunction,secondarycolor-dark] (0,0) -- node [pos=0.9,below=2pt,sloped,edgelabel,text=secondarycolor] {$\frac{1}{\maxchargingspeed(\vertexb)}$} (3,2) -- (5,5);
    \end{pgfinterruptboundingbox}
   \end{scope}
  \end{tikzpicture}
  \caption{}%
  \label{fig:cfp-backward-search-linking:a}%
 \end{subfigure}%
 \begin{subfigure}[b]{.5\textwidth}%
  \centering%
  \begin{tikzpicture}[figure,scale=\ChargingFunctionExampleScale]
   \ExampleDrawCoordinateSystemWithPositiveInfty{1}{5}{0}{5}{0}{0}

   % Draw parallel lines.
   \begin{scope}
    \begin{pgfinterruptboundingbox}
     \clip (0,5) -- (5,5) -- (5,1.5) -- (0,2.75) -- cycle;
     \draw[black30,line cap=rect] (0,5) -- (6,1);
     \draw[black30,line cap=rect] (0,6) -- (6,2);
     \draw[black30,line cap=rect] (0,7) -- (6,3);
     \draw[black30,line cap=rect] (0,8) -- (6,4);
    \end{pgfinterruptboundingbox}
   \end{scope}
   
   % Axis names.
   \PlotXAxisName{5}{0}{$\soc$}
   \PlotYAxisName{0}{5}{$\convexlowerboundfunction_\vertexb(\soc)$}
   
   % Draw function.
   \begin{scope}
    \begin{pgfinterruptboundingbox}
     \clip (0,0) rectangle (5,6+\ExampleInftyOffset);
     \draw[function,primarycolor-light,dashed] (0,5+\ExampleInftyOffset) -- (2,5+\ExampleInftyOffset);
     \draw[function,primarycolor-light,dashed] (2,4) -- (3,2);
     \draw[function,primarycolor-dark] (0,4) -- node [pos=0.5,below=2pt,sloped,edgelabel] {$-\frac{1}{\maxchargingspeed(\vertexb)}$} (3,2)  -- (5,1.5);
    \end{pgfinterruptboundingbox}
   \end{scope}

   % Markers at discontinuities.
   \begin{pgfonlayer}{foreground}
    \node[discontinuityblank,draw=primarycolor-light] (d1) at (2,5+\ExampleInftyOffset) {};
    \node[discontinuityfilled,primarycolor-light] (d2) at (2,4) {};
   \end{pgfonlayer}
  \end{tikzpicture}
  \caption{}%
  \label{fig:cfp-backward-search-linking:b}%
 \end{subfigure}%
 \caption{Extending functions in the function propagating backward search. (a) The (expanded) inverse charging function $\simplechargingfunctionInv_\vertexb$ of a charging station $\vertexb\in\chargingstations$ maps \gls*{soc} to charging time. The dashed line indicates its (inverse) maximum slope, which we use to approximate the (inverse) charging function with a single segment. (b) The lower bound function $\convexlowerboundfunction_\vertexb$ (dashed) and the result $\convexlowerboundfunction$ of extending it with the charging function~$\simplechargingfunction_\vertexb$ (dark blue)}%
\label{fig:cfp-backward-search-linking}%
\end{figure}
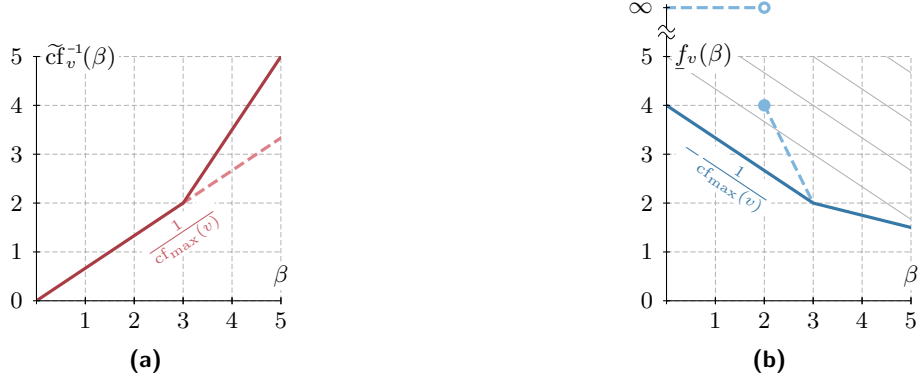

Since we ignore battery constraints, scanning an incoming arc $(\vertexa,\vertexb)\in\arcs$ of $\vertexb$ boils down to shifting all breakpoints of the current function $\convexlowerboundfunction_\vertexb$ by $\consumptionfunction(\vertexa,\vertexb)$ and $\drivingtimefunction(\vertexa,\vertexb)$ on the x- and y-axis, respectively. More formally, given the breakpoints $\convexfunctionbreakpoints_\vertexb=[(\soc_1,\atime_1),\dots,(\soc_k,\atime_k)]$ of~$\convexlowerboundfunction_\vertexb$, we compute a function $\convexlowerboundfunction$ with
\begin{align*}
 \convexfunctionbreakpoints:=[(\soc_1+\consumptionfunction(\vertexa,\vertexb),\atime_1+\drivingtimefunction(\vertexa,\vertexb)),\dots,(\soc_k+\consumptionfunction(\vertexa,\vertexb),\atime_k+\drivingtimefunction(\vertexa,\vertexb))]\text{.}
\end{align*}
Then, we check whether the function $\convexlowerboundfunction$ is smaller than the function $\convexlowerboundfunction_\vertexa$ in the label of $\vertexa$ for at least one~$\soc\in\reals$.
If this is the case, we \emph{merge} $\convexlowerboundfunction$ and~$\convexlowerboundfunction_\vertexa$, \ie, we compute the function defined as their pointwise minimum $\convexlowerboundfunction_\vertexa(\soc):=\min\{\convexlowerboundfunction(\soc),\convexlowerboundfunction_\vertexa(\soc)\}$ for all~$\soc\in\reals$.
This operation requires a linear-time scan over the breakpoints of both involved functions, similar to label-correcting profile searches in time-dependent route planning~\cite{Bat13,Bau16b,Del09b}.
The resulting function $\convexlowerboundfunction_\vertexb$ is again piecewise linear, but may no longer be convex. Therefore, we compute, during each merge operation, the \emph{convex lower hull} of the result using Graham's scan~\cite{Gra72}. While (slightly) deteriorating the quality of the bound, this reduces the number of breakpoints and simplifies handling of charging stations.
We obtain a convex function, which is stored in the label of~$\vertexb$.
The vertex $\vertexb$ is also updated in the priority queue.

Again, it is easy to see that the resulting potential function $\convexpotential\colon\vertices\times[0,\maxbattery]\cup\{-\infty\}\to\posreals\cup\{\infty\}$ is consistent when using the computed bounds by setting $\convexpotential(\vertex,\soc):=\convexlowerboundfunction_\vertex(\soc)$ for $\vertex\in\vertices$ and~$\soc\in[0,\maxbattery]$.
Lemma~\ref{lem:convex-feasible-potential} proves this formally.
The consistent key $\queuekey^*(\alabel)=\min_{\atime\ge0}\atime+\convexpotential(\vertex,\socfunction{\alabel}(\atime))$ of a label $\alabel$ at some vertex~$\vertex\in\vertices$ is computed in a linear scan over $\socfunction{\alabel}$ and the breakpoints of the piecewise linear function~$\convexlowerboundfunction_\vertex$.

\begin{lemma}
\label{lem:convex-feasible-potential}
The potential function $\convexpotential$ is consistent.
\end{lemma}

\begin{proof}
For an arbitrary arc~$(\vertexa,\vertexb)\in\arcs$, consider the piecewise linear functions $\convexlowerboundfunction_\vertexa$ and $\convexlowerboundfunction_\vertexb$ at $\vertexa$ and~$\vertexb$, respectively, after the backward search has terminated.
We show that the reduced costs $\reduceddrivingtimefunction((\vertexa,\vertexb),\soc)$ are nonnegative for all~$\soc\in[0,\maxbattery]$.
We know that $\convexlowerboundfunction_\vertexa$ is upper bounded by the result of shifting~$\convexlowerboundfunction_\vertexb$ by the costs $\consumptionfunction(\vertexa, \vertexb)$ and $\drivingtimefunction(\vertexa, \vertexb)$ of the arc $(\vertexa, \vertexb)$ traversed in backward direction, since this function was merged with $\convexlowerboundfunction_\vertexa$ during the search.
This implies that $\convexpotential(\vertexb,\soc-\consumptionfunction(\vertexa, \vertexb))+\drivingtimefunction(\vertexa,\vertexb)\ge\convexpotential(\vertexa,\soc)$ holds for arbitrary $\soc\in[0,\maxbattery]$.
Moreover, $\socprofile_{[\vertexa,\vertexb]}(\soc)$ is a lower bound on $\soc-\consumptionfunction(\vertexa,\vertexb)$, so we have $\convexpotential(\vertexb,\soc-\consumptionfunction(\vertexa, \vertexb))\le\convexpotential(\vertexb,\socprofile_{[\vertexa, \vertexb]}(\soc))$ because $\convexpotential$ decreases with increasing~\gls*{soc}. Plugging this into the above inequality, we obtain $\drivingtimefunction(\vertexa, \vertexb)-\convexpotential(\vertexa, \soc)+\convexpotential(\vertexb,\socprofile_{[\vertexa, \vertexb]}(\soc))\ge0$ for all $\soc\in[0,\maxbattery]$, which proves the claim.
\end{proof}

\paragraph{Potential Search on Demand.}
Computing the potential function for every vertex in the graph is wasteful for short-range queries.
To speed up such queries, we run the backward searches that compute vertex potentials \emph{on demand}: Whenever the \gls*{cfp} search requires the potential of some vertex $\vertex\in\vertices$ that was not scanned by the backward search yet, the backward search is executed until $\vertex$ is reached. It is then suspended and only resumed if the potential of another vertex is required that the backward search has not visited.
This procedure yields consistent potentials if the backward search is label setting (otherwise, there is no guarantee that a lower bound was computed when a vertex is scanned for the first time). In case of the potential function $\omegapotential$, this can be ensured by applying Johnson's shifting technique~\cite{Joh77} to its backward searches.
For the potential function~$\convexpotential$, however, the function-propagating search is label correcting, so computing $\convexpotential$ on demand becomes more involved. We describe modifications to the search and the lower bounds to ensure that $\convexpotential$ is indeed a consistent potential, even if the search is suspended before it terminates.

First, we can ensure that the minimum key in the priority queue of the backward search is nondecreasing during the course of the algorithm: After scanning an arc $(\vertexa,\vertexb)\in\arcs$, consider two functions~$\convexlowerboundfunction$ and $\convexlowerboundfunction_\vertexa$ corresponding to the result of scanning the arc $(\vertexa,\vertexb)$ in backward direction and the current label at~$\vertexa$, respectively, before merging these two functions (see Line~\ref{line:cfp-potential-search-merge} in Algorithm~\ref{alg:cfp-potential-backward-search}). Let $\convexlowerboundfunction_\vertexa^*$ denote the result after merging. Since $\convexlowerboundfunction_\vertexa^*$ is the convex lower hull of the minimum of $\convexlowerboundfunction$ and~$\convexlowerboundfunction_\vertexa$, every breakpoint in $\convexlowerboundfunction_\vertexa^*$ must also be contained in $\convexlowerboundfunction$ or~$\convexlowerboundfunction_\vertexa$. Let $(\soc,\atime)$ be the breakpoint with minimum trip time $\atime\in\posreals$ contained in the corresponding sequence $\ubar{\convexfunctionbreakpoints}_\vertexa^*$ of~$\convexlowerboundfunction_\vertexa^*$, but not in the sequence $\ubar{\convexfunctionbreakpoints}_\vertexa$ of~$\convexlowerboundfunction_\vertexa$, \ie, $(\soc,\atime)\in\ubar{\convexfunctionbreakpoints}_\vertexa^*$ and $(\soc,\atime)\notin\ubar{\convexfunctionbreakpoints}_\vertexa$. If the result of merging improves the label at~$\vertexa$, such a point must exist.
We set the key of $\vertexa$ in the priority queue to the minimum of $\atime$ and its current key.
Since driving time is nonnegative, scanning an arc may only increase the driving time of newly added breakpoints.
As a result, propagating breakpoints never decreases the minimum key in the priority queue.

Assume that the backward search is suspended at some point and let $\atime^*\in\posreals$ be current minimum key of the priority queue. For each vertex~$\vertex\in\vertices$, consider its current label~$\convexlowerboundfunction_\vertex$. Let $\convexlowerboundfunction_\vertex^*$ denote the function obtained after applying Graham's scan to the result of merging $\convexlowerboundfunction_\vertex$ and the function induced by the single breakpoint~$[(0,\atime^*)]$. We claim that the potential function $\convexpotential^*$ with $\convexpotential^*(\vertex,\soc):=\convexlowerboundfunction_\vertex^*(\soc)$ for all $\vertex\in\vertices$ and $\soc\in[0,\maxbattery]$ is consistent.
To see this, consider a (multi-)graph $\graph'=(\vertices,\arcs')$ constructed from the input graph $\graph$ by adding, for every~$\vertex\in\vertices$, an arc $(\vertex,\target)$ with driving time $\drivingtimefunction(\vertex,\target):=\atime^*$ and consumption~$\consumptionfunction(\vertex,\target):=0$. It is easy to verify that the potential function $\convexpotential^*$ on $\graph$ is equivalent to the potential function $\convexpotential$ on~$\graph'$.
Hence, $\convexpotential^*$ is a consistent potential function for~$\graph'$. Observe that this immediately implies that $\convexpotential^*$ is also consistent on~$\graph$, since reduced arc costs must be nonnegative for the subset $\arcs\subseteq\arcs'$.
Lemma~\ref{lem:potentialsondemand} follows immediately from this observation and Lemma~\ref{lem:convex-feasible-potential}.

\begin{lemma}
 \label{lem:potentialsondemand}
 The potential function $\convexpotential^*$ is consistent.
\end{lemma}

Note that we have to keep the original function $\convexlowerboundfunction_\vertex$ after suspending the backward search, in case it is resumed later.
Hence, we do not store $\convexpotential^*$ explicitly, but perform the necessary merge operation and Graham's scan on demand when the potential is requested.
Given that keys of labels are consistent and by Lemma~\ref{lem:omega-feasible-potential}, Lemma~\ref{lem:convex-feasible-potential}, and Lemma~\ref{lem:potentialsondemand}, we obtain Theorem~\ref{thm:cfp-feasible-potentials}, which summarizes our findings on A*~search.

\begin{theorem}
 \label{thm:cfp-feasible-potentials}
 The \gls*{cfp} algorithm computes the correct output when using either of the potential functions~$\omegapotential$,~$\convexpotential$, or~$\convexpotential^*$.
\end{theorem}

\paragraph{Implementation Details.}
When using potentials on demand, we can suspend the backward search at any time and derive lower bounds for \gls*{cfp}. However, bound quality of $\convexpotential^*$ may deteriorate if the search is suspended too early, as it depends on the current minimum key in the priority queue of the backward search.
Therefore, we do not abort the search immediately when a vertex $\vertex\in\vertices$ in question is scanned for the first time, but continue until the minimum key $\atime^*$ in the priority queue is significantly greater than the key induced by the first breakpoint $(\soc_1,\atime_1)$ of $\convexlowerboundfunction_\vertex$ (in our experiments, we suspend the search if $\atime^*>\min\{2\atime_1,\atime_1+3\,600\}$, where time is measured in seconds).

%%%%%%%%%%%%%%%%%%%%%%%%%%%%%%%%%%%%%%%%%%%%%%%%%%%%%%%%%%%%%%%%%%%%%%%%%%%%%%%%
\section{Contraction Hierarchies}\label{sec:ch}
%%%%%%%%%%%%%%%%%%%%%%%%%%%%%%%%%%%%%%%%%%%%%%%%%%%%%%%%%%%%%%%%%%%%%%%%%%%%%%%%

During an offline step, \gls*{ch} preprocessing~\cite{Gei12b} iteratively contracts the vertices of the input graph and adds \emph{shortcuts} in the remaining graph to retain correct distances. These shortcuts then help reducing the search space in online queries.
When adapting \gls*{ch} to our scenario, vertex contraction becomes more expensive, as each shortcut represents a pair consisting of driving time and an \gls*{soc} profile. Moreover, we need a shortcut for every nondominated path. Hence, the resulting search graph may contain multi-arcs.

We compute a \emph{partial}~\gls*{ch}, \ie, we contract only some vertices~(the \emph{component}), leaving an uncontracted \emph{core} graph---a common approach in complex scenarios~\cite{Dib15,Fun13,Sto12c}.
As in Storandt~\cite{Sto12c}, we keep all charging stations in the graph uncontracted.
Thus, complexity induced by charging stations only is contained within the core (simplifying the search in the component).
Shortcuts store the driving time and the \gls*{soc} profile (represented by three values as described in Section~\ref{sec:problem}) of the path they represent.

\paragraph{Witness Search.}
During preprocessing, we perform \emph{witness searches} when contracting a vertex, to test whether all shortcut candidates are necessary to maintain distances in the current overlay. Given a shortcut candidate $(\vertexa,\vertexb)$ with $\vertexa\in\vertices$ and~$\vertexb\in\vertices$, we run a variant of the \gls*{bsp} algorithm that propagates labels consisting of the driving time and the \gls*{soc} profile of a path (represented by three values), starting from~$\vertexa$.
A label \emph{dominates} another label in this search if its driving time is smaller or equal to the driving time of the other label and its \gls*{soc} profile dominates that of the other label.
Keys in the priority queue follow a lexicographic order of the labels. The witness search stops if either the shortcut candidate is dominated by a label at $\vertexb$ or the minimum key in the priority queue exceeds the key induced by the shortcut candidate.

In order to reduce preprocessing time, we simplify the witness search as follows.
First, we only search for single witnesses that dominate a shortcut candidate. In other words, we only perform pairwise comparisons between labels at the head of a shortcut candidate and the candidate itself. Thereby, we might insert an unnecessary shortcut whose \gls*{soc} profile is dominated by the upper envelope of multiple \gls*{soc} profiles corresponding to labels with lower driving time.
Second, during the witness search, we limit the number of labels per vertex to a small constant (10 in our experiments).
Whenever this size is exceeded, we identify (in a linear scan over the sorted labels) the pair of labels that has the minimum difference in terms of driving time. Of these two labels, we remove the one with smaller difference to its next closest label (in order to keep the gap between the remaining labels small).
Finally, following Geisberger~et~al.~\cite{Gei12b}, we prune the search after a fixed hop limit (20 in our experiments).
Taking these measures, we may possibly insert unnecessary shortcuts. Thus, queries may slow down slightly, but correctness is not affected.

\paragraph{Queries.}
Since we compute a partial~\gls*{ch}, the query algorithm consists of two phases. Given a source~$\source\in\vertices$, a target~$\target\in\vertices$, and the initial \gls*{soc}~$\soc\in[0,\maxbattery]$, the first phase runs a backward \gls*{ch} search from~$\target$, scanning only upward arcs \wrt the vertex order. This search operates on the component, so it is pruned at core vertices (\ie, core vertices are not inserted into the priority queue). As the component contains no charging stations, a basic variant of the \gls*{bsp} algorithm suffices.
Note, however, that the \gls*{soc} at~$\target$ is yet unknown and therefore, the search algorithm computes \gls*{soc} profiles instead of \gls*{soc} values (as in the witness search).
For every nondominated label at any vertex visited by the search, we add a (temporary) shortcut from this vertex to the target.
The second phase runs \gls*{cfp} from~$\source$ and is restricted to upward arcs, core arcs, and the temporary arcs added by the backward search.

\paragraph{Implementation Details.}
During preprocessing of~\gls*{ch}, the next vertex in the contraction order is determined from the measures~\gls*{ed},~\gls*{dn}, and~\gls*{cq} according to Geisberger~et~al.~\cite{Gei12b}.
The priority of a vertex (higher priority means higher importance) is then set to $64 \operatorname{\gls*{ed}}+\operatorname{\gls*{dn}}+\operatorname{\gls*{cq}}$.
To reduce the number of necessary witness searches, we cache shortcuts computed during the computation of~\gls*{ed}. This requires a simulated contraction; see Geisberger~et~al.~\cite{Gei12b}.
Whenever witness searches for multiple shortcuts with the same source are required during contraction of some vertex, we run a single multi-target search instead.
Further, to improve query times, we reorder vertices after preprocessing, such that core vertices are in consecutive memory.

%%%%%%%%%%%%%%%%%%%%%%%%%%%%%%%%%%%%%%%%%%%%%%%%%%%%%%%%%%%%%%%%%%%%%%%%%%%%%%%%
\section{CHArge}\label{sec:charge}
%%%%%%%%%%%%%%%%%%%%%%%%%%%%%%%%%%%%%%%%%%%%%%%%%%%%%%%%%%%%%%%%%%%%%%%%%%%%%%%%

Combining \gls*{ch} and A*~search (restricting A*~search to the core), we obtain our fastest exact algorithm, \emph{\acrshort*{charge}\glsunset{charge} (\acrlong*{charge})}.
The query algorithm consists of three phases, namely, a unidirectional (backward) phase from $\target$ in the component to add temporary shortcuts, a backward search in the (much smaller) core enriched with temporary shortcuts to compute a potential function (either $\omegapotential$ or~$\convexpotential^*$), and a forward phase running \gls*{cfp} (augmented with A*~search) from~$\source$, which is restricted to upward arcs, core arcs, and temporary arcs.
Potentials of component vertices are set to 0 for this search. Observe that consistency of the resulting potential function is not violated, since there are no arcs pointing from the core into the component.
As described in Section~\ref{sec:astar}, potentials in the core can be computed on demand, in which case the second and third phase are interweaved and their searches are executed alternately.

\paragraph{Computing Potentials in the Core.}
To decrease running time of the second phase, we precompute lower bounds of core shortcuts for the potential function~$\convexpotential^*$, \ie, for each (ordered) pair $\vertexa\in\vertices$ and $\vertexb\in\vertices$ of vertices connected by at least one shortcut, we compute a decreasing and convex piecewise linear function that yields a lower bound on driving time from $\vertexa$ to $\vertexb$ for a given \gls*{soc}.
To this end, we perform Graham's scan~\cite{Gra72} on all pairs $[(\soc,\atime)]$ of minimum required \gls*{soc} $\soc\in[0,\maxbattery]$ and driving time $\atime\in\posreals$ corresponding to some shortcut arc between $\vertexa$ and $\vertexb$ (recall that the component may contain multi-arcs).
However, this also requires us to adapt the function-propagating search in the second phase, since scanning an arc no longer consists of simply shifting a function by two constant values (\cf Line~\ref{line:cfp-potential-search-shift} of Algorithm~\ref{alg:cfp-potential-backward-search}).
Instead, piecewise linear functions of labels have to be \emph{linked} with shortcut arcs, which are represented by piecewise linear functions as well.
In what follows, we describe how this is done in time linear in the number of breakpoints of both functions.

\begin{figure}[t]
 \centering%
 \begin{subfigure}[b]{.28\textwidth}%
 \centering%
 \begin{tikzpicture}[figure,scale=\ChargingFunctionBoundLinkScale]
  \ExampleDrawCoordinateSystemWithTicks{1}{4}{0}{5}{0}{0}
 
  % Draw function.
  \begin{scope}
   \begin{pgfinterruptboundingbox}
    \clip (0,0) rectangle (5,6+\ExampleInftyOffset);
    \draw [function] (1,3.5) -- (1.5,1.5) -- (3.5,1) -- (4,1);
   \end{pgfinterruptboundingbox}
  \end{scope}

  % Markers at discontinuities and breakpoints.
  \begin{pgfonlayer}{foreground}
   \node[breakpoint] (d2) at (1,3.5) {};
   \node[breakpoint] (b1) at (1.5,1.5) {};
   \node[breakpoint] (b2) at (3.5,1) {};
  \end{pgfonlayer}
 
  % Axes names.
  \PlotXAxisName{4}{0}{$\soc$}
  \PlotYAxisName{0}{5}{$\convexfunctiona(\soc)$}
 \end{tikzpicture}
 \caption{}%
 \label{fig:charge-potential-bound-linking:first}%
 \end{subfigure}%
 \begin{subfigure}[b]{.28\textwidth}%
 \centering%
 \begin{tikzpicture}[figure,scale=\ChargingFunctionBoundLinkScale]
  \ExampleDrawCoordinateSystemWithTicks{1}{4}{0}{5}{0}{0}
 
  % Draw function.
  \begin{scope}
   \begin{pgfinterruptboundingbox}
    \clip (0,0) rectangle (5,6+\ExampleInftyOffset);
    \draw [function] (1,4) -- (3,2) -- (4,2);
   \end{pgfinterruptboundingbox}
  \end{scope}

  % Markers at discontinuities and  breakpoints.
  \begin{pgfonlayer}{foreground}
   \node[breakpoint] (d2) at (1,4) {};
   \node[breakpoint] (b1) at (3,2) {};
  \end{pgfonlayer}
 
  % Axes names.
  \PlotXAxisName{4}{0}{$\soc$}
  \PlotYAxisName{0}{5}{$\convexfunctionb(\soc)$}
 \end{tikzpicture}
 \caption{}%
 \label{fig:charge-potential-bound-linking:second}%
 \end{subfigure}%
 \begin{subfigure}[b]{.42\textwidth}%
 \centering%
 \begin{tikzpicture}[figure,scale=\ChargingFunctionBoundLinkScale]
  % Draw domain.
  \node (p) at (4.5,3.5) {};
  \node (q) at (4,5.5) {};
  \node (r) at (4.5,5) {};
  \node (s) at (2.5,5.5) {};
  \coordinate (x) at (intersection of p--q and r--s);
  \fill[draw=black30,fill=black7] (2,7.5) -- (2.5,5.5) -- (4.5,3.5) -- (6.5,3) -- (4.5,5) -- (x) -- (4,5.5) -- cycle;

  \ExampleDrawCoordinateSystemWithTicks{2}{7}{2}{8}{1}{2}
 
  % Draw function.
  \begin{scope}
   \begin{pgfinterruptboundingbox}
    \clip (1,2) rectangle (7,9+\ExampleInftyOffset);
    \draw [function] (2,7.5) -- (2.5,5.5) -- (4.5,3.5) -- (6.5,3) -- (7,3);
   \end{pgfinterruptboundingbox}
  \end{scope}
 
  % Markers at discontinuities and breakpoints.
  \begin{pgfonlayer}{foreground}
   \node[breakpoint] (d2) at (2,7.5) {};
   \node[breakpoint] (b1) at (2.5,5.5) {};
   \node[breakpoint] (b2) at (4.5,3.5) {};
   \node[breakpoint] (b3) at (6.5,3) {};
  \end{pgfonlayer}
 
  % Axes names.
  \PlotXAxisName{7}{2}{$\soc$}
  \PlotYAxisName{1}{8}{$\convexfunctionc(\soc)$}
 \end{tikzpicture}
 \caption{}%
 \label{fig:charge-potential-bound-linking:result}%
 \end{subfigure}%
 \caption{Linking piecewise linear lower bound functions. Linear segments between indicated breakpoints show (finite) values of two functions and the result of linking them. (a)~The function $\convexfunctiona$ is defined by three breakpoints. (b)~The function $\convexfunctionb$ is defined by two breakpoints. (c)~Linking $\convexfunctiona$ and $\convexfunctionb$ yields the function~$\convexfunctionc$. It is the lower envelope of the shaded area, which corresponds to (finite) values of $\convexfunctiona(\soc^*)+\convexfunctionb(\soc-\soc^*)$ for different choices of~$\soc^*\in\reals$.}%
\label{fig:charge-potential-bound-linking}%
\end{figure}
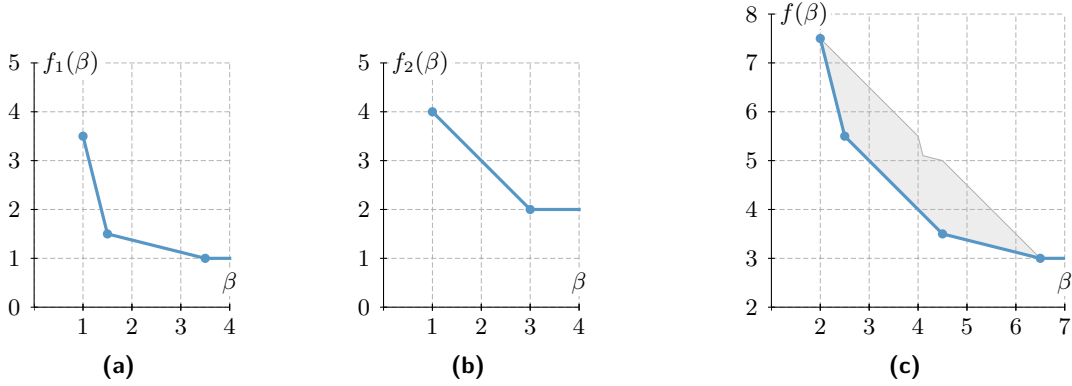

Consider two piecewise linear functions $\convexfunctiona\colon\reals\to\posreals\cup\{\infty\}$ and $\convexfunctionb\colon\reals\to\posreals\cup\{\infty\}$ mapping \gls*{soc} to lower bounds on the trip time along two paths in the graph, such that both are decreasing and convex on their subdomains with finite image.
Let the functions be given by their respective sequences $\convexfunctionbreakpointsa=[(\soc_\functionasubscript^1,\atime_\functionasubscript^1),\dots,(\soc_\functionasubscript^k,\atime_\functionasubscript^k)]$ and $\convexfunctionbreakpointsb=[(\soc_\functionbsubscript^1,\atime_\functionbsubscript^1),\dots,(\soc_\functionbsubscript^\ell,\atime_\functionbsubscript^\ell)]$ of breakpoints.
The \emph{link operation} takes the functions $\convexfunctiona$ and $\convexfunctionb$ as input and computes a function $\convexfunctionc$ that reflects the concatenation of both paths.
Hence, given an arbitrary \gls*{soc}~$\soc\in[0,\maxbattery]$, the value $\convexfunctionc(\soc)$ is a lower bound on the trip time when traversing the paths represented by $\convexfunctiona$ and~$\convexfunctionb$, such that overall consumption does not exceed the \gls*{soc}~$\soc$.
To this end, the link operation identifies values $\soc_1\in\reals$ and $\soc_2\in\reals$, such that $\soc_1+\soc_2=\soc$ and $\convexfunctiona(\soc_1)+\convexfunctionb(\soc_2)$ is minimized.
Hence, we seek to compute the function $\convexfunctionc\colon\reals\to\posreals\cup\{\infty\}$ with
\begin{align}
 \convexfunctionc(\soc):=\min_{\soc^*\in\reals}\convexfunctiona(\soc^*)+\convexfunctionb(\soc-\soc^*)\text{,}\label{eq:charge-linked-function}
\end{align}
which yields the desired lower bound on the trip time for traversing $\convexfunctiona$ and~$\convexfunctionb$.
Figure~\ref{fig:charge-potential-bound-linking} shows an example.
Below, we describe an algorithm that computes a sequence $\convexfunctionbreakpointsc$ of breakpoints that represents this function~$\convexfunctionc$. Afterwards, we prove its correctness.

Starting with an empty sequence $\convexfunctionbreakpointsc=\emptyset$, the link operation iteratively appends breakpoints to~$\convexfunctionbreakpointsc$, each of which is the sum of two breakpoints from $\convexfunctionbreakpointsa$ and~$\convexfunctionbreakpointsb$.
For~$\soc=\soc^1_\functionasubscript+\soc^1_\functionbsubscript$, there exists exactly one value $\soc^*=\soc^1_\functionasubscript$ in Equation~\ref{eq:charge-linked-function} that yields a finite trip time.
We obtain $\convexfunctionc(\soc^1_\functionasubscript+\soc^1_\functionbsubscript)=\atime_\functionasubscript^1+\atime^1_\functionbsubscript$ and the first breakpoint of $\convexfunctionbreakpointsc$ is $(\soc_\functionasubscript^1+\soc_\functionbsubscript^1,\atime_\functionasubscript^1+\atime_\functionbsubscript^1)$.
For subsequent breakpoints, the basic idea is to follow the function that offers the better (\ie, lower) slope.
Assume that the previous breakpoint added to $\convexfunctionbreakpointsc$ is $(\soc_\functionasubscript^i + \soc_\functionbsubscript^{\smash{j}},\atime_\functionasubscript^i + \atime_\functionbsubscript^{\smash{j}})$ for some $i\in\{1,\dots,k\}$ and~$j\in\{1,\dots,\ell\}$. Consider the slope
\begin{align*}
 \convexfunctionslope^i_1 := \begin{cases} \frac{\soc^{i+1}_1 - \soc^i_1}{\atime^{i+1}_1 - \atime^i_1 } &\mbox{if } i < k,\\ 0 &\mbox{otherwise,}  \end{cases}
\end{align*}
of the next segment of the function $\convexfunctiona$. Let the slope~$\convexfunctionslope_\functionbsubscript^{\smash{j}}$ be defined symmetrically. Then the next breakpoint is $(\soc^{i+1}_\functionasubscript + \soc_\functionbsubscript^{\smash{j}},\atime^{i+1}_\functionasubscript + \atime_\functionbsubscript^{\smash{j}})$ if $\convexfunctionslope_\functionasubscript^i< \convexfunctionslope_\functionbsubscript^{\smash{j}}$ and $(\soc_\functionasubscript^i + \soc^{\smash{j+1}}_\functionbsubscript,\atime_\functionasubscript^i + \atime^{\smash{j+1}}_\functionbsubscript)$ if $\convexfunctionslope_\functionasubscript^i> \convexfunctionslope_\functionbsubscript^{\smash{j}}$.
In other words, we pick the next breakpoint of the currently steeper function (keeping the same point as before for the other function).
In the special case~$\convexfunctionslope_\functionasubscript^i = \convexfunctionslope_\functionbsubscript^{\smash{j}}$ we obtain three collinear points, so the next breakpoint is~$(\soc^{i+1}_\functionasubscript + \soc_\functionbsubscript^{\smash{j+1}},\atime^{i+1}_\functionasubscript + \atime_\functionbsubscript^{\smash{j+1}})$.
The scan is stopped as soon as the last point $(\soc_\functionasubscript^k + \soc_\functionbsubscript^\ell,\atime_\functionasubscript^k + \atime_\functionbsubscript^\ell)$ is reached and added to~$\convexfunctionbreakpointsc$.

Clearly, the scan described above runs in linear time in the number $k+\ell$ of breakpoints of $\convexfunctiona$ and~$\convexfunctionb$. Moreover, as segments are appended in increasing order of original slope, the resulting function $\convexfunctionc$ is also decreasing and convex. Lemma~\ref{lem:charge-linking-convex-lower-bounds} formally proves that the link operation in fact computes the correct result.

\begin{lemma}
\label{lem:charge-linking-convex-lower-bounds}
 Let $\convexfunctiona\colon\reals\to\posreals\cup\{\infty\}$ and $\convexfunctionb\colon\reals\to\posreals\cup\{\infty\}$ be piecewise linear functions defined by sequences $\convexfunctionbreakpointsa=[(\soc_\functionasubscript^1,\atime_\functionasubscript^1),\dots,(\soc_\functionasubscript^k,\atime_\functionasubscript^k)]$ and $\convexfunctionbreakpointsb=[(\soc_\functionbsubscript^1,\atime_\functionbsubscript^1),\dots,(\soc_\functionbsubscript^\ell,\atime_\functionbsubscript^\ell)]$ of breakpoints, respectively, such that both functions are decreasing and convex on their subdomains $[\soc_\functionasubscript^1,\infty)$ and $[\soc_\functionbsubscript^1,\infty)$ with finite image.
 The link operation described above computes a sequence $\convexfunctionbreakpointsc$ of breakpoints that corresponds to a function $\convexfunctionc\colon\reals\to\posreals\cup\{\infty\}$ with~$\convexfunctionc(\soc)=\min_{\soc^*\in\reals}\convexfunctiona(\soc^*)+\convexfunctionb(\soc-\soc^*)$.
\end{lemma}

\begin{proof}
 For~$\soc<\soc^1_\functionasubscript+\soc^1_\functionbsubscript$, there exists no value $\soc^*\in\reals$ such that $\convexfunctiona(\soc^*)$ and $\convexfunctionb(\soc-\soc^*)$ are both finite, so the result of linking equals~$\infty$.
 For $\soc=\soc^1_\functionasubscript + \soc^1_\functionbsubscript$ there is exactly one such value $\soc^*\in\reals$ that yields a finite trip time and we obtain $\convexfunctionc(\soc^1_\functionasubscript+\soc^1_\functionbsubscript)=\atime_\functionasubscript^1+\atime^1_\functionbsubscript$, which corresponds to the first breakpoint of~$\convexfunctionbreakpointsc$.

 Let $(\soc_\functionasubscript^i+\soc_\functionbsubscript^{\smash{j}},\atime_\functionasubscript^i+\atime_\functionbsubscript^{\smash{j}})$ with $i\in\{1,\dots,k\}$ and $j\in\{1,\dots,\ell\}$ denote the last point that was added to~$\convexfunctionbreakpointsc$.
 Without loss of generality, assume that $i<k$ and the next segment of the function $\convexfunctiona$ is at least as steep as the next segment of~$\convexfunctionb$, \ie, $\convexfunctionslope_\functionasubscript^i\le\convexfunctionslope_\functionbsubscript^{\smash{j}}$.
 Hence, our algorithm sets~$(\soc_\functionasubscript^{i+1} + \soc_\functionbsubscript^{\smash{j}}, \atime_\functionasubscript^{i+1} + \atime_\functionbsubscript^{\smash{j}})$ as the next breakpoint of~$\convexfunctionbreakpointsc$ (or adds a segment that contains this point in the special case~$\convexfunctionslope_\functionasubscript^i=\convexfunctionslope_\functionbsubscript^{\smash{j}}$).
 Thus, the slope of $\convexfunctionc$ is $\convexfunctionslope_\functionasubscript^i$ for all~$\soc\in[\soc_\functionasubscript^i+\soc_\functionbsubscript^{\smash{j}},\soc_\functionasubscript^{i+1}+\soc_\functionbsubscript^{\smash{j}}]$.
 To prove the claim, we show that
 \begin{align*}
  \convexfunctionc(\soc)=\atime_\functionasubscript^i+\atime_\functionbsubscript^j+\convexfunctionslope_\functionasubscript^i(\soc-\soc_\functionasubscript^i-\soc_\functionbsubscript^j )\le\convexfunctiona(\soc^*)+\convexfunctionb(\soc-\soc^*)
 \end{align*}
 holds for all $\soc\in[\soc_\functionasubscript^i+\soc_\functionbsubscript^{\smash{j}},\soc_\functionasubscript^{i+1}+\soc_\functionbsubscript^{\smash{j}}]$ and $\soc^*\in\reals$.
 Since both $\convexfunctiona$ and $\convexfunctionb$ are convex, we know that $\convexfunctiona(\soc)\ge\atime_\functionasubscript^i+\convexfunctionslope_\functionasubscript^i(\soc-\soc_\functionasubscript^i)$ and $\convexfunctionb(\soc)\ge\atime_\functionbsubscript^{\smash{j}}+\convexfunctionslope^{\smash{j}}_\functionbsubscript(\soc-\soc_\functionbsubscript^{\smash{j}})$ hold for all~$\soc\in\reals$.
 This immediately yields
 \begin{align*}
  \convexfunctiona(\soc^*)+\convexfunctionb(\soc-\soc^*)&\ge\atime_\functionasubscript^i+\convexfunctionslope_\functionasubscript^i(\soc^*-\soc_\functionasubscript^i)+\atime_\functionbsubscript^j+\convexfunctionslope_\functionbsubscript^j(\soc-\soc^*-\soc_\functionbsubscript^j)\\
  &\ge\atime_\functionasubscript^i+\convexfunctionslope_\functionasubscript^i(\soc^*-\soc_\functionasubscript^i)+\atime_\functionbsubscript^j+\convexfunctionslope_\functionasubscript^i(\soc-\soc^*-\soc_\functionbsubscript^j)\\
  &=\atime_\functionasubscript^i+\atime_\functionbsubscript^j+\convexfunctionslope_\functionasubscript^i(\soc-\soc_\functionasubscript^i-\soc_\functionbsubscript^j)\\
  &=\convexfunctionc(\soc)\text{.}
 \end{align*}
 Since $\convexfunctionc(\soc)\ge\min_{\soc^*\in\reals}\convexfunctiona(\soc^*)+\convexfunctionb(\soc-\soc^*)$ must hold for all $\soc\in\reals$ by construction, this completes our proof.
\end{proof}

%%%%%%%%%%%%%%%%%%%%%%%%%%%%%%%%%%%%%%%%%%%%%%%%%%%%%%%%%%%%%%%%%%%%%%%%%%%%%%%%
\section{Heuristic Approaches}\label{sec:heuristics}
%%%%%%%%%%%%%%%%%%%%%%%%%%%%%%%%%%%%%%%%%%%%%%%%%%%%%%%%%%%%%%%%%%%%%%%%%%%%%%%%

With an \NP-hard problem at hand, we propose heuristic approaches based on~\gls*{charge}, which drop optimality to reduce query times.
Their basic idea is as follows. During the third phase of \gls*{charge} (running the \gls*{cfp} algorithm on the core graph), whenever the search scans multiple shortcuts $(\vertexa,\vertexb)$ between two vertices $\vertexa\in\vertices$ and~$\vertexb\in\vertices$, at most one new label is added to~$\labelset_\text{uns}(\vertexb)$.
This saves time for dominance checks and label insertion in the label set~$\labelset_\text{uns}(\vertexb)$.
We use the potential at $\vertexb$ to determine a shortcut that minimizes the key of the new label (\ie, the trip time from the source $\source$ to $\vertexb$ plus a lower bound on the trip time from $\vertexb$ to the target~$\target$), and add only this label to~$\labelset_\text{uns}(\vertexb)$.
Recall that the potential depends on the \gls*{soc} at~$\vertexb$, hence scanning different shortcuts may result in different potentials at~$\vertexb$.

Our first heuristic, denoted~\gls*{charge}-$\algoheuconvex$, uses the potential function~$\convexpotential$ to determine the best shortcut.
Given a label $\alabel$ at some vertex $\vertexa\in\vertices$, scanning an outgoing shortcut $(\vertexa,\vertexb)$ to a vertex $\vertexb\in\vertices$ results in some label $\alabel'$ at $\vertexb$. We compute its key, which minimizes the sum $\atime+\convexpotential(\vertexb,\socfunction{\alabel'}(\atime))$ for arbitrary~$\atime\in\posreals$ (as described in Section~\ref{sec:astar}). However, the label is only added to $\labelset_\text{uns}(\vertexb)$ afterwards if this key is minimal among all labels at $\vertexb$ constructed within the current vertex scan.

The idea of our second heuristic, denoted~\gls*{charge}-$\algoheuomega$, is to use the potential function $\omegapotential$ instead of~$\convexpotential$.
Additionally, when identifying the only label to be inserted into the set $\labelset_\text{uns}(\vertexb)$ of a vertex~$\vertexb\in\vertices$, we ignore battery constraints and presume that we are not close to the target, \ie, that we are in the case $\soc<\distance_{\consumptionfunction}(\vertexb,\target)$ of Equation~\ref{eq:cfp:omegapotential} for arbitrary \gls*{soc}~$\soc\in[0,\maxbattery]$.
Then the best shortcut $(\vertexa,\vertexb)$ does not depend on the \gls*{soc} at~$\vertexb$, but only on the distance~$\distance_\omega(\vertexb,\target)$.
Hence, we can precompute the optimal shortcut for each pair of neighbors $\vertexa\in\vertices$ and~$\vertexb\in\vertices$, namely, the one that minimizes the cost $\omegaweightfunction(\vertexa,\vertexb)$ of the shortcut.
During a query, instead of scanning all shortcuts, we always use the precomputed shortcut for each neighbor $\vertexb$ of~$\vertexa$.

A third, even more aggressive variant, which we denote by~\gls*{charge}-$\algoheuomegaaggressive$, uses the same idea as in \gls*{charge}-$\algoheuomega$ already during vertex contraction for~\gls*{ch}, keeping only the \emph{optimal} shortcut \wrt the cost function~$\omegaweightfunction$ for each pair of vertices. Thus, we no longer allow the creation of multi-arcs during preprocessing.
This significantly reduces the total number of shortcuts in the core graph, allowing the contraction of further vertices.
While the resulting search graph can no longer be used for exact queries, \gls*{charge}-$\algoheuomegaaggressive$ is capable of answering heuristic queries much faster.
The query algorithm of \gls*{charge}-$\algoheuomegaaggressive$ is identical to~\gls*{charge}-$\algoheuomega$, however, solutions may differ as it operates on a sparser graph.

\paragraph{An Optimality Criterion.}
Despite their heuristic nature, it is actually possible to formally grasp under which circumstances the heuristics \gls*{charge}-$\algoheuomega$ and \gls*{charge}-$\algoheuomegaaggressive$ use an optimal shortcut~\cite{Zue15}.
Basically, we know that if charging is inevitable and charging at a rate of $\maxchargingspeed$ is possible when needed, then $\distance_\omega(\cdot,\cdot)$ yields a \emph{tight} bound on the remaining trip time.
The following Lemma~\ref{lem:cfp-opt-shortcut} formalizes this insight.
Recall that when computing shortcut arcs~$(\vertexa,\vertexb)$ for two vertices $\vertexa\in\vertices$ and~$\vertexb\in\vertices$, the only possible charging stations on the underlying $\vertexa$--$\vertexb$~path are $\vertexa$ and~$\vertexb$ (see Section~\ref{sec:ch}).

\begin{lemma}
\label{lem:cfp-opt-shortcut}
 Given two vertices $\vertexa\in\vertices$ and~$\vertexb\in\vertices$, a $\vertexa$--$\vertexb$~path $\apath$ that contains no charging stations (except possibly $\vertexa$ and~$\vertexb$) and minimizes the cost $\omegaweightfunction(\apath)$ \wrt the function $\omegaweightfunction$ among all $\vertexa$--$\vertexb$~paths is a subpath of a fastest feasible $\source$--$\target$~path if the following conditions are met.
 \begin{enumerate}
  \item The fastest feasible $\source$--$\target$~path contains $\vertexa$ and $\vertexb$ in this order, but no charging station on the subpath from $\vertexa$ to $\vertexb$ (except $\vertexa$ and~$\vertexb$).
  \item The \gls*{soc} at $\vertexa$ is not sufficient to reach $\target$ without recharging.
  \item The \gls*{soc} never reaches the capacity $\maxbattery$ (such that battery constraints need to be applied) on the path from $\vertexa$ to the next charging station that is used.
  \item There is a charging station available on the subpath from $\vertexb$ to $\target$ before the battery runs out and the uniform charging speed at this station is~$\chargingfunction_{\max}$.
\end{enumerate}
\end{lemma}

\begin{proof}
For the sake of simplicity, we assume that paths are unique \wrt the weight function~$\omegaweightfunction$.
Assume that all four conditions are met and let~$\apath$ be the $\vertexa$--$\vertexb$~path that minimizes~$\omegaweightfunction(\apath)$. Furthermore, let~$\atime_{\apath}$ be the minimal trip time from~$\source$ to~$\target$ when using~$\apath$ as the subpath from~$\vertexa$ to~$\vertexb$. Assume for contradiction that there exists another path~$\apath'$ from~$\vertexa$ to~$\vertexb$ (without a charging station except~$\vertexa$ and~$\vertexb$) inducing a total trip time from~$\source$ to~$\target$ of~$\atime_{\apath'} < \atime_{\apath}$.

\case{Case 1} We first consider the case of~$\drivingtimefunction(\apath') \le \drivingtimefunction(\apath)$, \ie, the driving time of~$\apath'$ is less than the driving time of~$\apath$.
Since $\apath$ minimizes the cost~$\omegaweightfunction(\apath)$, it follows that~$\consumptionfunction(\apath') > \consumptionfunction(\apath)$ (where $\consumptionfunction(\apath)$ denotes the sum of the consumption values of all arcs in the path~$\apath$).
We replace~$\apath'$ in the fastest~$\source$--$\target$~path with~$\apath$.
This increases the total trip time of the $\source$--$\target$~path by~$\drivingtimefunction(\apath)-\drivingtimefunction(\apath')\ge0$. It remains feasible, since the energy consumption does not increase.
By Condition 2,~$\target$ is not reachable without recharging after reaching~$\vertexa$.
Hence, there has to be a charging station on the~$\vertexb$--$\target$~path that is used.
Due to Condition 3, the battery limit is never reached between~$\vertexb$ and this charging station. Therefore, when replacing $\apath'$ with~$\apath$, the \gls*{soc} at the charging station increases by~$\consumptionfunction(\apath')-\consumptionfunction(\apath)>0$.
Consequently, we do not have to recharge this amount of energy, which reduces the charging time by~$(\consumptionfunction(\apath')-\consumptionfunction(\apath))/\chargingfunction_{\max}$.
Thus, the trip time when replacing $\apath'$ with $\apath$ is given by
\begin{align*}
 \atime_{\apath'}+\drivingtimefunction(\apath)-\drivingtimefunction(\apath')-\frac{\consumptionfunction(\apath')-\consumptionfunction(\apath)}{\chargingfunction_{\max}} \geq \atime_{\apath}\text{.}
\end{align*}
This value cannot be less than~$\atime_{\apath}$, since~$\atime_{\apath}$ is the minimal trip time when using~\apath. We use that~$\atime_{\apath}-\atime_{\apath'}>0$ holds per assumption, which leads to
\begin{align*}
&\drivingtimefunction(\apath)-\drivingtimefunction(\apath')-\frac{\consumptionfunction(\apath')-\consumptionfunction(\apath)}{\chargingfunction_{\max}} \geq \atime_{\apath}-\atime_{\apath'} > 0\\
\Leftrightarrow~~&\drivingtimefunction(\apath)+\frac{\consumptionfunction(\apath)}{\chargingfunction_{\max}} > \drivingtimefunction(\apath')+\frac{\consumptionfunction(\apath')}{\chargingfunction_{\max}}\\
\Leftrightarrow~~&\omega(\apath) > \omega(\apath')\text{,}
\end{align*}
contradicting the fact that~$\apath$ is the~$\vertexa$--$\vertexb$~path that minimizes the cost $\omega$ when going from $\vertexa$ to~$\vertexb$.

\case{Case 2} We consider the case of~$\drivingtimefunction(\apath') > \drivingtimefunction(\apath)$. Hence, we have~$\consumptionfunction(\apath') \le \consumptionfunction(\apath)$.
Again, we replace~$\apath'$ with~$\apath$ in the~$\source$--$\target$~path, which reduces the driving time by~$\drivingtimefunction(\apath')-\drivingtimefunction(\apath)$ and the \gls*{soc} at~$\vertexb$ by~$\consumptionfunction(\apath)-\consumptionfunction(\apath')$.
Condition 4 ensures that there is a charging station available on the path from~$\vertexb$ to~$\target$, which we use to recharge the missing energy.
Condition 3 states that the \gls*{soc} is always below the battery limit between~$\vertexa$ and this charging station.
Therefore, when reaching the charging station, the amount of energy missing is~$\consumptionfunction(\apath)-\consumptionfunction(\apath')$ compared to the~$\source$--$\target$~path using~$\apath'$. We use the charging station to recharge this amount of energy, to ensure that the path to~$\target$ remains feasible.
Altogether, the trip time for reaching~$\target$ when replacing~$\apath'$ with~$\apath$ is given by
\begin{align*}
 \atime_{\apath'}-\left(\drivingtimefunction(\apath')-\drivingtimefunction(\apath)\right)+\frac{\consumptionfunction(\apath)-\consumptionfunction(\apath')}{\chargingfunction_{\max}} \geq \atime_{\apath}\text{.}
\end{align*}
As in the first case, this is equivalent to~$\omega(\apath) > \omega(\apath')$, which is a contradiction to our assumption that $\apath$ minimizes the cost~$\omega$. This completes the proof.
\end{proof}

%%%%%%%%%%%%%%%%%%%%%%%%%%%%%%%%%%%%%%%%%%%%%%%%%%%%%%%%%%%%%%%%%%%%%%%%%%%%%%%%
\section{Experiments}\label{sec:experiments}
%%%%%%%%%%%%%%%%%%%%%%%%%%%%%%%%%%%%%%%%%%%%%%%%%%%%%%%%%%%%%%%%%%%%%%%%%%%%%%%%

All implementations are in~C++ using~g++ 4.8.3 (-O3) as compiler. Experiments were conducted on a single core of a 4-core Intel Xeon E5-1630v3 clocked at 3.7\,GHz, 128\,GiB of DDR4-2133 RAM, 10\,MiB of L3 cache, and 256\,KiB of L2 cache.

\begin{figure} [t]
 \centering
 \includegraphics[width=0.4\textwidth]{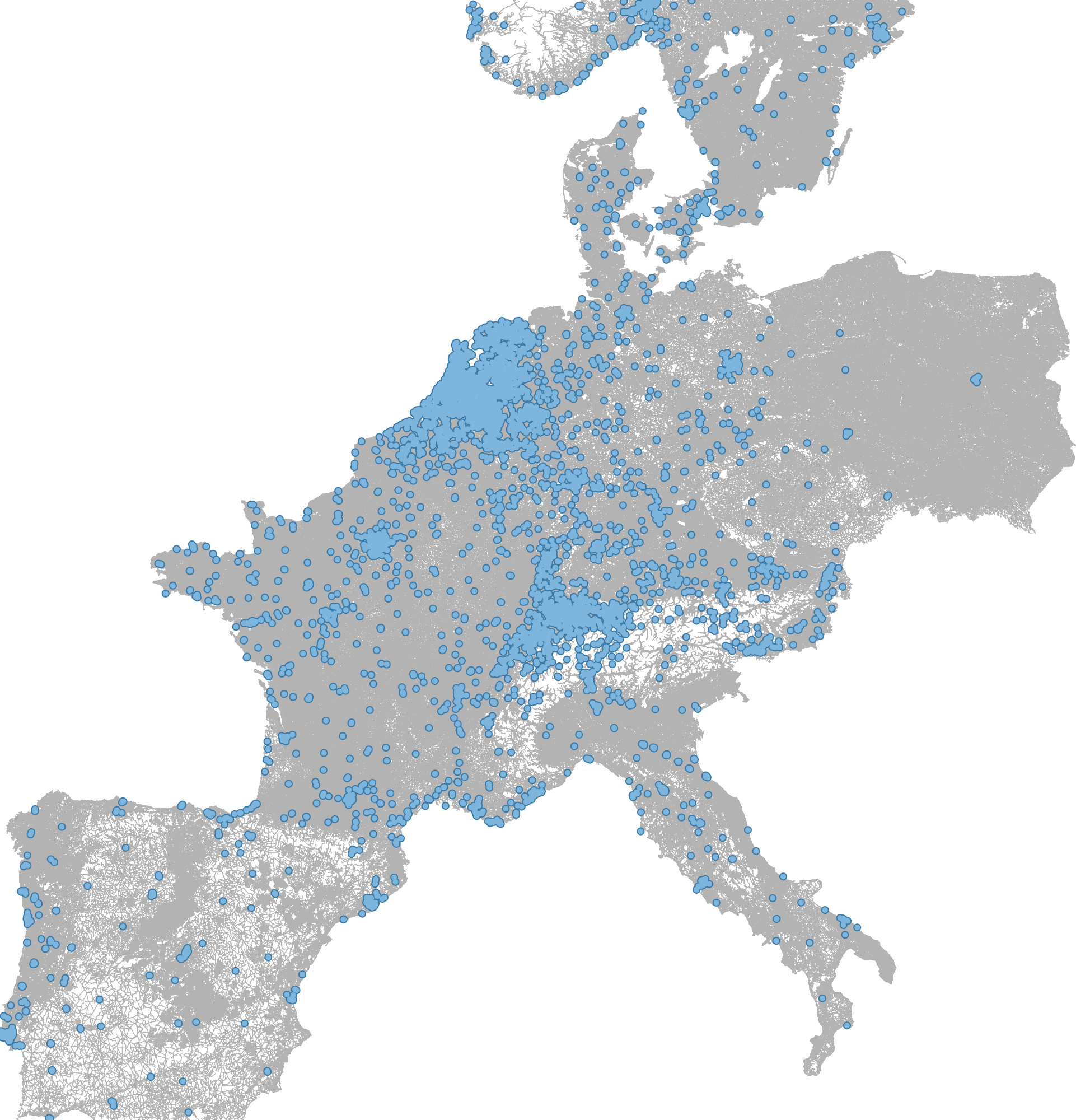}
 \caption{Distribution of charging stations in Europe according to ChargeMap data. We see differences in the distribution of charging stations, which is very dense in the Netherlands and Switzerland. On the other hand, Spain, Italy, and Poland contain relatively few charging stations in our data set. Note that the United Kingdom is not contained in our instance, because we extracted the largest strongly connected component of the input (in which the United Kingdom is only reachable via ferry).}
 \label{fig:chargemap}
\end{figure}

\begin{table}[b]
\caption{Instances. We report the number of vertices (including charging stations), arcs, negative consumption arcs (as a fraction of total arcs), as well as charging stations alone~(CS) obtained from ChargeMap. Note that the Southern Germany instance based on \acrshort*{osm} data has many degree-2 vertices, meant for visualization.}
\label{tbl:graphs}
\setlength{\tabcolsep}{1.5ex}
\centering\small
\begin{tabular}{lrrrrr}
\toprule
Instance & \#\,Vertices & \#\,Arcs & \multicolumn{2}{r}{\#\,Arcs with $\consumptionfunction < 0$} & \#\,CS \\ 
\midrule
Germany          &   4\,692\,091 & 10\,805\,429 & 1\,119\,710 &\hspace{-10pt}(10.36\%) & 1\,966 \\
Europe           &  22\,198\,628 & 51\,088\,095 & 6\,060\,648 &\hspace{-10pt}(11.86\%) & 13\,810 \\
Southern Germany &   5\,588\,146 & 11\,711\,088 & 1\,142\,391 &\hspace{-10pt}(9.75\%) & 643\\ 
 \bottomrule
\end{tabular}
\end{table}

\paragraph{Instances.}
Our main instances are based on the road network of Europe~(Eur) and the subnetwork of Germany~(Ger), kindly provided by~PTV~AG (\url{ptvgroup.com}). As in our previous work~\cite{Bau13a}, we extracted average speeds and categories of road segments, augmented by elevation data from the \gls*{srtm},~v4.1 (\url{srtm.csi.cgiar.org}). We derived realistic energy consumption from a detailed micro-scale emission model~\cite{Hau09}, calibrated to a real Peugeot~iOn. It has a battery capacity of 16\,kWh, but we also evaluate for 85\,kWh, as in high-end Tesla models (with a range of 400--500\,km).
Moreover, we located charging stations on ChargeMap (\url{chargemap.com}).
Figure~\ref{fig:chargemap} shows the distribution of charging stations on the Europe instance according to ChargeMap data.
We also evaluate the largest instance considered by Storandt~\cite{Sto12c}. It uses \gls*{osm} data of Southern Germany enriched with \gls*{srtm}, a basic physical consumption model, and 100--1\,000 randomly chosen charging stations, but we also test on ChargeMap stations.
See Table~\ref{tbl:graphs} for an overview of instances and their sizes.
We constructed different charging functions to model certain types of stations, namely,~\gls*{bss},~superchargers (charging an empty battery to 80\,\% in 34\,min for both vehicle models), as well as regular stations with fast~(44\,kW), medium~(22\,kW), or slow~(11\,kW) charging power. For the three latter types of functions, we used the physical model of Uhrig~et~al.~\cite{Uhr15} and approximated the corresponding charging functions with a piecewise linear function (using six breakpoints at~0\,\%,~80\,\%,~85\,\%,~\,90\,\%,~95\,\%, and~100\,\% \gls*{soc}).
See Figure~\ref{fig:csp:experiments-charging-functions} for a plot showing the resulting charging functions for a battery capacity of~85\,kWh.
We set initialization time to 3\,min for~\gls*{bss} and 1\,min for all other types of charging stations.
Both the road network and the energy consumption data mentioned above are based on proprietary data sources, but this allowed us to perform experiments on highly detailed and realistic input data.

\begin{figure}[t]
 \centering
 \begin{tikzpicture}[figure]
  \pgfplotsset{
   grid style = {dash pattern = on 1pt off 1pt, black15,line width = 0.5pt  }
  }

  \begin{axis}[
   name=border,
   height=5.2cm,
   width=0.6\textwidth,
   clip marker paths=true,
   xmin=0,
   xmax=10,
   ymin=0,
   ymax=100,
   /pgf/number format/.cd,
   use comma,
   1000 sep={\,},
   xlabel={Time [h]},
   ylabel={SoC [\%]},
   ymode=linear,
   grid=major,
   legend entries={supercharger, 44\,kWh, 22\,kWh, 11\,kWh},
   legend cell align=left,
   legend style={at={(0.98,0.04)},
   anchor=south east}
  ]

  \addlegendimage{line width=1pt,thesisgreen}
  \addlegendimage{line width=1pt,thesisred}
  \addlegendimage{line width=1pt,thesisyellow}
  \addlegendimage{line width=1pt,thesisblue}
  
  \addplot [color=thesisgreen,line width=1pt] table {
          0   0
      0.567  80
     10      80
  };

  \addplot [color=thesisblue,line width=1pt] table {
          0   0
      6.231  80
      6.677  85
      7.284  90
      8.178  95
     10.829 100
  };
  
  \addplot [color=thesisyellow,line width=1pt] table {
          0   0
      3.116  80
      3.334  85
      3.642  90
      4.089  95
      5.414 100
  };
    
  \addplot [color=thesisred,line width=1pt] table {
          0   0
      1.558  80
      1.670  85
      1.821  90
      2.045  95
      2.710 100
  };
  \end{axis}
  \draw (border.north east) rectangle (border.south west);
 \end{tikzpicture}
 \caption{Charging functions we use in our experiments for a battery capacity of~85\,kWh. Swapping stations (which yield a departure \gls*{soc} of 100\% for any charging time) are not shown in the plot.}
 \label{fig:csp:experiments-charging-functions}
\end{figure}
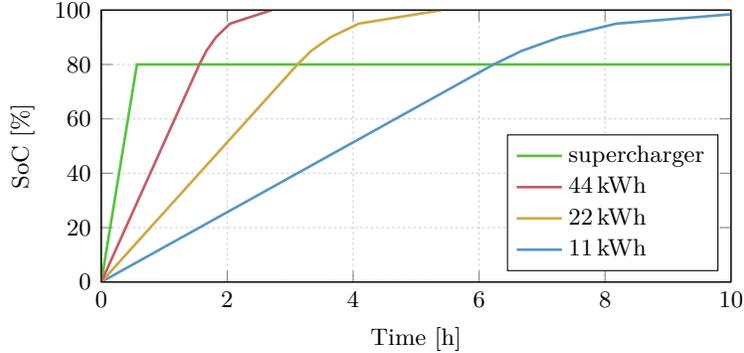

\begin{table}[b]
\caption{Types of charging stations used in the different scenarios.}
\label{tbl:charging-stations}
\centering\small
\begin{tabular}{lrrrrr}
\toprule
Scenario  & 11\,kWh & 22\,kWh & 44\,kWh & Sup.\,Ch. &     BSS \\ 
\midrule
BSS       &     --- &     --- &     --- &       --- & 100\,\% \\
Mixed     &  30\,\% &  30\,\% &  20\,\% &    10\,\% &  10\,\% \\
Realistic &  50\,\% &  40\,\% &  10\,\% &       --- &     --- \\
\bottomrule
\end{tabular}
\end{table}

If not stated otherwise, queries are always generated by picking source and target vertices uniformly at random and an initial \gls*{soc} of~$\soc_\source=\maxbattery$.
In our evaluation, we consider three different distributions of charging stations; see Table~\ref{tbl:charging-stations} for an overview. In the first, all charging stations are~\gls*{bss}. This is the simplest scenario we evaluate, which allows us to compare our algorithm to previous approaches based on simpler models of charging functions~\cite{Sto12c}.
The second uses a mixed composition of chargers, for which we randomly pick 10\,\% of all stations as~\gls*{bss}, 10\,\% as superchargers, 20\,\% as fast chargers, and each 30\,\% as medium and slow chargers. Containing various types of charging stations, this scenario is much more difficult to solve for our approach.
Finally, the third scenario reflects a more realistic distribution of charging stations, where 10\,\% of all stations are fast, 40\,\% are medium, and 50\,\% are slow chargers. This roughly corresponds to the distribution of charging stations in real-world road networks.
In each scenario, the composition is fixed, \ie, every charging station retains its type for all queries.

\paragraph{Evaluating Queries.}
We discuss preprocessing and query performance of our algorithms. We only report the fastest exact method (\gls*{charge} with the potential~$\convexpotential$) and our heuristic approaches---for results on plain \gls*{cfp} see below.

\begin{table}[b]
\caption{Impact of core size on performance (\acrshort*{charge},~Germany,~85\,kWh). We stopped contraction if the average degree in the core graph exceeded a certain threshold~(\O\,Deg). We report the core size (\#\,Vertices), preprocessing time, and average query times of 1\,000 queries answered by~\acrshort*{charge}.}
\label{tbl:csp:core-size-charge}
\setlength{\tabcolsep}{1.5ex}
\centering\small
\begin{tabular}{rrrrrr}
 \toprule
 \multicolumn{2}{c}{Core size} & Prepr. & \multicolumn{3}{c}{Query [ms]} \\
 \cmidrule(lr){1-2}\cmidrule(lr){3-3}\cmidrule(lr){4-6}
 \,\O\,Deg. & \#\,Vertices & [h:m:s] & BSS & Mixed & Realistic \\ \midrule
      8     & 344\,066 (7.33\%) &      2:58 &    4\,461.9 & 203\,576.7 & 179\,430.6 \\
     16     & 116\,917 (2.49\%) &      4:01 &    1\,773.4 &   6\,739.0 &   4\,423.1 \\
     32     &  65\,375 (1.39\%) &      5:03 &    1\,020.1 &   3\,037.3 &   2\,425.5 \\
     64     &  43\,036 (0.91\%) &      7:07 &    1\,062.2 &   3\,184.5 &   2\,533.9 \\
    128     &  30\,526 (0.65\%) &     11:16 &    1\,217.2 &   3\,728.3 &   2\,901.2 \\
    256     &  22\,592 (0.48\%) &     20:22 &    1\,566.0 &   4\,987.5 &   4\,144.4 \\
    512     &  17\,431 (0.37\%) &     37:11 &    2\,285.8 &   6\,936.1 &   5\,498.0 \\
\bottomrule
\end{tabular}
\end{table}

Table~\ref{tbl:csp:core-size-charge} shows details on preprocessing effort and query performance for different core sizes on Germany (for a battery capacity of 16\,kWh). In this experiment, vertex contraction during preprocessing was stopped as soon as the average vertex degree in the core reached the threshold shown in the first column of the table.
We report resulting core graph sizes, preprocessing times, and average query times of~\gls*{charge} for 1\,000 random queries on the \gls*{bss}, mixed, and realistic station composition, respectively.
Apparently, preprocessing effort increases with decreasing core size. We achieve best query performance at an average core degree of~32. At higher degrees, the rather dense core causes query times to increase. Therefore, we use a core degree of 32 as threshold to stop vertex contraction in all further experiments. This results in relative core sizes of 1.3--1.7\,\% on Germany and Europe.
Regarding query times, we observe that the \gls*{bss} composition is the easiest to solve, because vertex potentials are very accurate with only one available charging station type.

\begin{table}[t]
\caption{Preprocessing and query performance for different instances, charging station types, and battery capacities. We report regular preprocessing times for \gls*{charge} (which also applies to the heuristics \gls*{charge}-\algoheuconvex and \gls*{charge}-\algoheuomega) and preprocessing times for~\gls*{charge}-\algoheuomegaaggressive, the percentage of feasible queries, as well as average exact and heuristic query times of 1\,000 queries.}
\label{tbl:csp:charge-performance}
\centering\small
\begin{tabular}{lclrrcrrrrr}
\toprule
&&& \multicolumn{2}{c}{Pr.\,[m:s]} &&&  \multicolumn{4}{c}{Query\,[ms]}\\
\cmidrule(lr){4-5}\cmidrule(lr){8-11}
Ins. & $\maxbattery$ & CS & reg. & \algoheuomegaaggressive  && Feas. & Exact & \algoheuconvex & \algoheuomega & \algoheuomegaaggressive \\
\midrule
\addlinespace
\multirow{4}{*}{\vspace{-65pt}\hspace{10pt}\vertical{Germany}}
 & 16\,kWh & BSS       &  5:03 &  4:33 && 100\,\% &   1\,607.6 &   1\,278.3 &      548.8 &      23.5 \\
 & 16\,kWh & Mixed     &  5:03 &  4:33 && 100\,\% &  12\,191.0 &   4\,217.4 &   1\,781.9 &     189.2 \\
 & 16\,kWh & Realistic &  5:03 &  4:33 && 100\,\% &   6\,177.5 &   2\,613.0 &   1\,119.3 &     103.6 \\
\addlinespace[3pt]
 & 85\,kWh & BSS       &  4:59 &  5:31 && 100\,\% &   1\,020.1 &      988.1 &      174.3 &      31.6 \\
 & 85\,kWh & Mixed     &  4:59 &  5:31 && 100\,\% &   3\,037.3 &   2\,346.4 &      555.3 &      51.7 \\
 & 85\,kWh & Realistic &  4:59 &  5:31 && 100\,\% &   2\,425.5 &   1\,405.7 &      358.7 &      45.6 \\
\addlinespace
\multirow{4}{*}{\vspace{-55pt}\hspace{10pt}\vertical{Europe}}
 & 16\,kWh & BSS       & 30:32 & 28:38 &&  68\,\% &   5\,958.7 &   3\,050.8 &   4\,046.7 &     153.0 \\
 & 16\,kWh & Mixed     & 30:32 & 28:38 &&  68\,\% &  35\,532.8 &  16\,830.7 &  19\,338.9 &  3\,436.2 \\
 & 16\,kWh & Realistic & 30:32 & 28:38 &&  68\,\% &  20\,310.0 &  11\,650.4 &  13\,986.7 &  2\,146.5 \\
\addlinespace[3pt]
 & 85\,kWh & BSS       & 30:16 & 29:47 && 100\,\% &  31\,548.3 &  22\,602.6 &   1\,571.3 &      52.8 \\
 & 85\,kWh & Mixed     & 30:16 & 29:47 && 100\,\% & 100\,448.7 &  54\,279.8 &  38\,520.0 &     718.3 \\
 & 85\,kWh & Realistic & 30:16 & 29:47 && 100\,\% &  74\,362.2 &  40\,828.3 &  13\,596.2 &     514.9 \\
\bottomrule
\end{tabular}
\end{table}

Table~\ref{tbl:csp:charge-performance} shows detailed timings on performance for 1\,000 queries on each considered instance.
We evaluate three scenarios (\gls*{bss} only, mixed, and realistic compositions of charging stations) and use our instances Germany and Europe with typical battery capacities (16\,kWh and 85\,kWh).
Preprocessing times are quite practical, considering the problem at hand, ranging from about 5--30\,minutes. 
As before, the mixed and realistic scenarios are harder to solve. As a general observation, increasing the maximum battery capacity leads to faster queries. This can be explained by the fact that less charging is required, so goal direction provided by the potential functions of A* search is more helpful. A notable exception is the capacity of 16\,kWh on Europe. In this case, the number of feasible queries drops significantly, due to a highly non-uniform distribution of charging stations (sparse in parts of Southern and Eastern Europe; see Figure~\ref{fig:chargemap}).
This benefits approaches based on the potential function~$\convexpotential$, as we can often detect infeasibility already during potential computation (the lower bound on trip time evaluates to~$\infty$). On the other hand, infeasible queries often deteriorate the performance when using the potential function~$\omegapotential$: Because the target is never reached, large parts of the graph are explored until the queue runs empty.

All in all, running times of the exact algorithm are below 15~seconds on average on Germany and at most 100~seconds on Europe, which is quite notable given that we could not even run a single (long-distance) \gls*{cfp} query on this instance in several hours.
Note that the mixed composition is a rather difficult configuration for our algorithms, due to large differences in the charging speeds of the available charging station types (which results in less accurate vertex potentials).
When using the potential function $\omegapotential$ for (exact) \gls*{charge} (not reported in the table), running times increase by up to an order of magnitude, depending on the scenario. Hence, plugging in the more sophisticated potential function $\convexpotential$ pays off.
For heuristic approaches, we see that---in contrast to \gls*{charge}---those based on the potential function $\omegapotential$ are faster (except for instances with many infeasible~queries).

\begin{table}[t]
\caption{Detailed query performance of \gls*{charge} (Germany,~85\,kWh) for mixed and realistic charging stations~(CS). For exact \gls*{charge} (Ex.) and the different heuristics (\algoheuconvex,~\algoheuomega, and~\algoheuomegaaggressive), we report average values of the number of settled labels (\#\,Labels), pairwise dominance checks~(\#\,Dom.), and running times for 1\,000 queries. For the resulting trips, we report the percentage of feasible and optimal paths, average and maximum increase in trip time compared to an optimal solution, as well as average trip time and average total charging time on the resulting routes.}
\label{tbl:csp:charge-detailed-queries}
\centering\small
\setlength{\tabcolsep}{0.98ex}
\begin{tabular}{llrrrcrrrrcrr}
 \toprule
 & &\multicolumn{3}{c}{Query} && \multicolumn{4}{c}{Trip Quality} && \multicolumn{2}{c}{Trip Times}\vspace{1pt}\vspace{1pt}\\
 \cmidrule(lr){3-5}\cmidrule(lr){7-10}\cmidrule(lr){12-13}
 CS & Algo.  & \#\,Labels & \#\,Dom. &  Time\,[ms] && Feas.  & Opt. & Avg. & Max.\vspace{2pt} && $\triptime$ & $\chargingtime$ \\
 \midrule
 \multirow{4}{*}{\vspace{-24pt}\hspace{10pt}\vertical{Mixed}}
 & Ex.                     & 490\,421 & 61\,183\,835 & 3\,037.3 && 100\,\% & 100\,\% & 1.0000 & 1.0000 && 3:25:45 & 2:01 \\
 & \algoheuconvex          & 468\,977 &     161\,103 & 2\,346.4 && 100\,\% &  84\,\% & 1.0012 & 1.0519 && 3:26:00 & 2:01 \\
 & \algoheuomega           & 313\,018 &  7\,864\,494 &    555.3 && 100\,\% &  82\,\% & 1.0004 & 1.0198 && 3:25:50 & 2:01 \\
 & \algoheuomegaaggressive &  16\,110 &      45\,471 &     51.7 && 100\,\% &  54\,\% & 1.0198 & 1.1833 && 3:29:49 & 2:02\vspace{2pt}\\
 \addlinespace
 \multirow{4}{*}{\vspace{-32pt}\hspace{10pt}\vertical{Realistic}}
 & Ex.                     & 406\,961 & 49\,901\,425 & 2\,425.5 && 100\,\% & 100\,\% & 1.0000 & 1.0000 && 3:43:58 & 23:11 \\
 & \algoheuconvex          & 345\,397 &     112\,669 & 1\,405.7 && 100\,\% &  84\,\% & 1.0006 & 1.0362 && 3:44:05 & 23:11 \\
 & \algoheuomega           & 219\,620 &  4\,061\,188 &    358.7 && 100\,\% &  69\,\% & 1.0010 & 1.0524 && 3:44:12 & 23:11 \\
 & \algoheuomegaaggressive &  16\,837 &      38\,691 &     45.6 && 100\,\% &  59\,\% & 1.0153 & 1.1575 && 3:47:24 & 23:29 \\
 \bottomrule
\end{tabular}
\end{table}

Table~\ref{tbl:csp:charge-detailed-queries} reports detailed figures (for the same set of queries as in Table~\ref{tbl:csp:charge-performance}) on Germany, using a battery capacity of~85\,kWh. We argue that this results in the most sensible queries: Due to a reasonably dense distribution of charging stations in Germany (see Figure~\ref{fig:chargemap}), all queries are feasible and the target is always reachable with a small number of charging stops.
Consequently, we obtain an average trip time of~3\,h\,25\,min\,45\,sec and an average charging time of 2\,min\,1\,sec on Germany for the mixed composition of charging stations. In contrast to that, queries across Europe evaluated in Table~\ref{tbl:csp:charge-performance} (which we analyze rather to show scalability of our algorithms) yield an average trip time of over~10\,h.
We also add figures for the realistic scenario, which yields an average trip time of 3\,h\,43\,min\,58\,sec and an average charging time of~23\,min\,11\,sec.
Independent of the scenario and the utilized algorithm, the number of required charging stops is $0.57$ on average and no path has more than four charging stops in total.
Note that our algorithm does not exhibit any limit on the number of charging stops. In theory, it is therefore easy to construct examples where an optimal route contains a large number of charging stops (even in cases where the target is reachable with only a few or no stops at all). However, this is unlikely to happen in practice, since a charging stop typically involves a detour and we also take initialization time into account.
Our experiments confirm this: The number of charging stops is small despite the fact that we assess a rather moderate initialization time of only one minute.

Regarding the performance of our algorithms, we observe that the number of settled labels is a good indicator of the running time.
Moreover, all approaches are quite practical. Computing optimal routes takes only a few seconds.
The heuristic \gls*{charge}-\algoheuomega provides a good trade-off between running times (300--600\,ms) and resulting errors (0.1\,\% on average, 5\,\% in the worst case).
Our aggressive approach \gls*{charge}-\algoheuomegaaggressive allows query times of about 50\,ms on average, which is fast enough even for interactive applications. However, we see that solution quality is up to~18\,\% off the optimum in the worst case. Still, the average error of all heuristics is very small, and the optimal solution is found in more than half of the cases.

\begin{figure}[t]
 \centering
 \input{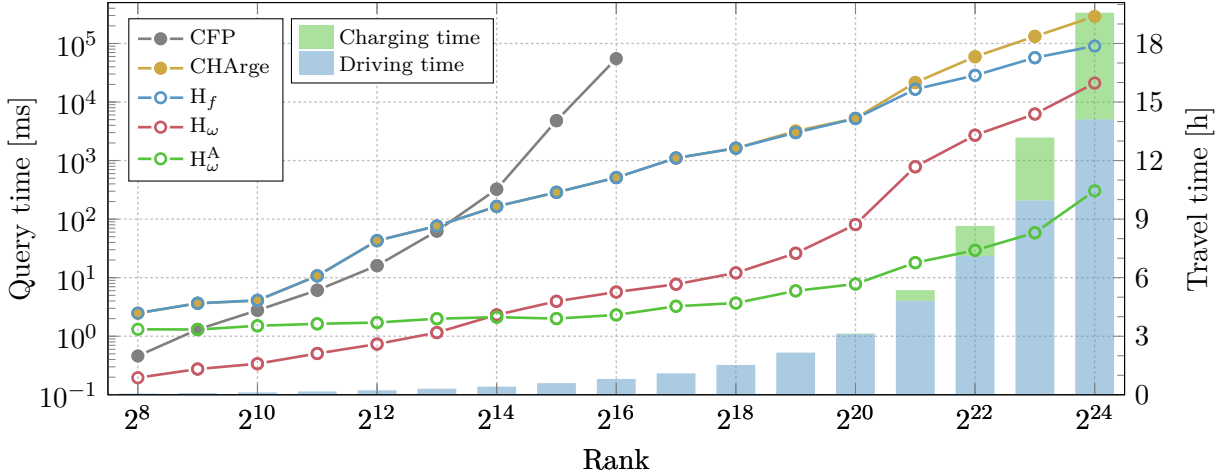}
 \caption{Running times of \acrshort*{charge} subject to Dijkstra rank. Each point in the plot corresponds to the median running time of 100 random queries of the indicated Dijkstra rank on Europe with a 85\,kWh battery and the realistic composition of charging stations. Each bar shows the average driving time (blue) and charging time (green) that an optimal solution of the respective rank requires (computed by \acrshort*{charge}).}
 \label{fig:csp:charge-rankplot-combined}
\end{figure}

\paragraph{Evaluating Scalability.}
In our evaluations above, we considered queries with source and target vertices picked uniformly at random. Note that this typically results in long-distance queries (in relation to the graph diameter). To distinguish local and long-distance queries, we consider the \emph{Dijkstra rank} of a query, which equals the number of queue extractions when running Dijkstra's algorithm from the source to the target with unconstrained driving time as arc costs~\cite{Bas14,San05}. Thus, higher ranks reflect queries that are more difficult to solve. Moreover, charging en route is only required for queries of high rank.

Figure~\ref{fig:csp:charge-rankplot-combined} shows median query running times, as well as average driving and charging times, on our Europe instance with realistic charging stations, subject to their Dijkstra rank.
We ran 100 queries per rank. Query times for \gls*{cfp} are only reported up to a rank of $2^{16}$, because for higher ranks at least one query did not terminate within a predefined limit of one hour.
The colored bars show that higher Dijkstra ranks correspond to routes with longer travel time.
The green part of each bar shows the portion of the overall travel time that is spent at charging stations.
Given a relatively large battery capacity of~85\,kWh, charging stops are only necessary for queries of the highest Dijkstra ranks, starting from rank~$2^{21}$.
Consequently, running times of our faster approaches increase significantly at this rank. It also becomes evident that all approaches have exponential running times, with timings of \gls*{charge} in the order of minutes for the highest ranks. Nevertheless, performance remains practical for ranks beyond~$2^{20}$. Running times of our fastest heuristic (\gls*{charge}-\algoheuomegaaggressive) are well below a second for the highest ranks.

\begin{table}[t]
\caption{Comparison of \acrshort*{charge} with existing work~\cite{Sto12c}. For different distributions of charging stations and ranges, we report preprocessing times of the existing approach~\cite{Sto12c},~\gls*{charge}, and the heuristic~\gls*{charge}-\algoheuomegaaggressive. Regarding queries, we show the percentage of feasible paths, as well as average exact and heuristic query times of 1\,000 queries. Results from the literature~\cite{Sto12c} are reported as is from a Core~i3-2310M,~2.1\,GHz.}
\label{tbl:csp:charge-comparison}
\centering\small
\setlength{\tabcolsep}{0.98ex}
\begin{tabular}{lcrrrcrrrrrr}
\toprule
&& \multicolumn{3}{c}{Prepr.\,[m:s]} &&& \multicolumn{5}{c}{Query\,[ms]}\\
\cmidrule(lr){3-5}\cmidrule(lr){8-12}
CS & $\maxbattery$ & \cite{Sto12c} & \acrshort*{charge} & \algoheuomegaaggressive  && Feas. & \cite{Sto12c} & \acrshort*{charge} & \algoheuconvex & \algoheuomega & \algoheuomegaaggressive\\
\midrule
r1000 & 100\,km & 51:41 & 11:37 & 2:30 && 100\,\% &    539.0 & 122.1 &   93.7 & 38.1 & 2.7\\
r100  & 150\,km & 16:21 & 11:10 & 2:15 &&  99\,\% & 1\,150.0 & 206.0 &  116.5 & 26.9 & 1.4\\
c643  & 100\,km &   --- & 11:21 & 2:32 &&  98\,\% &      --- & 326.3 &  282.6 & 48.3 & 3.3\\
c643  & 150\,km &   --- & 11:28 & 2:29 &&  99\,\% &      --- & 308.3 &  270.0 & 22.7 & 2.9\\
\bottomrule
\end{tabular}
\end{table}

\paragraph{Comparison with Related Work.}
We also consider the largest instance that was used to evaluate the fastest approach for routes with charging stops in the literature~\cite{Sto12c}.
It utilizes a simple physical consumption model, where energy consumption of an arc is a linear combination of its horizontal distance and slope.
All charging stations in the original work are~\gls*{bss}, picked uniformly at random from all vertices of the input graph.
Table~\ref{tbl:csp:charge-comparison} shows detailed figures, using 100--1\,000 randomly placed charging stations~(r100,~r1000) and battery capacities that translate to a certain range in the physical model.
Query times are average values of 1\,000 random queries.
For completeness, we also run tests on ChargeMap stations (only \gls*{bss}), 643 of which are included in the considered road network. As before, query times are average values of 1\,000 random queries.
We observe that our approach is faster \wrt both preprocessing and queries, even when taking differences in hardware into account. At the same time, \gls*{charge} is more general and not inherently restricted to~\gls*{bss}.
Furthermore, note that a denser distribution of random charging stations enables faster query times: Although preprocessing effort increases slightly (more charging stations are kept in the core), the potential functions more often rightly assume that a station will be available close to the remaining path.
We see that using charging station locations from ChargeMap~(c643) gives a slightly harder instance.
In conclusion, we are able to solve instances that are harder than those evaluated in the literature.

%%%%%%%%%%%%%%%%%%%%%%%%%%%%%%%%%%%%%%%%%%%%%%%%%%%%%%%%%%%%%%%%%%%%%%%%%%%%%%%%
\section{Conclusion}\label{sec:conclusion}
%%%%%%%%%%%%%%%%%%%%%%%%%%%%%%%%%%%%%%%%%%%%%%%%%%%%%%%%%%%%%%%%%%%%%%%%%%%%%%%%

We introduced \gls*{cfp}, a novel approach for \gls*{ev} route planning that computes shortest feasible paths with charging stops, minimizing overall trip time. Combining vertex contraction and charging-aware goal-directed search, we developed \gls*{charge}, which solves this \NP-hard problem optimally and with practical performance even on large realistic inputs.
For reasonable distances, our algorithm computes solutions within seconds, making \gls*{charge} the first practical exact approach---with running times similar to previous methods that are inexact or use less accurate models~\cite{Goo14,Sto12c}.
Finally, we proposed heuristic variants that enable even faster queries, while offering high-quality solutions.

Future work includes allowing multiple driving speeds~(and consumption values) per road segment, \eg, in order to save energy on the highway by driving at reasonable lower speeds~\cite{Bau17a,Har14}.
Moreover, we are interested in more detailed models for energy consumption on turns.
In preliminary experiments, we modified our input instances according to a known arc-based approach for integrating turns~\cite{Gei11a}. Using a simple model that takes consumption overhead for acceleration and deceleration along turns (or when speed limits change) into account, we observed that preprocessing took slightly longer but also resulted in smaller cores, while query times remained nearly the same.
This can be explained by the fact that graph size increases when it is enriched with turn costs, but at the same time the number of nondominated solutions decreases (minor detours no longer allow energy savings). This clearly indicates that our techniques can be readily extended to handle turn costs and turn restrictions.

For integration of historic and live traffic, the adaptation of customizable techniques appears to be a natural extension of our approaches. While preliminary experiments showed that our techniques based on \gls*{ch} perform equally fast on vertex orders required by Customizable \gls*{ch}~\cite{Dib16}, a major challenge is the design of fast customization algorithms that can handle the dynamic data structures required by label sets in our approach.
On the other hand, preprocessing time for \gls*{charge} is already in the order of minutes and can even be reduced by increasing core size at the cost of higher query times; see Table~\ref{tbl:csp:core-size-charge} in Section~\ref{sec:experiments}. This is fast enough for many applications that require frequent updates of the underlying graph.

Furthermore, we are interested in extending our algorithms to more complex settings, \eg, by taking into account that charging stations might currently be in use by other customers. In such a scenario, some charging stations could also be reserved in advance to keep waiting times low~\cite{Qin11}.
For fast heuristic variants, it might also be useful to precompute potentials for A*~search, as in \gls*{alt}~\cite{Gol05b}.
From a more theoretical perspective, approximability of both problem settings considered in this chapter is an open question. Results by Strehler~et~al.~\cite{Mer15} include an FPTAS for a very similar problem concerning routes with charging stops, which might extend to our setting.

% Acknowledgments here
\paragraph{Acknowledgements.}
The authors would like to thank Raphael Luz for providing consumption data, Martin Uhrig for charging functions, and Sabine Storandt for benchmark data.

% Appendix here
% Options are (1) APPENDIX (with or without general title) or 
%             (2) APPENDICES (if it has more than one unrelated sections)
% Outcomment the appropriate case if necessary
%
% \begin{APPENDIX}{<Title of the Appendix>}
% \end{APPENDIX}
%
%   or 
%
% \begin{APPENDICES}
% \section{<Title of Section A>}
% \section{<Title of Section B>}
% etc
% \end{APPENDICES}

% References here (outcomment the appropriate case) 

% CASE 1: BiBTeX used to constantly update the references 
%   (while the paper is being written).
\bibliographystyle{plain} % outcomment this and next line in Case 1
\bibliography{references-trsc} % if more than one, comma separated

% CASE 2: BiBTeX used to generate mypaper.bbl (to be further fine tuned)
%\input{mypaper.bbl} % outcomment this line in Case 2

\end{document}